\documentclass[12pt,a4paper]{article}
\usepackage{amsmath,amssymb,amsthm,graphicx,color}
\usepackage[dvipdfmx]{pict2e}
\newtheorem{theorem}{Theorem}[section]
\newtheorem{proposition}[theorem]{Proposition}
\newtheorem{lemma}[theorem]{Lemma}
\newtheorem{corollary}[theorem]{Corollary}
\newtheorem{definition}[theorem]{Definition}
\newtheorem{assumption}[theorem]{Assumption}

\def\supp{\mathop{\rm supp}\nolimits}
\def\dist{\mathop{\rm dist}\nolimits}
\def\Ker{\mathop{\rm Ker}\nolimits}
\def\Ran{\mathop{\rm Ran}\nolimits}

\def\min{\mathop{\rm min}\limits}
\def\max{\mathop{\rm max}\limits}

\def\codim{\mathop{\rm codim}\nolimits}
\makeatletter
    
    \@addtoreset{equation}{section}
\makeatother

\begin{document}
\begin{center}
 {\large\bf 
Integrated density of states for the Poisson point interactions
on $\mathbf{R}^3$
}

by

Masahiro Kaminaga\footnote{
Department of Information Technology,
Tohoku Gakuin University,
3-1, Shimizu-Koji, Sendai 984-8588, Japan.\ 
E-mail: kaminaga@mail.tohoku-gakuin.ac.jp
}, %
Takuya Mine\footnote{
Faculty of Arts and Sciences,
Kyoto Institute of Technology,
Matsugasaki, Sakyo-ku,
Kyoto 606-8585, Japan.\ 
E-mail: mine@kit.ac.jp
}, %
and Fumihiko Nakano\footnote{
Mathematical institute,
Tohoku University,
Sendai 980-8578, Japan.\ 
E-mail: fumihiko.nakano.e4@tohoku.ac.jp
}
\end{center}

\begin{abstract}
We determine the principal term of
the asymptotics of the integrated density of states (IDS) 
$N(\lambda)$
for the 
Schr\"odinger operator with point interactions on $\mathbf{R}^3$
as $\lambda \to -\infty$,
provided that 
the set of positions of the point obstacles is the Poisson configuration,
and the interaction parameters are bounded i.i.d.\ random variables.
In particular, we prove $N(\lambda) =O(|\lambda|^{-3/2})$ as $\lambda\to -\infty$.
In the case that all interaction parameters are equal to a constant, 
we give a more detailed asymptotics
of $N(\lambda)$,
and verify the result by a numerical method using 
the R programming language.
As a byproduct, we give a rigorous definition of
the Schr\"odinger operators with point interactions 
on a bounded open set with the Dirichlet or Neumann 
boundary conditions by the method of quadratic form, 
and study fundamental properties
about the counting functions of eigenvalues; 
e.g., Dirichlet--Neumann bracketing, etc.
\end{abstract}

\section{Introduction}

Let $Y$ be a set of points in $\mathbf{R}^d$ ($d=1,2,3$) 
which is locally finite,
that is,
$\#\{y\in Y;|y|<R\}$
is finite for every $R>0$,
where $\#S$ is the cardinality of a set $S$.
Let $\alpha=(\alpha_y)_{y\in Y}$ be a sequence of real numbers.
We define 
a linear operator
$-\Delta_{\alpha,Y}$
on $L^2(\mathbf{R}^d)$ by
\begin{align*}
-\Delta_{\alpha,Y}u
:=&\ -\Delta|_{\mathbf{R}^d\setminus Y} u,\\
D(-\Delta_{\alpha,Y})
:=&\ \{u \in L^2(\mathbf{R}^d)\cap H^2_{\rm loc}(\mathbf{R}^d\setminus Y)
\,;\, -\Delta|_{\mathbf{R}^d\setminus Y} u\in L^2(\mathbf{R}^d),\\
&\qquad
u \mbox{ satisfies }(BC)_y\mbox{ for every }y\in Y
\},
\end{align*}
where 
$\Delta|_{\mathbf{R}^d\setminus Y} u =\sum_{j=1}^d\partial^2 u/\partial x_j^2$ 
is defined as an element of the Schwartz distribution ${\cal D}'(\mathbf{R}^d\setminus Y)$,
and $D(H)$ is the domain of a linear operator $H$
(the notation $A:=B$ means that $A$ is defined by $B$).
The boundary condition $(BC)_y$ at the point $y\in Y$
is defined as follows.
\begin{itemize}
 \item 
For $d=1$,
\begin{align*}
 u(y-0)=u(y+0)=u(y),\quad u'(y+0)-u'(y-0)=\alpha_y u(y).
\end{align*}
\item
For $d=2$,
\begin{align*}
 u(x)=u_{y,0}\log|x-y|+u_{y,1}+o(1)\quad (x\to y),\quad
2\pi\alpha_y u_{y,0}+u_{y,1}=0.
\end{align*}
\item
For $d=3$,
\begin{align*}
u(x)=u_{y,0}|x-y|^{-1}+u_{y,1}+o(1)\quad (x\to y),\quad
-4\pi\alpha_y u_{y,0}+u_{y,1}=0. 
\end{align*}
\end{itemize}
Here, $\displaystyle u(y\pm 0)=\lim_{\epsilon\to \pm 0}u(y\pm\epsilon)$,
and $u_{y,0}$ and $u_{y,1}$ are constants.

It is known that the operator $-\Delta_{\alpha,Y}$ is self-adjoint
on $L^2(\mathbf{R}^d)$
under suitable conditions on $\alpha$ and $Y$,
e.g.,
the uniform discreteness condition
\begin{align}
\label{01_01}
 d_*:=\inf_{y,y'\in Y,\ y\not= y'}|y-y'|>0
\end{align}
(see e.g., \cite{Al-Ge-Ho-Ho}).
Since the support of the interaction is concentrated on the set $Y$,
the quantum system governed by the Hamiltonian $-\Delta_{\alpha,Y}$ is called `point interaction model',
`zero-range interaction model', `delta interaction model',
or `Fermi pseudo potential model', etc.
%
%
%
%
The monograph
by Albeverio et al.\ \cite{Al-Ge-Ho-Ho}
contains
most of the basic facts about the operator $-\Delta_{\alpha,Y}$,
and references up to 2004.


Let us review known results
about the \textit{random point interactions},
that is, the case that $\alpha$ or $Y$ depends 
on some random parameter $\omega$.
There are numerous results in the case $d=1$
(references are found in \cite{Pa-Fi}).
%
When $d=2$ or $3$,
there are several results
in the case $Y=\mathbf{Z}^d$ or its random subset,
and 
$\alpha_\omega=(\alpha_{\omega,y})_{y\in \mathbf{Z}^d}$
are i.i.d.\ (independently, identically distributed) random variables.
Albeverio, H{\o}egh-Krohn, Kirsch, and Martinelli \cite{Al-Ho-Ki-Ma}
studied the spectrum of $-\Delta_{\alpha_\omega,\mathbf{Z}^d}$.
Boutet de Monvel, and Grinshpun  \cite{Bo-Gr}
and 
Hislop, Kirsch, and Krishna \cite{Hi-Ki-Kr1} proved
the Anderson localization near the bottom of the spectrum.
Recently
Hislop, Kirsch, and Krishna \cite{Hi-Ki-Kr2} studied
the eigenvalue statistics.
Dorlas, Macris, and Pul\'e \cite{Do-Ma-Pu}
studied the random point interactions on 
$\mathbf{R}^2$
in a constant magnetic field.

In the present paper, 
we consider the following case.
\begin{assumption}
\label{assumption_PA}
\begin{enumerate}
 \item 
The space $(\Omega,\mathcal{F},\mathbf{P})$ is a probability space.
The random set $Y_\omega$ ($\omega\in \Omega$) is the Poisson configuration
(the support of the Poisson point process) on $\mathbf{R}^d$
with intensity measure $\rho dx$ for some constant $\rho>0$.

 \item The parameters $\alpha_\omega=(\alpha_{\omega,y})_{y\in Y_\omega}$
are real-valued i.i.d.\ random variables
with common distribution $\nu$ on $\mathbf{R}$.
Moreover, $(\alpha_{\omega,y})_{y\in Y_\omega}$ 
are independent of $Y_\omega$.
\end{enumerate}
\end{assumption}
\noindent
We denote $H_\omega = -\Delta_{\alpha_\omega, Y_\omega}$ in the sequel.
For the definition of the Poisson point process,
see Appendix (section \ref{section_pp}).
For the rigorous meaning of 
`$(\alpha_{\omega,y})_{y\in Y_\omega}$ 
are independent of $Y_\omega$',
see the assumption (I7) in section \ref{subsection_eIDS}.
Under Assumption \ref{assumption_PA},
the Schr\"odinger operator with random scalar potential
\begin{align}
\label{01_02}
 -\Delta + \sum_{y\in Y_\omega}\alpha_{\omega,y}V_0(x-y)
\quad (\mbox{$V_0$: scalar function})
\end{align}
is a mathematical model which governs the motion of a quantum particle in an amorphous material,
and is extensively studied in the theory of random Schr\"odinger
operators (see e.g., \cite{Ca-La,Pa-Fi}).
Formally, 
the operator $-\Delta_{\alpha_\omega, Y_\omega}$
is a special case
of (\ref{01_02}), in the sense that $V_0$ is a zero-range potential.
Notice that the condition (\ref{01_01}) is almost surely not satisfied under
Assumption \ref{assumption_PA}.
When $d=1$, 
the operator $H_\omega$ 
under Assumption \ref{assumption_PA}
is  studied by several authors
(\cite{Fr-Ll,Lu-Sy, Ko, Minami}).
Especially, the self-adjointness is proved in \cite{Minami}
(or it can be proved by an application of 
the deterministic result \cite{Ko-Ma}).
When $d=2$ or $3$, 
the operator $H_\omega$ under Assumption \ref{assumption_PA}
 is first studied recently by the authors \cite{Ka-Mi-Na}.
The results of \cite{Ka-Mi-Na} are summarized as follows.
\begin{theorem}[\cite{Ka-Mi-Na}]
\label{theorem_kamina}
Let $d=1,2$ or $3$, and 
suppose that Assumption \ref{assumption_PA} holds.
Then, the following holds.

\begin{enumerate}
 \item[(1)] The operator $H_\omega$ is self-adjoint almost surely.
 \item[(2)]
The spectrum $\sigma(H_\omega)$ is given as follows.
When $d=1$, we have
\begin{align*}
 \sigma(H_\omega)=
\begin{cases}
 [0,\infty) & (\supp \nu \subset [0,\infty)),\\
 \mathbf{R} & (\supp\nu \cap (-\infty,0)\not=\emptyset),
\end{cases}
\end{align*}
almost surely.
When $d=2,3$, we have $ \sigma(H_\omega)=\mathbf{R}$ almost surely.
\end{enumerate}
\end{theorem}

Next we introduce the integrated density of states.
For a bounded open set $U$ in $\mathbf{R}^d$
with $Y\cap \partial U=\emptyset$,
let $-\Delta_{\alpha,Y,U}^D$ (resp.\ $-\Delta_{\alpha,Y,U}^N$)
be the restriction of the operator $-\Delta_{\alpha,Y}$ 
on $U$ with Dirichlet boundary conditions 
(resp.\ Neumann boundary conditions) on the boundary $\partial U$,
defined by the method of quadratic form
(for the rigorous definition, see section \ref{subsection_bdry}).
For $\lambda\in \mathbf{R}$ and $\sharp=D$ or $N$,
let $N_{\alpha,Y,U}^\sharp(\lambda)$ 
be the number of the eigenvalues of $-\Delta_{\alpha,Y,U}^\sharp$
less than or equal to $\lambda$, counted with multiplicity.
For $(\alpha_\omega,Y_\omega)$ satisfying Assumption \ref{assumption_PA},
we define for $\lambda\in \mathbf{R}$
\begin{align}
\label{01_03}
 N_\omega^\sharp(\lambda)
:=&\ \lim_{L\to \infty}\frac{N_{\alpha_\omega,Y_\omega, Q_L}^\sharp(\lambda)}{|Q_L|}
\quad(\sharp=D\mbox{ or }N),
\end{align}
where $Q_L=(0,L)^d$, $L\in \mathbf{N}$ 
($\mathbf{N}$ is the set of positive integers),
and 
$|E|$ is the Lebesgue measure of a Lebesgue measurable set $E$.
Under Assumption \ref{assumption_PA},
we see that $Y_\omega \cap \partial Q_L=\emptyset$ almost surely,
so $-\Delta^\sharp_{\alpha_\omega,Y_\omega,Q_L}$ is well-defined almost surely.
We can also prove that for every $\lambda\in \mathbf{R}$, 
the right hand side 
of (\ref{01_03}) converges almost surely under Assumption \ref{assumption_PA},
and 
\begin{align}
\label{01_04}
 N_\omega^\sharp(\lambda)
=&\
 \lim_{L\to \infty}\frac{\mathbf{E}[N_{\alpha_\omega,Y_\omega,Q_L}^\sharp(\lambda)]}{|Q_L|}
\quad(\sharp=D\mbox{ or }N)
\end{align}
almost surely,
where $\mathbf{E}[X]$ is the expectation of a random variable $X$
(see Theorem \ref{theorem_existence_IDS} below).
The equation (\ref{01_04}) implies that
$N^\sharp_\omega(\lambda)$ is 
deterministic,
that is, 
$N^\sharp_\omega(\lambda)$ is a random variable
that is constant w.r.t.\ $\omega$,
after a possible correction on an event of probability $0$.
From this reason, 
we denote the right hand side of (\ref{01_04}) by $N^\sharp(\lambda)$,
in the sequel.
Furthermore, we can prove that 
$N^N(\lambda+0)= N^D(\lambda+0)$
for every $\lambda\in \mathbf{R}$,
where $f(\lambda+ 0):=\lim_{\epsilon\to +0}f(\lambda+\epsilon)$
(see Theorem \ref{theorem_independence_IDS} below). 
We define $N(\lambda):= N^\sharp(\lambda+0)$ ($\sharp=D$ or $N$)
for $\lambda\in \mathbf{R}$,
then $N(\lambda)$ is independent of the boundary condition
$\sharp=D$ or $\sharp=N$).
Moreover, 
if $N^\sharp(\lambda)$ ($\sharp=D$ or $N$)
is continuous at $\lambda$
(this condition holds for all $\lambda\in \mathbf{R}$
except at most a countable subset of $\mathbf{R}$),
then $N^D(\lambda)=N^N(\lambda)=N(\lambda)$
(Theorem \ref{theorem_independence_IDS}). 
We call the function $N(\lambda)$ the 
\textit{integrated density of states (IDS)}.

For $\lambda<0$,
it is convenient to introduce other counting function
\begin{align*}
 N_{\alpha,Y}(\lambda):=&\ 
\#\{\mu \leq \lambda \,;\,
\mu \mbox{ is an eigenvalue of }- \Delta_{\alpha,Y} \},
\notag
\\
 N_{\alpha,Y,U}(\lambda):=&\ 
 N_{\alpha_U,Y_U}(\lambda),\quad
Y_U=Y\cap U,\ \alpha_U=\alpha|_{Y_U},
\end{align*}
where eigenvalues are counted with multiplicity,
$\alpha|_{Y_U}$ is the restriction of
the sequence $\alpha=(\alpha_y)_{y\in Y}$ on the subset 
$Y_U$.
Notice that the operator $-\Delta_{\alpha_U,Y_U}$
is defined on $L^2(\mathbf{R}^d)$ (not $L^2(U)$).
We can prove that
if 
$(\alpha_\omega,Y_\omega)$ satisfies Assumption \ref{assumption_PA},
$\lambda<0$ and
$N(\lambda) = N^D(\lambda)=N^N(\lambda)$, then
\begin{align}
\label{01_05}
N(\lambda)
=
\lim_{L\to \infty}
\frac{ N_{\alpha_\omega,Y_\omega,Q_L}(\lambda)}{|Q_L|}
=
\lim_{L\to \infty}
\frac{\mathbf{E}[N_{\alpha_\omega,Y_\omega,Q_L}(\lambda)]}{|Q_L|}
\end{align}
almost surely
(see Proposition \ref{proposition_DNbracketing2} below).
The advantage of the new expression (\ref{01_05}) is that
$N_{\alpha_\omega,Y_\omega,Q_L}(\lambda)$ can be calculated explicitly
(Proposition \ref{prop_al}).

The asymptotics of IDS near the bottom of the spectrum
is intensively studied in the theory of random Schr\"odinger operators
(see e.g., \cite{Ca-La, Pa-Fi,Ki-Me, Ve}).
In the present paper,
we give an asymptotic formula of IDS $N(\lambda) $ for our operator 
$H_\omega$ as $\lambda \to -\infty$ in the case $d=3$,
under Assumption \ref{assumption_PA}.

First we consider the case where $\alpha_\omega=(\alpha_{\omega,y})_{y\in Y_\omega}$ is a \textit{constant sequence},
that is, $\alpha_{\omega,y}$ is equal to a deterministic value 
independent of $\omega$ and $y$.
We  denote its common value by $\alpha\in \mathbf{R}$,
and write 
$-\Delta_{\alpha,Y_\omega}:=-\Delta_{\alpha_\omega,Y_\omega}$
for simplicity.
In order to state our main result, we introduce some auxiliary functions.
If $s>\max (0,-4\pi \alpha)$,
it is easy to see that 
the following equation with respect to $R$
\begin{align}
\label{01_06}
 s-\frac{e^{-sR}}{R}=-4\pi \alpha
\end{align}
has a unique positive solution $R$,
which we denote by $R_\alpha(s)$.
We denote the inverse function of $R=R_\alpha(s)$
by $s=s_\alpha(R)$,
then $s_\alpha$ is a real-valued monotone decreasing function defined on
$(0,1/(4\pi \alpha))$ for $\alpha>0$,
or 
on $(0,\infty)$ for $\alpha\leq 0$.
Particularly when $\alpha=0$,
the equation (\ref{01_06}) becomes
a simple equation $sR-e^{-sR}=0$.
We denote the unique positive solution of $t-e^{-t}=0$ by 
$t_0(\doteqdot 0.567...$).
Then we have explicitly
\begin{align*}
R_0(s) =\frac{t_0}{s},\quad
s_0(R)=\frac{t_0}{R}.
\end{align*}
For $\alpha\not= 0$, we still have 
\begin{align}
 \label{01_07}
R_\alpha(s)\sim \frac{t_0}{s} \quad(s\to \infty),\quad
s_\alpha(R)\sim \frac{t_0}{R} \quad(R\to +0),
\end{align}
where $f\sim g$ means $f/g\to 1$
(Proposition \ref{prop_ev2}).
The function $R_\alpha(s)$ can be 
computed with high accuracy
with the aid of some numerical tool.

\begin{theorem}
 \label{theorem_asymptotics_IDS}
Let $d=3$.
Suppose that $\alpha_\omega$ and $Y_\omega$ satisfy
Assumption \ref{assumption_PA}.
Suppose further that
$\alpha_\omega$ is a constant sequence
with its common value $\alpha$.
Let $N(\lambda)$ be the corresponding IDS
for the operator $H_\omega=-\Delta_{\alpha,Y_\omega}$.
Then, for any $\epsilon$ with $0<\epsilon<3$,
we have
\begin{align}
\label{01_08}
N(-s^2)=\frac{2\pi}{3}\rho^2 R_\alpha(s)^3 + O(s^{-6+\epsilon})
\quad(s\to \infty).
\end{align}
In particular, the principal term of the asymptotics is given by
\begin{align}
\label{01_09}
 N(-s^2) \sim \frac{2\pi}{3}\rho^2 t_0^3s^{-3}\quad
(s\to\infty).
\end{align}
\end{theorem}
\noindent
In section \ref{section_num}, we give a numerical verification
of Theorem \ref{theorem_asymptotics_IDS},
using the R programming language.

Next we consider the case that $\supp \nu$ is bounded,
i.e.,
there exist constants $\alpha_-$ and $\alpha_+$ such that
$\alpha_-\leq \alpha_{\omega, y} \leq \alpha_+$ 
for every $y\in Y_\omega$, almost surely.
By the monotonicity of the negative eigenvalues of $-\Delta_{\alpha,Y}$
with respect to $\alpha$ (statement (1) of Proposition \ref{prop_al2}),
we have
\begin{align*}
 N_{\alpha_-,Y_\omega,Q_L}(\lambda)
\geq
 N_{\alpha_\omega,Y_\omega,Q_L}(\lambda)
\geq
 N_{\alpha_+,Y_\omega,Q_L}(\lambda)
\end{align*}
for every $\lambda<0$ and every $L\in \mathbf{N}$
with $ Y_\omega\cap \partial Q_L=\emptyset$, almost surely.
Since the principal term of the asymptotics (\ref{01_09}) is independent of $\alpha$,
we have the following corollary.
\begin{corollary}
 Let $d=3$.
Suppose that $\alpha_\omega$ and $Y_\omega$ satisfy
Assumption \ref{assumption_PA}, 
and $\supp \nu$ is bounded.
Then, we have the asymptotics
\begin{align}
\label{01_10}
 N(-s^2) \sim \frac{2\pi}{3}\rho^2 t_0^3s^{-3}\quad
(s\to\infty).
\end{align}
\end{corollary}

Let us compare the asymptotics (\ref{01_10}) with 
the corresponding result for the scalar type Poisson random potential
\begin{align}
\label{01_11}
 -\Delta + \sum_{y\in Y_\omega}V_0(x-y),
\end{align}
where 
$V_0$ is a real-valued function and 
$Y_\omega$ is the Poisson configuration.
If $V_0\leq 0$, $V_0$ has global minimum at $x=0$,
and $V_0$ satisfies some regularity and decaying conditions, 
it is shown in the monograph of Pastur and Figotin
\cite[Theorem 9.4]{Pa-Fi} that
\begin{align}
\label{01_12}
 \log N(\lambda) = -\frac{|\lambda|}{|V_0(0)|}\log|\lambda|\cdot (1+o(1))
\quad(\lambda \to -\infty).
\end{align}
The asymptotics (\ref{01_12}) implies that $N(\lambda)$ decays super-exponentially,
much faster than the polynomial decay (\ref{01_10}).

Heuristically, the asymptotics (\ref{01_12}) is explained as follows.
If the operator (\ref{01_11})
has a bound state localized near a point $a$ 
with negative energy $\lambda$,
there must be a potential well
of depth $|\lambda|$ near the point $a$.
Thus it is necessary that at least
\begin{align}
\label{01_13}
 n=\left[\frac{|\lambda|}{|V_0(0)|}\right]
\end{align}
points of $Y_\omega$ are in the ball of small radius $\epsilon$
centered at $a$,
where $[x]$ is the greatest integer less than or equal to $x$.
The probability of this event is
\begin{align}
\label{01_14}
 \frac{(\rho |B_\epsilon(0)|)^n}{n!}e^{-\rho|B_\epsilon(0)|},
\end{align}
where $B_\epsilon(0)$ is a ball of radius $\epsilon$
centered at the origin.
Then, roughly speaking,
the asymptotics (\ref{01_12}) is deduced from
(\ref{01_13}), (\ref{01_14}), and
the Stirling formula 
\begin{align*}
 n!  \sim (2\pi n)^{1/2}\left({n}/{e}\right)^n\quad
(n \to \infty).
\end{align*}

In the case of point interactions,
the situation is completely different.
Consider the case where $Y=\{y_1,y_2\}$ is a two-point set,
and $\alpha$ is a constant sequence.
Then, 
by an explicit calculation,
we can prove that the operator $-\Delta_{\alpha,Y}$ has a negative eigenvalue
$\lambda = -(s_\alpha(R))^2$,
where $R:=|y_1-y_2|$ (see section 2),
and $s_\alpha(R)\to \infty$ as $R\to 0$
by (\ref{01_07}).
The leading asymptotics of $N(\lambda)$ is given by the expectation number
of pairs within the distance $R_{\alpha}(s)$ of
Poisson points per unit volume
(Lemma \ref{lemma_ids-est}).
The expectation of the number of such pairs
can be explicitly computed
(Proposition \ref{prop_cluster}),
and we obtain the above result.

The present paper is organized as follows.
In section 2,
we review the results about 
the spectrum of $-\Delta_{\alpha,Y}$ when $\# Y=1,2$.
In section 3,
we calculate the expectation of the number of $n$-point clusters
in the Poisson configuration.
In section 4, we prove Theorem \ref{theorem_asymptotics_IDS}.
In section 5, we give a numerical verification of Theorem \ref{theorem_asymptotics_IDS}, using the R programming language.
In section 6,
we prove fundamental facts about the definition of IDS,
namely, the existence of IDS and the independence of IDS
with respect to the boundary conditions.
As a byproduct, we give a rigorous foundation
about the operator $-\Delta^\sharp_{\alpha,Y,U}$
($\sharp=D,N$), by using the method of quadratic form.
In Appendix (section 7),
we review some known results about
the Poisson point process,
spectrum of the Hamiltonian with point interactions,
and the operator norm of a matrix operator.
The results in sections 6 and 7 are proved independently from
those in sections 2, 3, 4, and 5,
so we use the results in sections 6 and 7 in the proof of our main results.

\section{Spectrum of the Hamiltonian with one or two point interactions}
Let $d=3$,
$Y$ be a finite set in $\mathbf{R}^3$, and 
$\alpha=(\alpha_y)_{y\in Y}$ be a sequence of real numbers.
Put $N:=\#Y$.
It is known that
\begin{align*}
&\  \sigma(-\Delta_{\alpha,Y})
=
 \sigma_p(-\Delta_{\alpha,Y})\cup [0,\infty),\\
&\ 
 \sigma_p(-\Delta_{\alpha,Y})\subset (-\infty,0],
\end{align*}
where $\sigma_p(-\Delta_{\alpha,Y})$ is the point spectrum of 
$-\Delta_{\alpha,Y}$
(see \cite[Theorem II.1.1.4, and Errata in page 485]{Al-Ge-Ho-Ho},
see also Proposition \ref{prop_al} and Proposition \ref{proposition_zeroev} in Appendix).
More explicitly,
$\sigma_p(-\Delta_{\alpha,Y})\cap (-\infty, 0)$ for $N=1,2$ is given as follows.

When $N=1$,
Proposition \ref{prop_al} implies that 
$-s^2$ ($s>0$) is an eigenvalue of $-\Delta_{\alpha,Y}$
if and only if
\begin{align*}
 \alpha + \frac{s}{4\pi}=0.
\end{align*}
Then we see that
\begin{align*}
 \sigma_p(-\Delta_{\alpha,Y})\cap (-\infty,0)
=
\begin{cases}
 \emptyset & (\alpha \geq 0),\\
\{-(4\pi \alpha)^2\} & (\alpha <0).
\end{cases}
\end{align*}

When $N=2$, $Y=\{y_1,y_2\}$,
$\alpha_{y_1}=\alpha_{y_2}=\alpha$,
and $R:=|y_1-y_2|>0$,
Proposition \ref{prop_al} implies that 
$-s^2$ ($s>0$) is an eigenvalue of $-\Delta_{\alpha,Y}$
if and only if
\begin{align*}
&
\left|
\begin{matrix}
\displaystyle
 \alpha + \frac{s}{4\pi} & 
\displaystyle
-\frac{e^{-sR}}{4\pi R}\\[3mm]
\displaystyle
-\frac{e^{-sR}}{4\pi R} & 
\displaystyle
\alpha + \frac{s}{4\pi}
\end{matrix}
\right|
=0
\Leftrightarrow
\left(\alpha+\frac{s}{4\pi}\right)^2-\left(\frac{e^{-sR}}{4\pi R}\right)^2=0.
\end{align*}
This equation is satisfied if and only if
\begin{align}
s+\frac{e^{-sR}}{R}=-4\pi\alpha
\label{02_01}
\end{align}
or
\begin{align}
s-\frac{e^{-sR}}{R}=-4\pi\alpha.
\label{02_02}
\end{align}
The solution $s$ to the equation (\ref{02_01})
does not exist if $\alpha \geq 0$,
and is bounded by $-4\pi \alpha$ if $\alpha<0$.
The equation (\ref{02_02}) is the same as
(\ref{01_06}), and 
the solution is equal to $s_\alpha(R)$ by definition.
An asymptotics of $s_\alpha(R)$ 
and that of its inverse $R_\alpha(s)$ are given as follows.
\begin{proposition}
\label{prop_ev2}
Let 
$t_0$ be the unique positive solution of
$t-e^{-t}=0$. 
\begin{enumerate}
  \item[(1)] Suppose $\alpha=0$.
Then we have
$s_0(R)=t_0/R$
and
$R_0(s)=t_0/s$.

  \item[(2)] Suppose $\alpha\not=0$.
Then, for any $\epsilon>0$,
there exists a positive constant $R_0$ dependent on $\alpha$ and
$\epsilon$
such that
\begin{align}
\label{02_03}
\frac{t_0-\epsilon}{R}<s_\alpha(R)<\frac{{t_0+\epsilon}}{R}
\end{align}
for any $R$ with $0<R<R_0$,
and 
\begin{align}
\label{02_04}
 \frac{t_0-\epsilon}{s}<
R_\alpha(s)<\frac{t_0+\epsilon}{s}
\end{align}
for any $s$ with $s>(t_0+\epsilon)/R_0$.
 \end{enumerate}
\end{proposition}
\begin{proof}
 The proof of (1) is already given in the introduction.
We prove (2).
Let $\phi(t):=t-e^{-t}$.
Then the defining equation (\ref{01_06}) of $s_\alpha(R)$
is equivalent to
\begin{align}
 \phi(sR)+4\pi\alpha R=0.
\label{02_05}
\end{align}
Since
\begin{align*}
&\  \frac{\partial \phi}{\partial t}=1+e^{-t}>0,\\
&\ \lim_{s\to +0}\phi(sR)=-1,\quad
\lim_{s\to \pm\infty}\phi(sR)=\pm\infty,
\end{align*}
we see that 
the left hand side of (\ref{02_05}) 
is monotone increasing with respect to $s$,
and (\ref{02_05}) has a unique positive solution $s_\alpha(R)$ 
if $-1+4\pi\alpha R<0$.
Since $\phi(t_0)=0$ by definition,
we have $\phi(t_0-\epsilon)<0<\phi(t_0+\epsilon)$ for $\epsilon>0$,
and 
\begin{align*}
&\  \phi\left(t_0-\epsilon\right)+4\pi\alpha R<0,\\
&\  \phi\left(t_0+\epsilon\right)+4\pi\alpha R>0
\end{align*}
for small $R>0$.
So the left hand side of (\ref{02_05}) 
is negative for $s=(t_0-\epsilon)/R$,
is positive for $s=(t_0+\epsilon)/R$,
if $R$ is sufficiently small.
Thus (\ref{02_03}) follows from the intermediate value theorem.
Then we immediately have (\ref{02_04}),
since $R=R_\alpha(s)$ is the inverse function of $s=s_\alpha(R)$.
\end{proof}

\section{Number of clusters in Poisson configuration}
Let us calculate the expectation of the number of $n$-point clusters in the 
Poisson configuration in $\mathbf{R}^d$
($d=1,2,3,\ldots$).
For $y\in \mathbf{R}^d$ and $R>0$,
let $B_y(R):=\{z\in \mathbf{R}^d ; |z-y|<R\}$.
For the Poisson configuration $Y_\omega$ in $\mathbf{R}^d$,
we put
\begin{align}
\label{03_01}
n_y(R) :=\#(Y_\omega \cap B_y(R)).
\end{align}
\begin{proposition}
\label{prop_cluster}
Let $Y_\omega$ be the Poisson configuration 
with intensity $\rho dx$, where $\rho$ is a positive constant.
Let $L\in\mathbf{N}$ and put $Q_L:=(0,L)^d$.
Then, we have for $n=1,2,3,\ldots$ and $R>0$
\begin{align}
\label{03_02}
 \frac{\mathbf{E}[
\#\{y\in Y_\omega \cap Q_L
\,;\, n_y(R)=n
\}
]}
{|Q_L|}
=
\frac{1}{(n-1)!}|B_0(R)|^{n-1}\rho^n e^{-\rho|B_0(R)|}.
\end{align}
\end{proposition}
\begin{proof}
Let $Q_L'=(-R,L+R)^d$.
If $y\in Q_L$, then $B_y(R)\subset Q_L'$.
Moreover, $Y_\omega \cap Q_L'$
and $Y_\omega\cap (Q_L')^c$ are independently distributed.
So we can replace $Y_\omega$ in the definition of
$n_y(R)$ by $Y_\omega':=Y_\omega \cap Q_L'$,
when we calculate the left hand side of (\ref{03_02}).
Then,
\begin{align*}
&\  \mathbf{E}[
\#\{y\in Y_\omega\cap Q_L \,;\, n_y(R)=n\}
]\notag\\
=&\ 
\sum_{k=n}^\infty \mathbf{E}[
\#\{y\in Y_\omega \cap Q_L \,;\, n_y(R)=n\}
\mid \#Y_\omega'=k
]
\cdot 
\mathbf{P}(\#Y_\omega'=k),
\end{align*}
where $\mathbf{E}[A|B]$ is the conditional expectation of $A$
under the condition $B$.
Under the condition $\#Y_\omega'=k$,
the points $Y_\omega'=\{y_1,\ldots,y_k\}$ obey
independent uniform distributions on $Q_L'$
(Proposition \ref{prop_poisson_construction}).
Thus we have an expression about
the conditional probability
\begin{align*}
&\  \mathbf{P}(y_j \in Q_L,\ n_{y_j}(R)=n \mid \# Y_\omega'=k)\\
=&\ 
\frac{|Q_L|}{|Q_L'|}
\cdot
\frac{(k-1)!}{(n-1)!(k-n)!}
\cdot
\left(\frac{|B_{y_j}(R)|}{|Q_L'|}\right)^{n-1}
\cdot
\left(\frac{|Q_L'|-|B_{y_j}(R)|}{|Q_L'|}\right)^{k-n}
\end{align*}
for $j=1,\ldots, k$.
Then
\begin{align*}
&\  \mathbf{E}[\#\{y\in Y_\omega \cap Q_L\,;\, n_y(R)=n\}]\\
=&\ 
\sum_{k=n}^\infty
\sum_{j=1}^k \mathbf{P}(y_j \in Q_L,\ n_{y_j}(R)=n \mid \# Y_\omega'=k)
\cdot 
\mathbf{P}(\# Y_\omega'=k)\\
=&\ 
\sum_{k=n}^\infty
k\cdot
\frac{|Q_L|}{|Q_L'|}
\cdot
\frac{(k-1)!}{(n-1)!(k-n)!}
\cdot
\left(\frac{|B_{0}(R)|}{|Q_L'|}\right)^{n-1}
\cdot
\left(\frac{|Q_L'|-|B_{0}(R)|}{|Q_L'|}\right)^{k-n}\\
&\qquad
\cdot
\frac{(\rho |Q_L'|)^k}{k!}e^{-\rho|Q_L'|}\\
=&\ 
\frac{|Q_L||B_0(R)|^{n-1}}{(n-1)!}e^{-\rho|Q_L'|}
\sum_{k=n}^\infty \frac{\rho^k(|Q_L'|-|B_0(R)|)^{k-n}}{(k-n)!}\\
=&\ 
\frac{|Q_L||B_0(R)|^{n-1}}{(n-1)!}e^{-\rho|Q_L'|}\rho^n
e^{\rho(|Q_L'|-|B_0(R)|)}\\
=&\ 
\frac{|Q_L||B_0(R)|^{n-1}}{(n-1)!}\rho^n
e^{-\rho|B_0(R)|}.
\end{align*}
\end{proof}

\noindent
{\bf Remark.}
This result can also be deduced by Palm distribution;
see section \ref{subsection_palm}.

\section{Proof of Theorem \ref{theorem_asymptotics_IDS}}
In this section, 
we always suppose that the assumptions in 
Theorem \ref{theorem_asymptotics_IDS} hold,
and prove Theorem \ref{theorem_asymptotics_IDS}.
The relation between asymptotics of IDS
and the number of 2-point clusters
is given in the following lemma.
\begin{lemma}
\label{lemma_ids-est}
 For every $\delta$ with $1/2<\delta <1$,
and for every $m>0$,
there exist constants $R_0=R_0(\rho,\alpha,\delta,m)>0$
and $C=C(\rho,\alpha,\delta,m)>0$
such that
\begin{align}
& \mathbf{E}[N_{\alpha_\omega,Y_\omega,Q_L}(-(s_\alpha(R)-R^m)^2)]\notag\\
\geq&\  \frac{1}{2}
\mathbf{E}[
\#\{y\in Y_\omega\cap Q_L\,;\, n_y(R)=2\}]
-
C R^{6\delta}|Q_L|,
\label{04_01}
\\
& \mathbf{E}[N_{\alpha_\omega,Y_\omega,Q_L}(-(s_\alpha(R)+R^m)^2)]\notag\\
\leq&\  \frac{1}{2}
\mathbf{E}[
\#\{y\in Y_\omega\cap Q_L\,;\, n_y(R)=2\}]
+
C R^{6\delta}|Q_L|,
\label{04_02}
\end{align}
for every $0<R<R_0$ and every $L>R^{-2}$,
where $n_y(R)$ is defined in (\ref{03_01}).
\end{lemma}

First we show that Lemma \ref{lemma_ids-est} implies
Theorem \ref{theorem_asymptotics_IDS}.

\begin{proof}[Proof of Theorem \ref{theorem_asymptotics_IDS}]
Assume Lemma \ref{lemma_ids-est} holds.
Put $f(R)=s_\alpha(R)-R^m$.
By Proposition \ref{prop_ev2},
we see that $s_\alpha(R)\sim t_0/R$ ($R\to 0$) and
$R_\alpha(s)\sim t_0/s$ ($s\to \infty$).
Thus $f(R)\sim s_\alpha(R)\sim t_0/R$ ($R\to 0$).

Since 
the function $R=R_\alpha(s)$ is the inverse of 
the function $s=s_\alpha(R)$, we have
$R=R_\alpha(f(R)+R^m)$.
By the Taylor theorem, we have
\begin{align}
R=
R_\alpha(f(R))+\frac{d R_\alpha}{d s}(f(R)+\theta R^m)R^m,\quad 0<\theta<1.
\label{04_03}
\end{align}
In order to calculate the derivative $dR_\alpha/ds$,
recall that $R_\alpha$ is defined as the solution of the 
following equation with respect to $R$
\begin{align*}
 h(s,R)=-4\pi\alpha,\quad
\mbox{where }
h(s,R):=s- \frac{e^{-sR}}{R}.
\end{align*}
Then
\begin{align*}
 \frac{\partial h}{\partial s}=
1+e^{-sR}>0,\quad
 \frac{\partial h}{\partial R}=
\frac{(sR+1)e^{-sR}}{R^2}>0,
\end{align*}
and we have by the implicit function theorem
\begin{align}
&\frac{dR_\alpha}{ds}(s)
=-
\left.
\frac{\frac{\partial h}{\partial s}}{\frac{\partial h}{\partial R}}
\right|_{R=R_\alpha(s)}
=
-\frac{ 1+e^{-sR_\alpha(s)}}{(sR_\alpha(s)+1)e^{-sR_\alpha(s)}}R_\alpha(s)^2,
\label{04_04}\\
&\frac{ds_\alpha}{dR}(R)
=-
\left.
\frac{\frac{\partial h}{\partial R}}{\frac{\partial h}{\partial s}}
\right|_{s=s_\alpha(R)}
=
-
\frac{(s_\alpha(R)R+1)e^{-s_\alpha(R)R}}{ 1+e^{-s_\alpha(R)R}}R^{-2}.
\label{04_05}
\end{align}
Then there exist positive constants 
$C_1$, $C_2$, $C_3$, $C_4$ such that
\begin{align*}
&\frac{C_1}{R}\leq f(R)+\theta R^m \leq \frac{C_2}{R},\\
&C_3 R\leq R_\alpha(f(R) + \theta R^m)\leq C_4 R
\end{align*}
for every $0\leq \theta \leq 1$ and small $R>0$.
Thus (\ref{04_03}) and (\ref{04_04}) imply
\begin{align}
\label{04_06}
 R = R_\alpha(f(R)) + O(R^{m+2})
=R_\alpha(f(R))(1+O(R^{m+1}))
\quad (R\to 0).
\end{align}
Take $m\geq 2$.
Then, by (\ref{04_01}), (\ref{04_06}) and Proposition \ref{prop_cluster}
\begin{align*}
 \frac{\mathbf{E}[N_{\alpha_\omega,Y_\omega,Q_L}(-f(R)^2)]}{|Q_L|}
\geq&\ 
\frac{2\pi}{3}R^3\rho^2(1+O(R^3))
- C R^{6\delta}\\
\geq &\ 
\frac{2\pi\rho^2}{3}R_\alpha(f(R))^3 - C' f(R)^{-6\delta}
\end{align*}
for some $C'>0$,
 and sufficiently small $R$ and large $L$.
Moreover,
we have by (\ref{04_05})
\begin{align*}
\frac{df}{dR}(R)
=&\ 
-
\frac{(s_\alpha(R)R+1)e^{-s_\alpha(R)R}}{ 1+e^{-s_\alpha(R)R}}R^{-2}
 - mR^{m-1}\\
\sim&\ 
-\frac{(t_0+1)e^{-t_0}}{ 1+e^{-t_0}}R^{-2}
=-t_0 R^{-2}
<0
\end{align*}
(notice that $t_0=e^{-t_0}$) for sufficiently small $R>0$.
Thus 
there exists some $R_0>0$ such that
the map $s=f(R)$ is 
bijective from $(0,R_0)$ to $(f(R_0),\infty)$,
and the lower bound in (\ref{01_08}) is proved.

Similarly, in order to prove the upper bound in (\ref{01_08}),
we put $g(R)=s_\alpha(R)+R^m$.
Then we have from (\ref{04_02}) and Proposition \ref{prop_cluster}
\begin{align*}
\frac{\mathbf{E}[N_{\alpha_\omega,Y_\omega,Q_L}(-g(R)^2)]}{|Q_L|}
\leq &\ 
\frac{2\pi\rho^2}{3}R_\alpha(g(R))^3 + C'' g(R)^{-6\delta}
\end{align*}
for some $C''>0$,
 and sufficiently small $R$ and large $L$.
Making the change of variable $s=g(R)$,
and taking the limit $L\to \infty$,
we have the conclusion of Theorem \ref{theorem_asymptotics_IDS}.
\end{proof}

 \begin{proof}[Proof of Lemma \ref{lemma_ids-est}]
  Take $\delta$ with $1/2<\delta <1$.
Consider the decomposition
$Y_\omega \cap Q_L = \widetilde{Y_\omega}\cup Z_\omega$, where
\begin{align*}
 Z_\omega = \{ y \in Y_\omega \cap Q_L\mid n_y(R^\delta)\geq 3\},\quad
\widetilde{Y_\omega} = (Y_\omega \cap Q_L)\setminus Z_\omega.
\end{align*}
In the sequel, we often omit the subscript $\omega$ and
write $Y=Y_\omega$, etc.
For the counting function,
we also omit $Q_L$ and write
$N_{\alpha_\omega,Y_\omega, Q_L}(\lambda)
=
N_{\alpha,Y}(\lambda)
$, etc.
By the properties of counting functions
(Proposition \ref{prop_al2}),
we have for $\lambda <0$
\begin{align}
&\ 
 N_{\alpha,\widetilde{Y}}(\lambda)
\leq
 N_{\alpha,Y}(\lambda)
\leq
 N_{\alpha,\widetilde{Y}}(\lambda)
+\# Z,\notag\\
&\ 
\mathbf{E}[
N_{\alpha,\widetilde{Y}}(\lambda)
] 
\leq
\mathbf{E}[
 N_{\alpha,Y}(\lambda)
] 
\leq
\mathbf{E}[
 N_{\alpha,\widetilde{Y}}(\lambda)
] 
+
\mathbf{E}[
\# Z
].
\label{04_07}
\end{align}
By Proposition \ref{prop_cluster}, 
\begin{align}
 \mathbf{E}[\# Z]
=&\ 
\sum_{n=3}^\infty
\frac{1}{(n-1)!}|B_0(R^\delta)|^{n-1}\rho^n e^{-\rho|B_0(R^\delta)|}|Q_L|\notag\\
=&\ 
O(R^{6\delta})|Q_L|\quad (R\to 0).
\label{04_08}
\end{align}

Next, we put for $y\in \widetilde{Y}$
\begin{align*}
\widetilde{n_y}(R^\delta)
:=
\#(\widetilde{Y}\cap B_y(R^\delta)), 
\end{align*}
then $1\leq \widetilde{n_y}(R^\delta)\leq n_y(R^\delta)\leq 2$
since $\widetilde{Y}\subset Y$.
Thus we have another decomposition
$\widetilde{Y}=\widetilde{Y}^{(1)}\cup \widetilde{Y}^{(2)}$, where
\begin{align*}
&\ 
 \widetilde{Y}^{(j)}
=
\{y\in \widetilde{Y}\mid \widetilde{n_y}(R^\delta)=j\}\quad (j=1,2).
\end{align*}
Then, for any $y \in \widetilde{Y}^{(2)}$, there exists a 
$y'\in \widetilde{Y}^{(2)}$ such that 
$0<|y-y'|<R^\delta$.
This $y'$ is unique, since otherwise we have $\widetilde{n_y}(R^\delta)\geq 3$.
Thus we can number the elements of $\widetilde{Y}^{(2)}$ and 
$\widetilde{Y}^{(1)}$ as
\begin{align*}
  \widetilde{Y}^{(2)}
=&\
\{y_1,y_2\}
\cup
\{y_3,y_4\}
\cup
\cdots
\cup
\{y_{2n-1},y_{2n}\},\\
&\ 
0<|y_1-y_2|\leq |y_3-y_4|\leq\cdots \leq |y_{2n-1}-y_{2n}|<R^\delta,\\
\widetilde{Y}^{(1)}=&\ \{y_{2n+1},\ldots, y_{2n+p}\},
\end{align*}
where $2n=\# \widetilde{Y}^{(2)}$ and $p=\# \widetilde{Y}^{(1)}$.
By definition, we have
\begin{align*}
 |y_k-y_{k'}|\geq R^\delta,
\end{align*}
if $\{k,k'\}\not=\{2j-1,2j\}$  for any $j=1,\ldots,n$.
  
Let $A$ and $B$ be
$(2n+p)\times (2n+p)$ matrices given by
\begin{align*}
&\  A=\alpha I_{2n+p} + \frac{1}{4\pi}B,\\
&\ 
B=(b_{jk}),\quad
b_{jk}
=
\begin{cases}
 s & (j=k),\\
-\frac{e^{-s|y_j-y_k|}}{|y_j-y_k|}
&
(j\not=k),
\end{cases}
\end{align*}
where $I_m$ is the identity matrix of degree $m$.
For $j=1,\ldots,2n+p$, let $\mu_j=\mu_j(s)$ be the $j$-th smallest eigenvalue of $A$,
counted with multiplicity.
By Proposition \ref{prop_al},
$\lambda =-s^2$ ($s>0$) is an eigenvalue of 
$-\Delta_{\alpha,\widetilde{Y}}$
if and only if
$\mu_j(s)=0$ for some $j$.

We introduce another matrix $B'$ given by
\begin{align*}
 B'=(b'_{jk}),\quad
b'_{jk}=
\begin{cases}
 b_{jk} & (|y_j-y_k|<R^\delta),\\
 0 & (|y_j-y_k|\geq R^\delta).
\end{cases}
\end{align*}
Then the matrix $B'$ is written as the direct sum
\begin{align*}
&\ 
 B'=\bigoplus_{j=1}^{n}B_j \oplus s I_p,\quad
B_j
=
\begin{pmatrix}
 s & -\frac{e^{-s|y_{2j-1}-y_{2j}|}}{|y_{2j-1}-y_{2j}|}\\
-\frac{e^{-s|y_{2j-1}-y_{2j}|}}{|y_{2j-1}-y_{2j}|} & s
\end{pmatrix}.
\end{align*}
Put $A'=\alpha I_{2n+p}+\frac{1}{4\pi}B'$.
Then, the $j$-th smallest eigenvalue
$\mu'_j=\mu'_j(s)$ of $A'$ is given explicitly as
\begin{align}
 &\ \mu'_j
= 
\alpha+\frac{1}{4\pi}
\left(s-\frac{e^{-s|y_{2j-1}-y_{2j}|}}{|y_{2j-1}-y_{2j}|}\right)
\quad
(1\leq j \leq n),
\label{04_09}
\\
&\ \mu'_{n+j}=
\alpha+\frac{s}{4\pi}
\quad
(1\leq j \leq p),
\notag
\\
&\ \mu'_{2n+p+1-j}= 
\alpha+\frac{1}{4\pi}
\left(s+\frac{e^{-s|y_{2j-1}-y_{2j}|}}{|y_{2j-1}-y_{2j}|}\right)
\quad
(1\leq j \leq n).
\notag
\end{align}

By the min-max principle,
the difference of $\mu_j$ and $\mu_j'$ is estimated as
\begin{align}
\label{04_10}
 |\mu_j-\mu_j'|\leq \|A-A'\|=\frac{1}{4\pi}\|B-B'\|,
\end{align}
where $\|\cdot\|$ is the operator norm on $\ell^2(\widetilde{Y})
\simeq \mathbf{C}^{2n+p}$.
The components of the matrix $B-B'$ are given as
\begin{align*}
 b_{jk}-b_{jk}'
=
\begin{cases}
 0 & (|y_j-y_k|<R^\delta),\\
-\frac{e^{-s|y_j-y_k|}}{|y_j-y_k|}
&(|y_j-y_k|\geq R^\delta).
\end{cases}
\end{align*}
By Proposition \ref{prop_operator-norm},
we have
\begin{align}
\label{04_11}
 \|B-B'\|
\leq \max_j \sum_{k;|y_j-y_k|\geq R^\delta}
\frac{e^{-s|y_j-y_k|}}{|y_j-y_k|}.
\end{align}
In order to estimate the right hand side of (\ref{04_11}),
we fix $y_j$,
divide $\mathbf{R}^3$ into cubes of side 
$\widetilde{R_\delta}:=R^\delta/\sqrt{3}$,
and classify the cubes into infinite layers,
like Figure \ref{figure_boxes}.
\begin{figure}[htbp]
\begin{center}
\begin{picture}(240,240)(0,0)
\setlength\unitlength{1pt}
\multiput(0,0)(0,40){7}{\line(1,0){240}}
\multiput(0,0)(40,0){7}{\line(0,1){240}}
\put(120,120){\circle*{5}}
\put(125,125){$y_j$}
\put(100,100){Layer 0}
\put(100,60){Layer 1}
\put(100,20){Layer 2}
\put(15,125){$\widetilde{R_\delta}$}
\linethickness{3pt}
\put(0,0){\polygon(0,0)(0,240)(240,240)(240,0)}
\put(40,40){\polygon(0,0)(0,160)(160,160)(160,0)}
\put(80,80){\polygon(0,0)(0,80)(80,80)(80,0)}
\end{picture}
\end{center}
\caption{Division by cubes of side $\widetilde{R_\delta}$}
\label{figure_boxes}
\end{figure}
The diameter of each cube is $R^\delta$.
By the construction of $\widetilde{Y}$,
we see that
\begin{enumerate}
 \item[(i)]
there is no point $y_k$ of $\widetilde{Y}$ 
with $|y_j-y_k|\geq R^\delta$ in Layer 0,

 \item[(ii)] there are at most two points of 
$\widetilde{Y}$ in each cube in Layer $\ell$ ($\ell\geq 1$),

 \item[(iii)] the minimal distance from $y_j$ to Layer $\ell$ is 
$\ell \widetilde{R_\delta}$,

 \item[(iv)] the number of cubes in Layer $\ell$ is
\begin{align*}
(2\ell+2)^3-(2\ell)^3=8(3\ell^2+3\ell+1)\leq 56\ell^2
\quad (\ell\geq 1).
\end{align*}
\end{enumerate}
The function $e^{-sR}/R$ is monotone decreasing with respect to $R$.
Thus we have by (\ref{04_11}) and (\ref{04_10})
\begin{align}
 \|B-B'\|
\leq&\ 
\sum_{\ell=1}^\infty
2\cdot 56\ell^2\cdot \frac{e^{-s\ell \widetilde{R_\delta}}}{\ell \widetilde{R_\delta}}
=
\frac{112}{\widetilde{R_\delta}}\sum_{\ell=1}^\infty
\ell e^{-s \ell \widetilde{R_\delta}}
=
\frac{112}{\widetilde{R_\delta}}\frac{e^{-s\widetilde{R_\delta}}}{(1-e^{-s\widetilde{R_\delta}})^2},\notag\\
\|A-A'\|
\leq &\ 
\frac{28}{\pi \widetilde{R_\delta}}
\frac{e^{-s\widetilde{R_\delta}}}{(1-e^{-s\widetilde{R_\delta}})^2},
\label{04_12}
\end{align}
where we used the formula
\begin{align*}
 \sum_{\ell=1}^\infty \ell e^{-\ell x}=\frac{e^{-x}}{(1-e^{-x})^2}
\quad (x>0).
\end{align*}
By (\ref{04_10}) and (\ref{04_12}), we have
\begin{align}
\label{04_13}
\mu_j'-g(s) \leq \mu_j \leq \mu_j'+g(s),
\quad
g(s):=\frac{28}{\pi \widetilde{R_\delta}}
\frac{e^{-s\widetilde{R_\delta}}}{(1-e^{-s\widetilde{R_\delta}})^2}.
\end{align}

Let 
\begin{align*}
&\  q'=q_\omega':=
\frac{1}{2}\#\{y\in \widetilde{Y_\omega}\,;\, \widetilde{n_y}(R)=2\},\\
&\ 
\widetilde{n_y}(R)
=
\# (\widetilde{Y_\omega}\cap B_y(R)).
\end{align*}
If $0<R<1$, then $R<R^\delta$, and we have 
\begin{align}
&\ 
 |y_1-y_2|
\leq 
|y_3-y_4|
\leq
\cdots
\leq 
|y_{2q'-1}-y_{2q'}|
<
R,
\label{04_14}
\end{align}
and
\begin{align}
|y_{k}-y_{k'}|\geq R
\label{04_15}
\end{align}
if $\{k,k'\}\not=\{2j-1,2j\}$ for any $j=1,\ldots,q'$.

We shall show
\begin{align}
\label{04_16}
 \mu_j(s_{\alpha}(R)-R^m)<0\quad (j=1,\ldots,q')
\end{align}
for $m>0$ and sufficiently small $R$.
To see this, we put $h(s,r)=s-e^{-sr}/r$.
By (\ref{04_09}) and (\ref{04_13}),
\begin{align}
&\  \mu_j(s_{\alpha}(R)-R^m)\notag\\
\leq &\ 
\alpha + \frac{1}{4\pi}h\bigl(s_{\alpha}(R)-R^m, |y_{2j-1}-y_{2j}|\bigr)
+g(s_{\alpha}(R)-R^m).
\label{04_17}
\end{align}
Since
\begin{align}
&\  \frac{\partial h}{\partial r}=\frac{(sr+1)e^{-sr}}{r^2}>0,
\label{04_18}
\\
&\ 1< \frac{\partial h}{\partial s}=1+e^{-sr}<2,\quad
-r< \frac{\partial^2 h}{\partial s^2}=-r e^{-sr}<0,
\label{04_19}
\end{align}
we have by (\ref{04_14}),
(\ref{04_18}), (\ref{04_19}),
the Taylor theorem,
and the definition of $s_\alpha(R)$,
\begin{align}
&\ 
\alpha+ \frac{1}{4\pi}h\bigl(s_{\alpha}(R)-R^m, |y_{2j-1}-y_{2j}|\bigr)\notag\\
<&\ 
\alpha+ \frac{1}{4\pi}h\bigl(s_{\alpha}(R)-R^m, R\bigr)\notag\\
=&\ 
\alpha
+
\frac{1}{4\pi}
\left\{h(s_\alpha(R),R)
-\frac{\partial h}{\partial s}
(s_\alpha(R),R)R^m
+
\frac{1}{2}\frac{\partial^2 h}{\partial s^2}(s_\alpha(R)-\theta R^m,R)R^{2m}
\right\}\notag\\
&\ 
\hspace{10cm} (0<\theta<1)\notag\\
<&\ 
\alpha+\frac{1}{4\pi}(-4\pi \alpha -R^m)=-\frac{R^m}{4\pi}.
\label{04_20}
\end{align}
As for the last term of (\ref{04_17}),
we use the fact
\begin{align*}
 s_\alpha(R)-R^m
\geq \frac{C}{R}
\end{align*}
for some $C>0$ and sufficiently small $R>0$
(Proposition \ref{prop_ev2}).
Then we have
\begin{align*}
 ( s_\alpha(R)-R^m)\widetilde{R_\delta} 
\geq \frac{C}{R}\cdot \frac{R^\delta}{\sqrt{3}}=C' R^{-(1-\delta)},
\quad
C'=\frac{C}{\sqrt{3}}.
\end{align*}
Since $g(s)$ is monotone decreasing with respect to $s$,
we have
\begin{align}
 g(s_\alpha(R)-R^m)
\leq 
\frac{28}{\pi \widetilde{R_\delta}}
\frac{e^{-C'R^{-(1-\delta)}}}{(1-e^{-C'R^{-(1-\delta)}})^2}
\leq C'' R^{-\delta}e^{-C'R^{-(1-\delta)}},
\label{04_21}
\end{align}
for some $C''>0$ and sufficiently small $R$.
By (\ref{04_17}), (\ref{04_20}), (\ref{04_21}), and $1-\delta>0$,
we obtain (\ref{04_16}) for sufficiently small $R$.

Next we shall show
\begin{align}
\label{04_22}
 \mu_j(s_{\alpha}(R)+R^m)>0\quad (j\geq q'+1)
\end{align}
for $m>0$ and sufficiently small $R$.
It is sufficient show (\ref{04_22}) only when $j=q'+1$.
We only consider the case $q'+1 \leq n$
(the proof in the case $q'+1>n$ is easier).
By (\ref{04_09}) and (\ref{04_13}),
\begin{align}
&\  \mu_{q'+1}(s_{\alpha}(R)+R^m)\notag\\
\geq &\ 
\alpha + \frac{1}{4\pi}h\bigl(s_{\alpha}(R)+R^m, |y_{2q'+1}-y_{2q'+2}|\bigr)
-g(s_{\alpha}(R)+R^m).
\label{04_23}
\end{align}
We have by (\ref{04_15}),
(\ref{04_18}), (\ref{04_19}),
the Taylor theorem,
and the definition of $s_\alpha(R)$
\begin{align}
&\ 
\alpha+ \frac{1}{4\pi}h\bigl(s_{\alpha}(R)+R^m, |y_{2q'+1}-y_{2q'+2}|\bigr)\notag\\
\geq &\ 
\alpha+ \frac{1}{4\pi}h\bigl(s_{\alpha}(R)+R^m, R\bigr)\notag\\
=&\ 
\alpha
+
\frac{1}{4\pi}
\left\{h(s_\alpha(R),R)
+\frac{\partial h}{\partial s}
(s_\alpha(R),R)R^m
+
\frac{1}{2}\frac{\partial^2 h}{\partial s^2}(s_\alpha(R)+\theta R^m,R)R^{2m}
\right\}\notag\\
&\ \hspace{10cm}
(0<\theta<1)\notag\\
>&\ 
\alpha+\frac{1}{4\pi}(-4\pi \alpha +R^m-\frac{1}{2}R^{2m+1})
=\frac{1}{4\pi}\left(R^m - \frac{1}{2}R^{2m+1}\right)>0
\label{04_24}
\end{align}
for sufficiently small $R>0$.
Similarly to (\ref{04_21}),
we can prove 
\begin{align}
  g(s_\alpha(R)+R^m)
\leq C'' R^{-\delta}e^{-C'R^{-(1-\delta)}}
\label{04_25}
\end{align}
for some constants $C'>0$, $C''>0$ and 
sufficiently small $R>0$.
By (\ref{04_23}), (\ref{04_24}), (\ref{04_25}), and $1-\delta>0$,
we obtain (\ref{04_22}) for sufficiently small $R$.

We recall the fact that 
$\mu_j(s)$ is strictly monotone increasing with respect to $s$
and $\displaystyle\lim_{s\to \infty}\mu_j(s)=\infty$
(Proposition \ref{prop_al}).
Then, (\ref{04_16}) implies that
the equation $\mu_j(s)=0$ has a unique solution with $s>s_\alpha(R)-R^m$
for $j=1,\ldots,q'$.
Thus we have
\begin{align}
\label{04_26}
 N_{\alpha,\widetilde{Y}}(-(s_\alpha(R)-R^m)^2)\geq q'_\omega
\end{align}
for sufficiently small $R>0$.
Moreover, (\ref{04_22}) implies that
the equation $\mu_j(s)=0$ has no solution with $s>s_\alpha(R)+R^m$
for $j\geq q'+1$.
Thus we have
\begin{align}
\label{04_27}
 N_{\alpha,\widetilde{Y}}(-(s_\alpha(R)+R^m)^2)\leq q'_\omega
\end{align}
for sufficiently small $R>0$.
Taking the expectation of (\ref{04_26}) and (\ref{04_27}),
and using (\ref{04_07}) and (\ref{04_08}),
we have
\begin{align}
&\   
\mathbf{E}[N_{\alpha,Y}(-(s_\alpha(R)-R^m))]
\geq 
\mathbf{E}[q'_\omega],
\label{04_28}
\\
&\ 
\mathbf{E}[N_{\alpha,Y}(-(s_\alpha(R)+R^m))]
\leq 
\mathbf{E}[q'_\omega]+O(R^{6\delta})|Q_L|.
\label{04_29}
\end{align}

Lastly, we estimate the difference between
\begin{align*}
 q_\omega := 
\frac{1}{2}
\#\{y \in Y_\omega \cap Q_L \,;\, n_y(R)=2\}
\end{align*}
and 
\begin{align*}
 q_\omega' = 
\frac{1}{2}
\#\{y \in \widetilde{Y_\omega} \,;\, \widetilde{n_y}(R)=2\},
\end{align*}
provided that $0<R<1$ and $L>2$ (so $0<R<R^\delta$).

First we show
\begin{align}
 \label{04_30}
 q_\omega\geq q_\omega'.
\end{align}
Assume that $y\in \widetilde{Y_\omega}$ and $\widetilde{n_y}(R)=2$.
The first assumption  implies $n_y(R^\delta)\leq 2$, and
\begin{align*}
 2=\widetilde{n_y}(R)\leq n_y(R)\leq n_y(R^\delta) \leq 2.
\end{align*}
Thus we have 
$\{y\in \widetilde{Y_\omega};\widetilde{n_y}(R)=2\}\subset\{y\in Y_\omega\cap Q_L;n_y(R)=2\}$,
which proves (\ref{04_30}).

Next,
let $y\in \{y\in Y_\omega\cap Q_L;n_y(R)=2\}\setminus \{y\in \widetilde{Y_\omega};\widetilde{n_y}(R)=2\}$.
Then, we have
$n_y(R)=2$, and either following (i) or (ii) holds:
\begin{enumerate}
 \item[(i)] $y\in Z_\omega$, that is, $n_y(R^\delta)\geq 3$.

 \item[(ii)] $y\in \widetilde{Y_\omega}$, but $\widetilde{n_y}(R)=1$.
\end{enumerate}
The number of points satisfying (i) is bounded by $\#Z_\omega$.
In the case (ii), there exists a unique point $y'\in Y_\omega$
such that $|y-y'|<R$ and $y'\not\in \widetilde{Y_\omega}$.
Then, either following (iia) or (iib) holds:
\begin{enumerate}
 \item[(iia)] $y'\in Z_\omega$.
 \item[(iib)] $y'\not\in Q_L$.
\end{enumerate}
In the case (iia), we have $n_{y'}(R^\delta)\geq 3$,
and there exists $y''\in Y_\omega$ such that $y''\not= y$ and $0<|y''-y'|<R^\delta$.
So we have $|y-y''|<R+R^\delta < 2R^\delta$,
and $n_y(2R^\delta) \geq 3$.
In the case (iib), we have $y\in Q_L\setminus Q_L''$,
where $Q_L'':=(R,L-R)^3$.
Thus we have
\begin{align}
2(q_\omega-q_\omega')
\leq 
&
\# Z_\omega
 +
\#\{y\in Y_\omega\cap Q_L \,;\, n_{y}(2R^\delta)\geq 3\}\notag\\
&\ +
\#\{y\in Y_\omega\cap (Q_L\setminus Q_L'') \,;\, n_{y}(R)= 2\}.
 \label{04_31}
\end{align}
By (\ref{04_08}), we have $\mathbf{E}[\# Z_\omega]=O(R^{6\delta})|Q_L|$,
and we can prove that the expectation of the second term is also 
$O(R^{6\delta})|Q_L|$ similarly.
Moreover, we have by Proposition \ref{prop_cluster}
\begin{align}
&\  \mathbf{E}[\#\{y\in Y_\omega\cap (Q_L\setminus Q_L'') \,;\, n_{y}(R)= 2\}]
\notag\\
=&\ |B_0(R)|\rho^2 e^{-\rho|B_0(R)|}(|Q_L|-|Q_L''|)\notag\\
=&\ 
\frac{4\pi}{3}R^3 \rho^2 e^{-\frac{4\pi \rho}{3}R^3}
\{L^3-(L-2R)^3\}\notag\\
\leq &\ 
C L^2R^4
\label{04_32}
\end{align}
for $0<R<1$ and $L>2$,
where $C$ is a positive constant dependent on $\rho$.
By (\ref{04_08}), (\ref{04_30}), (\ref{04_31}), and (\ref{04_32}),
we have
\begin{align*}
&\ 
\mathbf{E}[ q_\omega']
\leq  
\mathbf{E}[q_\omega] 
\leq 
\mathbf{E}[q_\omega']+
\left(
O(R^{6\delta})
+O(L^{-1}R^4)
\right)|Q_L|,
\end{align*}
for $0<R<1$ and $L>2$.
If $L>R^{-2}$, then $O(L^{-1}R^4)=O(R^6)$,
and we have
\begin{align*}
&\ 
\mathbf{E}[ q_\omega']
\leq  
\mathbf{E}[q_\omega] 
\leq 
\mathbf{E}[q_\omega']+
O(R^{6\delta})|Q_L|
\end{align*}
since $1/2<\delta <1$.
This inequality together with (\ref{04_28}) and (\ref{04_29})
implies the conclusion.
 \end{proof}

\section{Numerical verification of Theorem \ref{theorem_asymptotics_IDS}}
\label{section_num}
Let us verify Theorem \ref{theorem_asymptotics_IDS}
by a numerical method.
Our calculation is based on Proposition \ref{prop_al},
which states that $-\Delta_{\alpha,Y}$ has an eigenvalue $-s^2$ $(s>0)$
if and only if $\det A(s)=0$, where the matrix $A(s)$ is given in 
(\ref{07_02}).
We use a simple algorithm as follows.
\begin{enumerate}
 \item[(i)] 
Choose a real parameter $\alpha$,
the intensity $\rho$ of the Poisson configuration $Y_\omega$,
and the size $L$ of the cube $Q_L$.

 \item[(ii)] Generate a random number $n_\omega$
obeying the Poisson distribution with parameter 
$\rho|Q_L|$.

 \item[(iii)] Generate $n_\omega$ random points 
$y_1,\ldots,y_{n_\omega}$ in $Q_L$,
each of which obeys the uniform distribution in $Q_L$.

 \item[(iv)] Choose $a>0$ and $M>0$,
and divide the interval $I=[0,a]$ into 
$M$ sub-intervals $I_j:=[s_{j-1},s_j]$ ($j=1,2,\ldots,M$), where $s_j=ja/M$.
We put $I_{M+1}:=(a,\infty)$.

 \item[(v)] For $j=0,\ldots,M$, calculate the $n_\omega\times n_\omega$
matrix $A(s_j)$ given in (\ref{07_02})
with 
$\alpha_j=\alpha$ and 
$Y=\{y_1,\ldots,y_{n_\omega}\}$, 
and calculate the sign of $\det A(s_j)$.

 \item[(vi)] 
Calculate the number $n_j$ of the solutions of 
$\det A(s)=0$ in the interval $I_j$, as follows.
\begin{align*}
 n_j :=&\ 
\begin{cases}
 1 & (\det A(s_{j-1})\det A(s_j)<0),\\
 0 & (\det A(s_{j-1})\det A(s_j)\geq 0),
\end{cases}
\quad(j=1,\ldots,M),\\
n_{M+1} :=&\ 
\begin{cases}
 1 & (\det A(s_M)<0),\\
 0 & (\det A(s_M)\geq 0).
\end{cases}
\end{align*}
Notice that $\displaystyle \lim_{s\to \infty} \det A(s)=\infty$.

 \item[(vii)] Repeat (i)-(vi) many times,
and calculate the average $\overline{n_j}$ of $n_j$'s.
Then, the numerical IDS $N(-s_j^2)$ is calculated by
\begin{align*}
N(-s_j^2) = \frac{1}{|Q_L|}\sum_{k={j+1}}^{M+1} \overline{n_k}.
\end{align*}
\end{enumerate}
A defect of the above algorithm is that
we might miss the solutions of $\det A(s)=0$
if there are two or more solutions in $I_j$.
We have to choose $M$ and $a$ sufficiently large,
so that the effect of the counting error is small.
This time we choose $\rho=1$, $L=5$, $a=20$, $M=2000$, and 
repeat 10000 tests for each value $\alpha=-0.5$, $0.0$, $0.5$.
The code is implemented by using R.
The results are given in Figure
\ref{figure-0.5}, \ref{figure0.0}, \ref{figure0.5}.

\begin{figure}[htbp]
\begin{center}
 \begin{tabular}{c}
\begin{minipage}{9cm}
  \includegraphics[width=9cm]{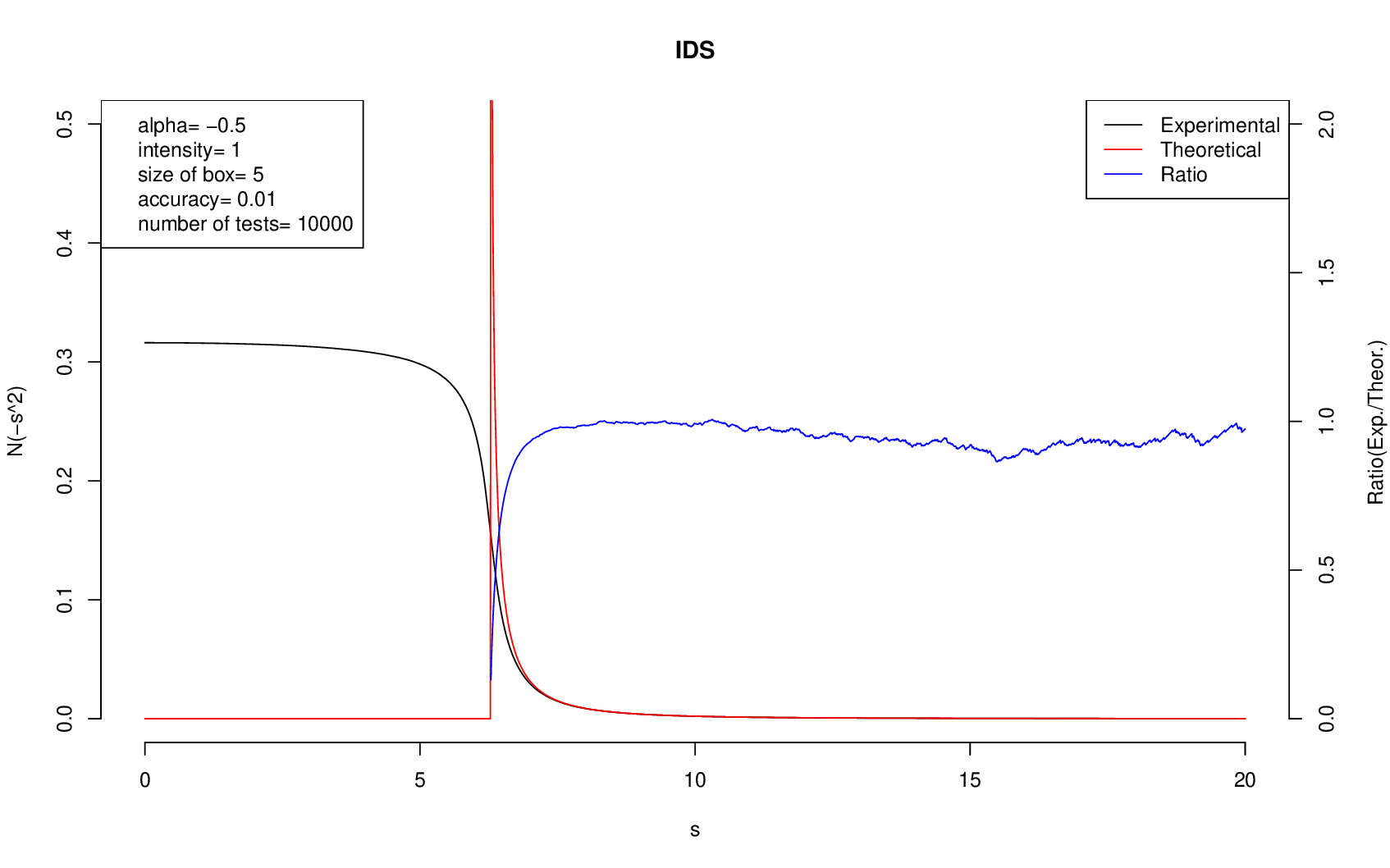}
\caption{$\alpha=-0.5$}
\label{figure-0.5}
\end{minipage}
 \\
\begin{minipage}{9cm}
  \includegraphics[width=9cm]{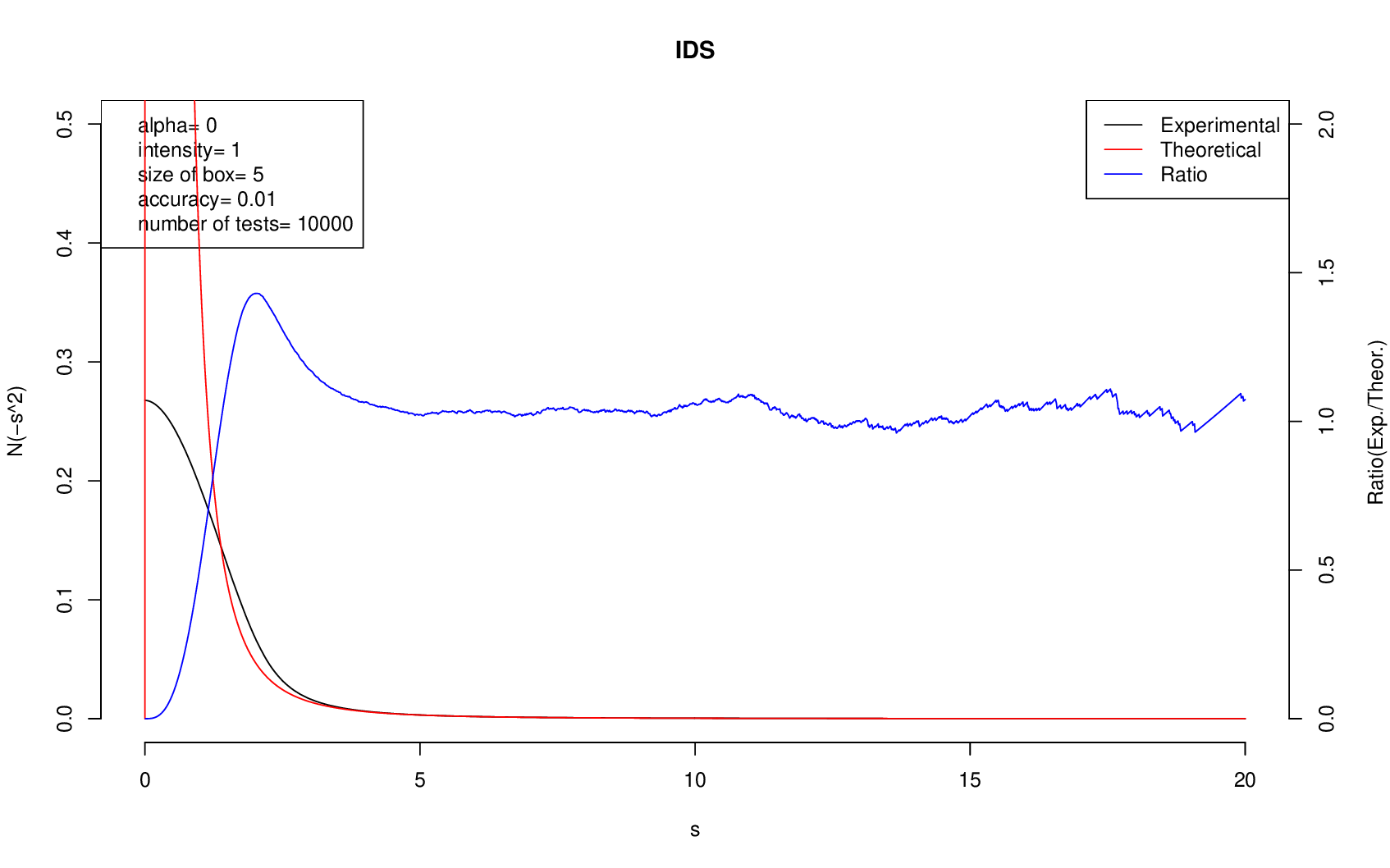}
\caption{$\alpha=0.0$}
\label{figure0.0}
\end{minipage}
\vspace{3mm}
 \\
\begin{minipage}{9cm}
  \includegraphics[width=9cm]{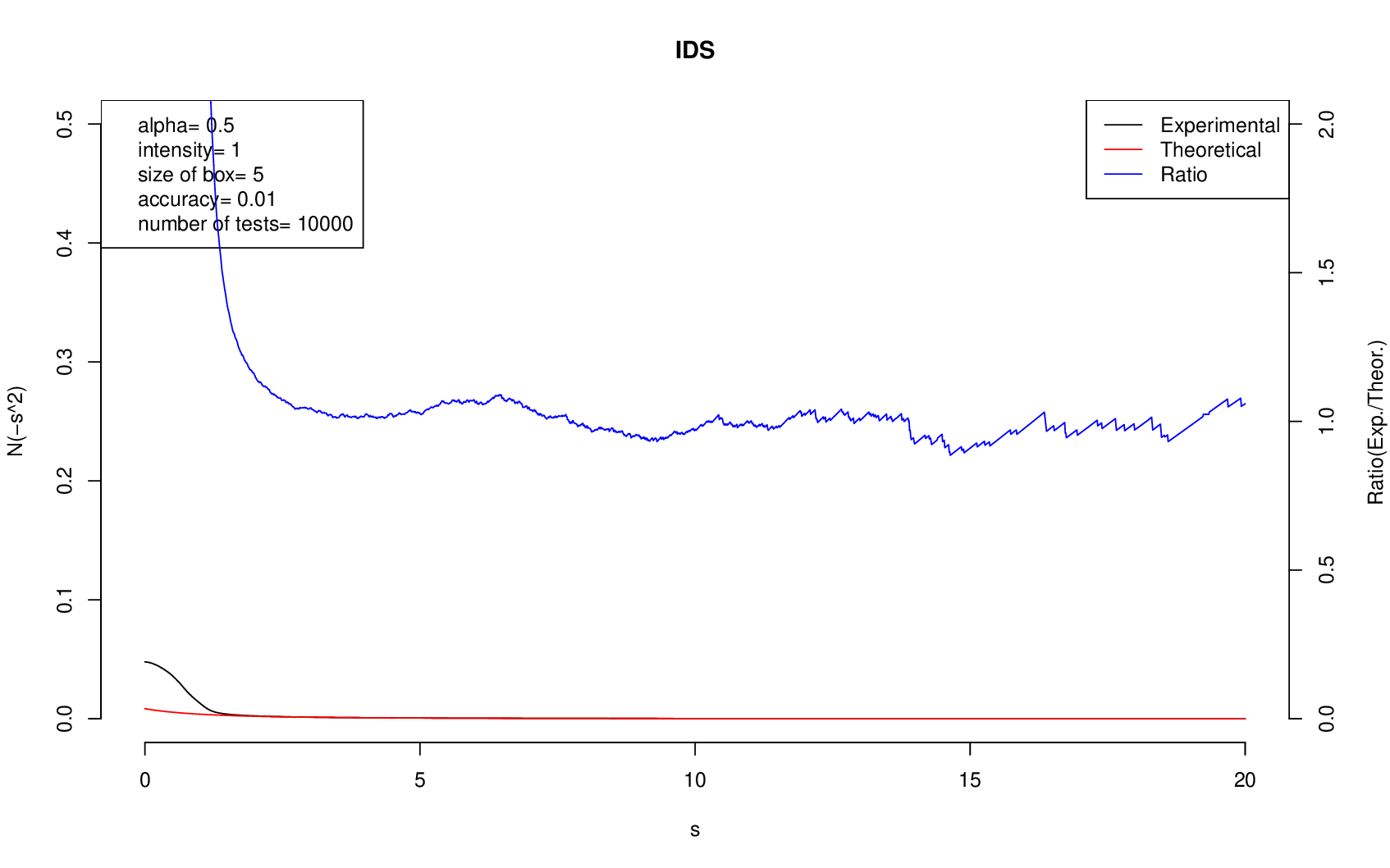}
\caption{$\alpha=0.5$}
\label{figure0.5}
\end{minipage}
 \end{tabular}
\end{center}
\end{figure}
In these figures, the horizontal axis is $s$-axis.
The curve diverging to $\infty$ as $s\to -4\pi\alpha$
is the first term in (\ref{01_08})
(it is defined only for $s>-4\pi\alpha$).
The another smooth curve 
is the numerical IDS $N(-s^2)$ calculated by the above algorithm.
The ticks for these two curves are given in the left side of each figure.
The rapidly fluctuating curve is 
the ratio of the numerical IDS
to the  the first term in (\ref{01_08}),
and the ticks for this curve is given in the right side.
In each figure, we can observe that the ratio is close to $1$ for large $s$,
which justifies our Theorem \ref{theorem_asymptotics_IDS}.

In Figure \ref{figure-0.5},
we find a jump of IDS near $s=-4\pi\alpha\doteqdot 6.28$.
The jump is considered to be created by the isolated points in $Y_\omega$,
since $-(4\pi \alpha)^2$ is the eigenvalue of 
$-\Delta_{\alpha,Y}$ when $Y$ is a one-point set
(section 2).
The height of the jump seems to be an interesting quantity,
although it is not well analyzed yet.

\section{Basics of IDS}
In this section,
we shall prove some fundamental results 
about the definition of IDS (the integrated density of states)
given by (\ref{01_03}),
namely,
the existence of IDS,
and its independence from the boundary conditions,
in the sense that $N^D(\lambda+0)=N^N(\lambda+0)$.
The results are well-known in the case of
the Schr\"odinger operators with random \textit{scalar} potential 
$H_\omega=-\Delta + V_\omega$ (cf.\ \cite{Ca-La, Pa-Fi}),
however, we check the results here again since
we cannot directly apply the known proofs to 
our operators $-\Delta_{\alpha_\omega, Y_\omega}$,
under Assumption \ref{assumption_PA}.
The results in this section are applicable in all dimensions $d=1,2,3$.
Though the results are well-known in one-dimensional case
\cite{Fr-Ll,Ko, Lu-Sy},
we give results here for readers' use.

We use the following notations and facts in the sequel.
Let $\mathcal{H}$ be a separable Hilbert space.
We denote the inner product on $\mathcal{H}$ by $(\cdot,\cdot)_{\mathcal{H}}$,
and the norm by $\|\cdot\|_\mathcal{H}$.
The inner product is assumed to be sesqui-linear,
that is,
conjugate linear with respect to the first component,
and linear with respect to the second component.
Let $A$ be a self-adjoint operator on $\mathcal{H}$.
We denote the spectral projection operator for $A$ on a interval $I$ 
by $P_I(A)$.
We say $A$ is \textit{lower semi-bounded} if there exists
$C\in \mathbf{R}$ such that
\begin{align*}
 (u,Au)_{\mathcal{H}}\geq C\|u\|_{\mathcal{H}}^2
\end{align*}
for every $u\in D(A)$.
Define a norm $\|\cdot\|_{Q(q_A)}$ on $D(A)$ by
\begin{align*}
 \|u\|_{Q(q_A)}^2:=(u,Au)_{\mathcal{H}}+(-C+1)\|u\|_{\mathcal{H}}^2
\end{align*}
for $u\in D(A)$.
We denote the completion of $D(A)$ with respect to the norm 
$\|\cdot\|_{Q(q_A)}$ by $Q(q_A)$,
which is a subspace of $\mathcal{H}$ since the norm $\|\cdot\|_{Q(q_A)}$
is stronger than $\|\cdot\|_{\mathcal{H}}$.
Then the norm $\|\cdot\|_{Q(q_A)}$ and the sesqui-linear form
\begin{align*}
 q_A(u,v):=(u,Av)_{\mathcal{H}}
\end{align*}
defined for $u,v\in D(A)$ is uniquely extended to $Q(q_A)$
by continuity.
Then $q_A$ is a quadratic form, that is,
$q_A$ is densely defined sesqui-linear form on $Q(q_A)$.
We denote $q_A[u]:=q_A(u,u)$.
The quadratic form $q_A$ is symmetric ($q_A(v,u)=\overline{q_A(u,v)}$),
closed ($(Q(q_A),\|\cdot\|_{Q(q_A)})$ is complete),
and lower semi-bounded ($q_A[u]\geq C\|u\|_\mathcal{H}^2$
for some $C\in \mathbf{R}$ and every $u\in Q(q_A)$).
We call $q_A$ the quadratic form associated with the operator $A$.
It is known that there exists one-to-one correspondence 
between a lower semi-bounded self-adjoint operator $A$
and a symmetric, lower semi-bounded, closed quadratic form $q_A$
(see \cite[section VIII.6]{Re-Si1}), as follows.
If $q_A$ is given, the operator domain $D(A)$ is given by
\begin{align*}
& D(A)=\{u\in Q(q_A);\mbox{there exists $C>0$ such that}\\
&\hspace{3cm} |q_A(v,u)|\leq C\|v\|_{\mathcal{H}}
\mbox{ for every }v\in Q(q_A)\}.
\end{align*}
Then, for $u\in D(A)$, the Riesz lemma implies that there exists
unique $Au\in \mathcal{H}$ such that $q_A(v,u)=(v,Au)$
for every $v\in Q(q_A)$,
and this is the definition of $Au$.

Let $U\subset \mathbf{R}^d$ be an open set.
Let $L^2(U)$ be the space of square-integrable functions,
and the inner product and the norm on $L^2(U)$
are given by 
\begin{align*}
 (u,v)_{L^2(U)}=\int_U \overline{u}v\,dx,\quad
\|u\|_{L^2(U)}=(u,u)^{1/2}
\end{align*}
for $u,v\in L^2(U)$.
When $U=\mathbf{R}^d$,
we sometimes use the abbreviation 
$\|u\|=\|u\|_{L^2(\mathbf{R}^d)}$, etc.
For vector-valued functions 
$u=(u_1,...,u_d),v=(v_1,...,v_d)\in L^2(U)^d$,
we denote
\begin{align*}
 (u,v)_{L^2(U)^d}=\int_U \overline{u}\cdot v\,dx,\quad
\|u\|_{L^2(U)^d}=(u,u)_{L^2(U)^d}^{1/2},
\end{align*}
where $\overline{u}\cdot v=\sum_{j=1}^d \overline{u}_jv_j$.
The space $H^j(U)$ ($j=1,2$) is the Sobolev space 
of $j$-th order, defined by
\begin{align*}
& H^1(U)=
\{u\in L^2(U);\nabla u \in L^2(U)^d\},\\
&\|u\|_{H^1(U)}^2
=
\|u\|_{L^2(U)}^2
+
\|\nabla u\|_{L^2(U)^d}^2,\\
&H^2(U)=\{u\in H^1(U);\frac{\partial^2 u}{\partial x_j\partial x_k}\in L^2(U)\ (j,k=1,...,d)\},\\
&\|u\|_{H^2(U)}^2=\|u\|_{H^1(U)}^2+\sum_{j,k=1}^d
\left\|\frac{\partial^2 u}{\partial x_j\partial x_k}\right\|_{L^2(U)}^2,
\end{align*}
where 
$\nabla u = (\frac{\partial u}{\partial x_1},...,\frac{\partial u}{\partial x_d})$,
and the derivatives are defined
as an element of $\mathcal{D}'(U)$
(the Schwartz distributions on $U$).
The space $C_0^\infty(U)$ consists of $u\in C^\infty(U)$ such that
$\supp u$ is compact and $\supp u\subset U$.
We also denote 
the completion of $C_0^\infty(U)$ with respect to
$H^j$-norm by $H^j_0(U)$ ($j=1,2$).
The space $L^2_{\rm loc}(U)$ consists of all 
measurable functions $u$
such that $\chi u\in L^2(U)$ for every $\chi\in C_0^\infty(U)$.
The spaces $H^j_{\rm loc}(U)$ ($j=1,2$) are defined similarly.

\subsection{Min-max principle}
In this subsection,
we recall the min-max principle and its several consequences.
In the sequel, we use the following notation.
Let $A$ be a lower semi-bounded self-adjoint operator 
on a separable Hilbert space $\mathcal{H}$.
If $N:=\dim \Ran P_{(-\infty,\inf \sigma_{\rm ess}(A))}(A)$,
then we denote
\begin{align*}
 \lambda_j(A)
:=
\begin{cases}
 \mbox{$j$-th smallest eigenvalue of $A$}&(1\leq j\leq N),\\
\inf \sigma_{\rm ess}(A)&(j\geq N+1),
\end{cases}
\end{align*}
for $j\in \mathbf{N}$,
where the eigenvalues are counted with multiplicity.
We define the eigenvalue counting function $N(\lambda;A)$ of $A$ by
\begin{align*}
 N(\lambda;A):=\dim\Ran P_{(-\infty,\lambda]}(A)
\end{align*}
for $\lambda\in \mathbf{R}$.
According to the \textit{min-max principle} \cite[Theorem XIII.1,2]{Re-Si4},
we have
\begin{align}
 \label{06_01}
\lambda_j(A)
=\sup_{v_1,...,v_{j-1}\in \mathcal{H}}
\left(
\inf_{u\in Q(q_A)\cap 
\langle v_1,...,v_{j-1}\rangle^{\bot}
,\, \|u\|_\mathcal{H}=1}q_A[u]
\right)
\end{align}
for every $j\in \mathbf{N}$, where 
$\langle v_1,...,v_{j-1}\rangle$
is the linear hull of 
$v_1,...,v_{j-1}$
and $W^\bot$ is the orthogonal complement 
of a subspace $W$.
Here we prepare two variants of the formula (\ref{06_01}).

\begin{proposition}
 \label{proposition_min-max}
Let $\mathcal{H}$ be a separable Hilbert space,
and $A$ be a lower semi-bounded self-adjoint operator on $\mathcal{H}$.
Then, we have
\begin{align}
 \label{06_02}
\lambda_j(A)
=\sup_{\codim V\leq j-1}
\left(
\inf_{u\in V\cap Q(q_A),\, \|u\|_\mathcal{H}=1}q_A[u]
\right)
\end{align}
for every $j\in \mathbf{N}$,
where the supremum ranges over all the linear subspaces $V$ of $\mathcal{H}$
with $\codim V:=\dim(\mathcal{H}/V)\leq j-1$.
Moreover, if $S$ is a countable dense subset of $\mathcal{H}$,
then we have
\begin{align}
 \label{06_03}
\lambda_j(A)
=\sup_{v_1,...,v_{j-1}\in S}
\left(
\inf_{u\in Q(q_A)\cap 
\langle v_1,...,v_{j-1}\rangle^{\bot}
,\, \|u\|_\mathcal{H}=1}q_A[u]
\right)
\end{align}
for every $j\in \mathbf{N}$.
\end{proposition}
\noindent\textbf{Remark.}
Compared with 
(\ref{06_01}),
$V$ is not necessarily a closed subspace in (\ref{06_02}).
This generalization is useful in the proof of the next proposition.
When the quadratic form 
$q_A[u]$ depends on some additional parameter
$\omega$ in a measure space $\Omega$,
the formula
(\ref{06_03}) is useful when we prove the measurability
of $\lambda_j(A)$ with respect to $\omega$.
\begin{proof}
 We denote the right hand side of 
(\ref{06_02}) and (\ref{06_03})
by $\mu_{1,j}(A)$ and $\mu_{2,j}(A)$,
respectively. 
Clearly we have
\begin{align*}
\mu_{2,j}(A)\leq \lambda_j(A)\leq \mu_{1,j}(A).
\end{align*}
We shall prove that the converse inequalities hold.

First we assume that 
$\mu_{1,j}(A)>\lambda_j(A)$,
and take $a\in \mathbf{R}$ such that $\mu_{1,j}(A)>a>\lambda_j(A)$.
Then there exists a subspace $V$ of $\mathcal{H}$ 
with $\codim V\leq j-1$ 
such that
\begin{align*}
 q_A[u]> a
\end{align*}
for every $u\in V\cap Q(q_A)$ with $\|u\|_{\mathcal{H}}=1$.
Since $a> \lambda_j(A)$, 
we have $\dim\Ran P_{(-\infty,a)}(A)\geq j$,
and there exist independent vectors 
$\varphi_k\in Q(q_A)$ ($k=1,...,j$) such that $\|\varphi_k\|_{\mathcal{H}}=1$
and 
\begin{align*}
 q_A[\varphi_k]<a\quad (k=1,\ldots,j).
\end{align*}
Thus we see that $\varphi_k\not\in V$,
which contradicts with the fact that $\codim V\leq j-1$.
Therefore we conclude that $\mu_{1,j}(A)=\lambda_j(A)$.

Next, we will show that $\mu_{2,j}(A)\geq \lambda_j(A)$.
Take a subspace $\Phi_j$ of
$\Ran P_{[\lambda_1(A), \lambda_{j-1}(A)]}(A)$
such that 
\begin{align*}
\dim \Phi_j= j-1,\quad 
\Phi_j^\bot \subset \Ran P_{[\lambda_j(A),\infty)}(A).
\end{align*}
Let $P_j$ be the orthogonal projection onto $\Phi_j$.
Then we have for $u\in Q(q_A)$
\begin{align}
 q_A[u]=&\ 
q_A[P_j u]+q_A[(1-P_j) u]\notag\\
\geq&\ 
\lambda_1(A)\|P_j u\|_{\mathcal{H}}^2
+\lambda_j(A)\|(1-P_j) u\|_{\mathcal{H}}^2\notag\\
=&\ 
\lambda_j(A)\|u\|_{\mathcal{H}}^2
-(\lambda_j(A)-\lambda_1(A))\|P_j u\|_{\mathcal{H}}^2.
\label{06_04}
\end{align}
For $v_1,\ldots,v_{j-1}\in \mathcal{H}$,
put $V:=\langle v_1,\ldots,v_{j-1}\rangle^\bot$,
and let $P_V$ be the orthogonal projection onto $V$.
Then we have by (\ref{06_04})
\begin{align}
&\ 
 \inf_{u\in Q(q_A)\cap V,\|u\|_{\mathcal{H}}=1}q_A[u]\notag\\
\geq 
&\ 
\lambda_j(A)-(\lambda_j(A)-\lambda_1(A))
 \sup_{u\in Q(q_A)\cap V,\|u\|_{\mathcal{H}}=1}
\|P_j u\|_{\mathcal{H}}^2\notag\\
\geq&\ 
\lambda_j(A)-(\lambda_j(A)-\lambda_1(A))\|P_jP_V\|,
\label{06_05}
\end{align}
where $\|P_jP_V\|$
is the operator norm of $P_jP_V$ on $\mathcal{H}$.

Let us take an orthonormal basis 
$\varphi_1,\ldots,\varphi_{j-1}$ of $\Phi_j$,
and take $v_{k,\ell}\in S$ ($k=1,\ldots,j-1$, $\ell\in \mathbf{N}$)
such that
$v_{k,\ell}\to \varphi_k$ as $\ell\to\infty$.
Put $V_\ell=\langle v_{1,\ell},\ldots,v_{j-1,\ell}\rangle^\bot$.
Then we can prove
$P_{V_\ell}\to 1-P_j$ as ($\ell\to\infty$)
in the operator norm,
so $\|P_j P_{V_\ell}\|\to 0$.
Then (\ref{06_05}) implies
\begin{align*}
 \mu_{2,j}(A)
\geq 
\lambda_j(A)-(\lambda_j(A)-\lambda_1(A))\|P_jP_{V_\ell}\|
\to \lambda_j(A) \quad (\ell\to\infty).
\end{align*}
Thus we have $\mu_{2,j}(A)\geq \lambda_j(A)$,
and we conclude $\mu_{2,j}(A)=\lambda_j(A)$.
\end{proof}

Next we recall some consequences about the eigenvalue counting function,
derived from the min-max principle.

\begin{proposition}
 \label{proposition_comparing}
For $j=1,2$, let $\mathcal{H}_j$ be a separable Hilbert space,
and let $A_j$ be a lower semi-bounded self-adjoint operator on $\mathcal{H}_j$.
Assume that there exists a linear isometry $T$ 
from $\mathcal{H}_1$ to $\mathcal{H}_2$  such that 
$T Q(q_{A_1})\subset Q(q_{A_2})$.
\begin{enumerate}
 \item[(1)]
Assume that
\begin{align*}
 q_{A_1}[u] \geq q_{A_2}[Tu]
\end{align*}
for every $u\in Q(q_{A_1})$.
Then, we have
\begin{align}
 N(\lambda;A_1) \leq  N(\lambda;A_2)
\label{06_06}
\end{align}
for every $\lambda\in \mathbf{R}$.
 \item[(2)]
Assume that
\begin{align*}
 q_{A_1}[u] \leq q_{A_2}[Tu]
\end{align*}
for every $u\in Q(q_{A_1})$, and
\begin{align}
\label{06_07}
 \dim(Q(q_{A_2})/TQ(q_{A_1}))=k.
\end{align}
Then, we have
\begin{align}
 N(\lambda;A_2) \leq  N(\lambda;A_1)+k
\label{06_08}
\end{align}
for every $\lambda\in \mathbf{R}$.
\end{enumerate}
\end{proposition}
\noindent\textbf{Remark.}
The statement (1) is well-known (e.g., \cite[Lemme 5.1]{Co}).
The statement (2) is considered to be well-known,
but we give a proof here for readers' use.
\begin{proof}
 Since $T$ is an isometry,
we can assume that $\mathcal{H}_1$
is a closed subspace of $\mathcal{H}_2$,
and $T$ is the inclusion operator $Tv=v$,
without loss of generality.
Then by assumption,
$Q(q_{A_1})\subset Q(q_{A_2})$.

(1) Since $Q(q_{A_1})\subset \mathcal{H}_1$, we have
\begin{align*}
& \{V\cap Q(q_{A_1});V\subset \mathcal{H}_1,\dim(\mathcal{H}_1/V)\leq j-1\}\\
=&\ 
 \{V\cap Q(q_{A_1});V\subset \mathcal{H}_2,\dim(\mathcal{H}_2/V)\leq j-1\}.
\end{align*}
Thus we have
\begin{align*}
 \lambda_j(A_1)
=&\ 
\sup_{V\subset \mathcal{H}_1,\codim V\leq j-1}
\inf_{u\in V\cap Q(q_{A_1}),\|u\|_{\mathcal{H}_1}=1}
q_{A_1}[u]\\
=&\ 
\sup_{V\subset \mathcal{H}_2,\codim V\leq j-1}
\inf_{u\in V\cap Q(q_{A_1}),\|u\|_{\mathcal{H}_2}=1}
q_{A_1}[u]\\
\geq&\ 
\sup_{V\subset \mathcal{H}_2,\codim V\leq j-1}
\inf_{u\in V\cap Q(q_{A_2}),\|u\|_{\mathcal{H}_2}=1}
q_{A_2}[u]\\
=&\ 
\lambda_j(A_2).
\end{align*}
Thus 
$\lambda_j(A_1)\leq \lambda$
implies
$\lambda_j(A_2)\leq \lambda$,
and (\ref{06_06}) holds.

(2) Let $W_1$ be a subspace of $\mathcal{H}_2$ such that
\begin{align*}
 W_1\cap Q(q_{A_2})=\{0\},\ 
W_1\oplus Q(q_{A_2})=\mathcal{H}_2.
\end{align*}
Put $W:=Q(q_{A_1})\oplus W_1$, then $\codim W=k$ by (\ref{06_07}).

If a subspace $V\subset \mathcal{H}_2$
satisfies $\codim V\leq j-1$,
then $\codim(V\cap W)\leq j+k-1$
and $V\cap Q(q_{A_1})=(V\cap W)\cap Q(q_{A_2})$.
Thus we have
\begin{align*}
 \lambda_{j+k}(A_2)
=&\ 
\sup_{\codim V\leq j+k-1}
\inf_{u\in Q(q_{A_2})\cap V}
q_{A_2}[u]\\
\geq&\ 
\sup_{\codim V\leq j-1}
\inf_{u\in Q(q_{A_2})\cap (V\cap W)}
q_{A_2}[u]\\
\geq&\ 
\sup_{\codim V\leq j-1}
\inf_{u\in Q(q_{A_1})\cap V}
q_{A_1}[u]\\
=&\ 
\lambda_j(A_1).
\end{align*}
Therefore,
$\lambda_{j+k}(A_2)\leq \lambda$ implies
$\lambda_j(A_1)\leq \lambda$,
and (\ref{06_08}) holds.
\end{proof}

\subsection{Boundary conditions}
\label{subsection_bdry}
In this subsection,
we give a rigorous definition of the operators $-\Delta^\sharp_{\alpha, Y, U}$
($\sharp=D$ or $N$),  given in the introduction.
The operators with point interactions on a bounded domain 
are studied by several authors (e.g., \cite{Ca-Sw,Bl-Fi-Ma, Lo-Mi}).
In particular,
the papers \cite{Bl-Fi-Ma, Lo-Mi}  define
$-\Delta^D_{\alpha, Y, U}$ by giving the operator domain,
using the basis of
the deficiency subspaces of the minimal operator
\begin{align*}
& -\Delta^D_{\min,Y,U}u
:= -\Delta|_{Y\cap U}u,\\
&u\in D( -\Delta^D_{\min,Y,U})
:=\{u\in H^2(U)\cap H^1_0(U);\\
&\hspace{4cm}u(y)=0\mbox{ for every }y\in Y\cap U\},
\end{align*}
provided that the boundary $\partial U$ 
is sufficiently smooth.
The paper \cite{Lo-Mi} also gives the form domain of the quadratic form
associated with $-\Delta^D_{\alpha, Y, U}$.
In Lemma \ref{lemma_DNdef} below,
we define the operators
$-\Delta^\sharp_{\alpha, Y, U}$ ($\sharp=D,N$)
by the method of quadratic form.
The form domain is defined by using compactly supported basis functions
$\phi_y$ ($y\in Y\cap U$), which is useful for our analysis.
The advantages of this method are as follows.
\begin{enumerate}
 \item[(i)] We can define
the operators
$-\Delta^\sharp_{\alpha, Y, U}$ ($\sharp=D,N$)
on any bounded open set $U$ with $\partial U\cap Y=\emptyset$.

 \item[(ii)] The properties of eigenvalue counting functions,
e.g., Dirichlet--Neumann bracketing,
are easily proved by the definitions of quadratic forms.
\end{enumerate}
If the boundary $\partial U$ of $U$ is sufficiently smooth,
then we prove that the corresponding self-adjoint operators
coincide with the ones defined in former papers.

For this purpose,
first we study the quadratic form $q_{\alpha,Y}$ 
associated with the operator $-\Delta_{\alpha,Y}$.
The quadratic form $q_{\alpha,Y}$  is constructed in 
\cite[Appendix F]{Al-Ge-Ho-Ho} and references therein,
by an indirect method using a quadratic form defined on a weighted $L^2$-space.
We give here a direct expression of $q_{\alpha,Y}$, 
which is suitable for our purpose.

When $d=1$, a simple expression of $q_{\alpha,Y}$ is 
well-known
(see e.g., \cite[page 79]{Al-Ge-Ho-Ho}).
Let $Y$ be a finite set in $\mathbf{R}$ and
$\alpha=(\alpha_y)_{y\in Y}$  a real-valued sequence.
Then the quadratic form $q_{\alpha,Y}$ associated with $-\Delta_{\alpha,Y}$
and the form domain $Q(q_{\alpha,Y})$ are given as
\begin{align*}
& q_{\alpha,Y}[u]:=\|u'\|^2_{L^2(\mathbf{R})}+\sum_{y\in Y}\alpha_y|u(y)|^2,
\\
&u\in Q(q_{\alpha,Y}):=H^1(\mathbf{R}).
\end{align*}
Notice that the value $u(y)$ makes sense,
since $u\in H^1(\mathbf{R})$ is continuous on $\mathbf{R}$
by the Sobolev embedding theorem.
Then, for a bounded open interval $U$ such that
$\partial U\cap Y=\emptyset$,
we define the operator $-\Delta^\sharp_{\alpha,Y,U}$ ($\sharp=D,N$) as
the self-adjoint operator corresponding to the quadratic form
$q^\sharp_{\alpha,Y,U}$ defined by
\begin{align*}
& q^D_{\alpha,Y,U}[u]:=
\|u'\|^2_{L^2(U)}+\sum_{y\in Y\cap U}\alpha_y|u(y)|^2,
\quad
u\in Q(q^D_{\alpha,Y,U}):=H^1_0(U),\\
& q^N_{\alpha,Y,U}[u]:=
\|u'\|^2_{L^2(U)}+\sum_{y\in Y\cap U}\alpha_y|u(y)|^2,
\quad
u\in Q(q^N_{\alpha,Y,U}):=H^1(U).
\end{align*}
Then it is easy to see that
$u\in D(-\Delta^D_{\alpha,Y,U})$ satisfies
the Dirichlet boundary condition $u(x)=0$ on $\partial U$,
and
$u\in D(-\Delta^N_{\alpha,Y,U})$ satisfies
the Neumann boundary condition $u'(x)=0$ on $\partial U$.

When $d=2$ or $3$,
we cannot define $q^\sharp_{\alpha,Y,U}$ in the same way,
because
\begin{enumerate}
 \item[(i)]
$u\in D(-\Delta_{\alpha,Y})$ \textit{does not} imply
$\nabla u\in L^2_{\rm loc}(\mathbf{R}^d)^d$,
so $\|\nabla u\|_{L^2(\mathbf{R}^d)^d}^2$ 
does not make sense, and

 \item[(ii)] $u(y)$ \textit{does not} make sense
for $u\in H^1(\mathbf{R}^d)$.
\end{enumerate}
So we use a different method using an operator core of $-\Delta_{\alpha,Y}$.

First we begin with an elementary lemma
about the Sobolev spaces.
\begin{lemma}
\label{lemma_density}
 Let $d\geq2$.
Let $Y$ be a finite set in $\mathbf{R}^d$.
Then, $C_0^\infty(\mathbf{R}^d\setminus Y)$
is dense in $H^1(\mathbf{R}^d)$.
\end{lemma}
\noindent\textbf{Remark.}
The statement of Lemma \ref{lemma_density} is a well-known consequence of 
the theory of polar sets (cf.\ \cite{Ad});
here we give a proof for readers' use.
The conclusion is true also for $d\geq 4$,
but it is false when $d=1$,
since the map $H^1(\mathbf{R})\ni u\mapsto u(y)$ is continuous.
\begin{proof}
Without loss of generality, we can assume that $Y=\{0\}$.
By \cite[Theorem 3.23]{Ad},
$C_0^\infty(\mathbf{R}^d\setminus\{0\})$ is dense in 
$H^1(\mathbf{R}^d)$ if and only if $\{0\}$ is $(1,2)$-polar,
that is,
if $u\in H^{-1}(\mathbf{R}^d)$ satisfies $\supp u\subset \{0\}$,
then $u=0$.
Here
\begin{align*}
 H^{-1}(\mathbf{R}^d)
:=
\{u\in \mathcal{S}'(\mathbf{R}^d); 
\int_{\mathbf{R}^d}|\hat{u}(\xi)|^2(1+|\xi|^2)^{-1}d\xi<\infty\},
\end{align*}
where $\mathcal{S}'(\mathbf{R}^d)$ is the space of the tempered distributions,
and $\hat{u}$ is the Fourier transform of $u$.

Assume that $u\in H^{-1}(\mathbf{R}^d)$ and $\supp u\subset \{0\}$.
Then, by \cite[Theorem 6.25]{Ru},
there exist $N\in \mathbf{N}$ and constants $c_{j_1,...,j_d}$ such that
\begin{align*}
 u=
\sum_{j_1+\cdots+j_d\leq N}c_{j_1,...,j_d}
\frac{\partial^{j_1+\cdots+j_d}}{\partial x_1^{j_1}\cdots \partial x_d^{j_d}}\delta_0,
\end{align*}
where 
the sum ranges over all non-negative integers 
$j_1,...,j_d$ with $j_1+\ldots +j_d\leq N$,
and $\delta_0$ is the Dirac delta distribution at $x=0$.
Taking the Fourier transform, we have
\begin{align*}
\hat{u}(\xi)
=
\frac{1}{(2\pi)^{d/2}}
\sum_{j_1+\cdots+j_d\leq N}c_{j_1,...,j_d}
(i \xi_1)^{j_1}\cdots (i\xi_d)^{j_d}.
\end{align*}
However, since $d\geq 2$, we have
\begin{align*}
 \int_{\mathbf{R}^d}|p(\xi)|^2(1+|\xi|^2)^{-1}d\xi=\infty
\end{align*}
for every non-zero polynomial $p(\xi)$.
Thus $u\in H^{-1}(\mathbf{R}^d)$ implies that $c_{j_1,...,j_d}=0$ for all
$j_1,...,j_d$, so $u=0$.
Thus the conclusion holds.
\end{proof}

Next we specify the quadratic form 
associated with  $-\Delta_{\alpha, Y}$,
when $d=2$ or $3$ and $Y$ is a finite set.
\begin{lemma}
 \label{lemma_quadratic-form}
Let $d=2$ or $3$.
Let $Y$ be a finite set in $\mathbf{R}^d$,
and $\alpha=(\alpha_y)_{y\in Y}$ be a sequence of real numbers.
Let $q_{\alpha,Y}:=q_{-\Delta_{\alpha,Y}}$ be the quadratic form 
associated with
the semi-bounded self-adjoint operator $-\Delta_{\alpha,Y}$,
and $Q(q_{\alpha,Y})$ be the form domain of $q_{\alpha,Y}$.
Let
\begin{align*}
 d_*:=\min_{y,y'\in Y, y\not=y'}|y-y'|,
\end{align*}
and $d_0$ be any number such that $0<d_0\leq d_*$.
Let $\chi\in C^\infty(\mathbf{R})$ such that
\begin{align*}
\begin{cases}
\chi(r)= 1 & (r\leq d_0/4),\\
0\leq \chi(r)\leq 1 &(d_0/4<r<d_0/3 ),\\
\chi(r)=0 & (r\geq d_0/3).
\end{cases}
\end{align*}
For each $y\in Y$, let
\begin{align*}
& \phi_y(x):=
\begin{cases}
 \left(-\frac{1}{2\pi}\log|x-y|+\alpha\right)\chi(|x-y|)
&(d=2),\\
 \left(\frac{1}{4\pi}|x-y|^{-1}+\alpha\right)\chi(|x-y|)
&(d=3),
\end{cases}
\\
& A_y:=\{x\in \mathbf{R}^d\,;\,d_0/4\leq |x-y|\leq d_0/3\}.
\notag
\end{align*}
Then, the following holds.
\begin{enumerate}
\item[(1)]
The form domain $Q(q_{\alpha,Y})$
is given as
\begin{align}
\label{06_09}
Q(q_{\alpha,Y})
=
H^1(\mathbf{R}^d)\oplus \bigoplus_{y\in Y}\mathbf{C}\phi_y,
\end{align}
where $\oplus$ denotes the algebraic direct sum of subspaces.
 \item[(2)]
The quadratic form $q_{\alpha,Y}$
is given as
\begin{align}
&\  q_{\alpha,Y}(u,v)
= 
(\nabla u_0,\nabla v_0)_{L^2(\mathbf{R}^d)^d}\notag\\
&\hspace{2cm}+
\sum_{y\in Y}
\Bigl(
(c_y \nabla \phi_y,\nabla v_0)_{L^2(A_y)^d}
+
(\nabla u_0,d_y \nabla \phi_y)_{L^2(A_y)^d}
\notag\\
&
\hspace{3cm}
-\overline{c_y}\int_{\partial A_y}\overline{\frac{\partial \phi_y}{\partial n}}v_0\,dS
-d_y\int_{\partial A_y}\overline{u_0}\frac{\partial \phi_y}{\partial n}dS
\notag\\
&\hspace{3cm}+
\overline{c_y}d_y(\phi_y,-\Delta \phi_y)_{L^2(A_y)}
\Bigr),
\label{06_10}
\\
&\ 
u=u_0+\sum_{y\in Y}c_y\phi_y,\ 
v=v_0+\sum_{y\in Y}d_y \phi_y\in Q(q_{\alpha,Y}),\notag\\
&\ 
u_0,v_0\in H^1(\mathbf{R}^d),\ c_y,d_y\in \mathbf{C}.
\notag
\end{align}
Here, $\partial \phi_y/\partial n=\nabla \phi_y\cdot n$,
$n$ is the unit outer normal vector on $\partial A_y$,
and $dS$ is the surface measure on $\partial A_y$.
\end{enumerate}
\end{lemma}
\noindent\textbf{Remark.}
1. The surface integrals in (\ref{06_10}) make sense,
since $u_0\in H^1(\mathbf{R}^d)$ implies that
$u_0|_{\partial A_y}\in H^{1/2}(\partial A_y;dS)\subset L^2(\partial A_y;dS)$,
by the trace theorem \cite[section 8]{Wl}.

2.
The equality (\ref{06_09}) implies that
the form domain $Q(q_{\alpha, Y})$ is independent of $\alpha$,
since the regular part of $\phi_y$, i.e.\ $\alpha_y \chi(|x-y|)$,
belongs to $H^1(\mathbf{R}^d)$,
and the singular part is independent of $\alpha$.
Similarly, the right hand side of (\ref{06_09}) 
is independent of the choice of $d_0$.
\begin{proof}
(1) By definition, $\phi_y\in L^2(\mathbf{R}^d)$ for $y\in Y$.
Moreover, it is easy to check that
\begin{align*}
& 
\Delta|_{\mathbf{R}^2\setminus Y}
\left(-\frac{1}{2\pi}\log|x-y|+\alpha\right)=0,\\
&\Delta|_{\mathbf{R}^3\setminus Y}
\left(\frac{1}{4\pi}|x-y|^{-1}+\alpha\right)=0.
\end{align*}
Then the Leibniz rule implies
\begin{align}
\supp (-\Delta|_{\mathbf{R}^d\setminus Y}\phi_y)
\subset \supp \nabla \chi(|\cdot-y|)
\subset A_y.
\label{06_11}
\end{align}
By (\ref{06_11}), we have $-\Delta|_{\mathbf{R}^d\setminus Y}\phi_y\in L^2(\mathbf{R}^d)$.
Moreover, since $\supp \phi_y \cap Y =\{y\}$ and
$\phi_y$ satisfies the boundary condition $(BC)_y$,
we see that $\phi_y\in D(-\Delta_{\alpha,Y})$.

Let $D_0=H^2_0(\mathbf{R}^d\setminus Y)$,
the closure of $C_0^\infty(\mathbf{R}^d\setminus Y)$
with respect to $H^2$-norm.
Since the symmetric operator $-\Delta|_{C_0^\infty(\mathbf{R}^d\setminus Y)}$
has deficiency indices $(N,N)$ ($N=\# Y$)(see \cite{Al-Ge-Ho-Ho}),
we see that
the dimension of the quotient space
$D(-\Delta_{\alpha,Y})/ D_0$
is equal to $N$.
Moreover, $\phi_y\not\in D_0$, since $D_0\subset H^2(\mathbf{R}^d)$
and any element of $H^2(\mathbf{R}^d)$ is continuous for $d=2$ or $3$,
by the Sobolev embedding theorem.
Thus the space
\begin{align*}
 \widetilde{D}
=C_0^\infty(\mathbf{R}^d\setminus{Y})\oplus 
\bigoplus_{y\in Y} \mathbf{C}\phi_y
\end{align*}
is an operator core of $-\Delta_{\alpha, Y}$.
So $\widetilde{D}$ is also a form core of $-\Delta_{\alpha, Y}$.

Since $Y$ is a finite set,
the operator $-\Delta_{\alpha,Y}$ is lower semi-bounded.
Let $M=-\inf \sigma(-\Delta_{\alpha,Y})+1$,
and define the norm $\|\cdot\|_Q$ on $\widetilde{D}$ by
\begin{align*}
 \|u\|_Q^2=(u,-\Delta|_{\mathbf{R}^d\setminus Y} u)+M\|u\|^2
\end{align*}
for $u\in \widetilde{D}$.
By the choice of $M$, we have
\begin{align*}
 \|u\|_Q^2\geq \|u\|^2
\end{align*}
for $u\in \widetilde{D}$,
so $\|\cdot\|_Q$ is actually a norm on $\widetilde{D}$.
Since $\widetilde{D}$ is a form core,
$Q(q_{\alpha,Y})$ is
the closure of $\widetilde{D}$ with respect to the norm $\|\cdot\|_Q$.
Thus, in order to prove (\ref{06_09}),
it is sufficient to prove
\begin{enumerate}
 \item $\widetilde{D}$ is dense in 
$\widetilde{Q}:=H^1(\mathbf{R}^d)\oplus\bigoplus_{y\in Y}\mathbf{C}\phi_y$,
and
 \item $\widetilde{Q}$ is a closed subspace of $Q(q_{\alpha,Y})$.
\end{enumerate}

To see (i), notice that
we have for $u_0\in C_0^\infty(\mathbf{R}^d\setminus Y)$
\begin{align*}
 \|u_0\|_Q^2
=
\|\nabla u_0\|^2+M\|u_0\|^2,
\end{align*}
which is equivalent to $H^1$-norm.
Thus 
the closure of 
$C_0^\infty(\mathbf{R}^d\setminus Y)$ 
with respect to $\|\cdot\|_Q$
coincides with
the closure of 
$C_0^\infty(\mathbf{R}^d\setminus Y)$ 
with respect to $H^1$-norm,
which is equal to $H^1(\mathbf{R}^d)$ 
by Lemma \ref{lemma_density}.
Thus (i) holds.

To see (ii),
notice that $Q(q_{\alpha,Y})$ is a Hilbert space
equipped with the inner product
\begin{align*}
 (u,v)_Q=q_{\alpha,Y}(u,v)+M(u,v),
\end{align*}
$H^1(\mathbf{R}^d)$ is a closed subspace of $Q(q_{\alpha,Y})$,
and $\dim(\widetilde{Q}/H^1(\mathbf{R}^d))=\#Y<\infty.$
Therefore
(ii) is a consequence of the following general statement;
`if $X$ is a Hilbert space,
$V_1$ is a closed subspace of $X$,
$V_2$ is a subspace of $X$ with $V_2\supset V_1$
and $\dim(V_2/V_1)<\infty$,
then $V_2$ is a closed subspace of $X$'.

(2)
Let $u,v\in \widetilde{D}$. Then $u$ is written as
\begin{align*}
& u=u_0+\sum_{y\in Y} c_y \phi_y,\ 
 v=v_0+\sum_{y\in Y} d_y \phi_y,\\
&
u_0,v_0\in C_0^\infty(\mathbf{R}^d\setminus Y),\ 
c_y,d_y \in \mathbf{C}.
\end{align*}
Then the quadratic form $q_{\alpha,Y}(u,v)=(u,-\Delta_{\alpha,Y}v)$ 
is calculated as follows.
\begin{align}
  (u,-\Delta_{\alpha,Y}v)
=&\ 
(u_0,-\Delta v_0)
+
\sum_{y\in Y}
\Bigl\{
(u_0,d_y(-\Delta|_{\mathbf{R}^d\setminus Y}\phi_y))\notag\\
&+
(c_y\phi_y,-\Delta v_0)
+
(c_y\phi_y,d_y(-\Delta|_{\mathbf{R}^d\setminus Y} \phi_y))
\Bigr\},
\label{06_12}
\end{align}
since $\supp \phi_y$ $(y\in Y)$ are disjoint.
By (\ref{06_11}), we have by integration by parts
\begin{align*}
 (u_0,-\Delta v_0)=&(\nabla u_0,\nabla v_0),\\
(u_0,d_y(-\Delta|_{\mathbf{R^d}\setminus Y} \phi_y))=&\ 
(u_0,d_y(-\Delta \phi_y))_{L^2(A_y)}\\
=&\ (\nabla u_0,d_y\nabla \phi_y)_{L^2(A_y)^d}
-d_y\int_{\partial A_y}\overline{u_0}
\frac{\partial \phi_y}{\partial n}dS,\\
(c_y\phi_y,-\Delta v_0)=&\ 
(c_y(-\Delta|_{\mathbf{R}^d\setminus Y}\phi_y), v_0)
=
(c_y(-\Delta \phi_y), v_0)_{L^2(A_y)}\\
=&\ 
(c_y \nabla \phi_y,\nabla v_0)_{L^2(A_y)^d}
-\overline{c_y}\int_{\partial A_y}
\overline{\frac{\partial \phi_y}{\partial n}}v_0
dS,\\
(c_y\phi_y,d_y(-\Delta|_{\mathbf{R}^d\setminus Y} \phi_y))
=&\ 
\overline{c_y}d_y
(\phi_y,-\Delta \phi_y)_{L^2(A_y)}.
\end{align*}
Combining these equalities with
(\ref{06_12}), we have (\ref{06_10})
for $u\in \widetilde{D}$.
Moreover, 
if we fix $c_y$ and $d_y$,
the right hand side of (\ref{06_10}) is continuous
with respect to $u_0,v_0\in H^1(\mathbf{R}^d)$,
by the trace theorem \cite{Wl}
(i.e.\ $H^1(\mathbf{R}^d)\ni u_0 \mapsto u_0|_{\partial A_y}\in L^2(\partial A_y;dS)$ is continuous).
Thus we have (\ref{06_10}) for every 
$u,v\in Q(q_{\alpha,Y})$.
\end{proof}

Now we can define our operators $-\Delta^\sharp_{\alpha, Y, U}$ 
($\sharp=D,N$) by the form method.
\begin{lemma}
 \label{lemma_DNdef}
Let $d=2$ or $3$. 
Let $Y$ be a locally finite set in $\mathbf{R}^d$,
and $\alpha=(\alpha_y)_{y\in Y}$ be a real-valued sequence.
Let $U$ be a bounded open set in $\mathbf{R}^d$
such that $\partial U\cap Y=\emptyset$,
and put $Y_U=Y\cap U$ and $\alpha_U=\alpha|_{Y_U}$.
Let
\begin{align*}
& d_*:=\min_{y,y'\in Y_U, y\not=y'}|y-y'|,
\end{align*}
and $d_0$ be any number satisfying
\begin{align*}
0< d_0\leq \min(d_*,\dist(Y_U,\partial U)).
\end{align*}
Define quadratic forms $q^\sharp_{\alpha,Y,U}$ ($\sharp=D,N$) by
\begin{align}
&Q(q^N_{\alpha,Y,U}):=
H^1(U)\oplus\bigoplus_{y\in Y_U}\mathbf{C}\phi_y,
\notag\\
&  q^N_{\alpha,Y,U}(u,v)
:=
(\nabla u_0,\nabla v_0)_{L^2(U)^d}\notag\\
&\hspace{2.5cm}
+
\sum_{y\in Y_U}
\Bigl(
(c_y \nabla \phi_y,\nabla v_0)_{L^2(A_y)^d}
+
(\nabla u_0,d_y \nabla \phi_y)_{L^2(A_y)^d}
\notag\\
&
\hspace{3.5cm}
-\overline{c_y}\int_{\partial A_y}\overline{\frac{\partial \phi_y}{\partial n}}v_0\,dS
-d_y\int_{\partial A_y}\overline{u_0}\frac{\partial \phi_y}{\partial n}dS
\notag\\
&\hspace{4cm}+
\overline{c_y}d_y(\phi_y,-\Delta \phi_y)_{L^2(A_y)}
\Bigr),
\label{06_13}
\\
&
u=u_0+\sum_{y\in Y_U}c_y \phi_y,\ 
v=v_0+\sum_{y\in Y_U}d_y \phi_y\in Q(q^N_{\alpha,Y,U}),\notag\\
&
u_0,v_0\in H^1(U),\ c_y,d_y\in \mathbf{C},
\notag\\
& Q(q^D_{\alpha,Y,U}):=H_0^1(U)\oplus\bigoplus_{y\in Y_U}\mathbf{C}\phi_y,
\notag\\
& q^D_{\alpha,Y,U}(u,v):=q^N_{\alpha,Y,U}(u,v),
\quad u,v\in Q(q^D_{\alpha,Y,U}),
\notag
\end{align}
where $\phi_y$ and $A_y$
are as in 
Lemma \ref{lemma_quadratic-form}.
Then, the following holds.
\begin{enumerate}
 \item[(1)] The quadratic forms
$q^\sharp_{\alpha,Y,U}$ ($\sharp=D,N$) are
symmetric, closed and lower semi-bounded.

 \item[(2)]
Let $-\Delta^\sharp_{\alpha,Y,U}$ ($\sharp=D,N$) be
the self-adjoint operators 
corresponding to the quadratic forms
$q^\sharp_{\alpha,Y,U}$.
\begin{enumerate}
 \item[(i)]
For $\sharp=D,N$,
every $u\in D(-\Delta^\sharp_{\alpha,Y,U})$
satisfies the boundary condition $(BC)_y$
for every $y\in Y_U$,
and belongs to $H^2_{\rm loc}(U\setminus Y)$.
Moreover,
\begin{align*}
 -\Delta^\sharp_{\alpha,Y,U}u=-\Delta|_{U\setminus Y}u
\end{align*}
for every $u\in D(-\Delta^\sharp_{\alpha,Y,U})$.

\item[(ii)]
Assume additionally that $U$ is $(2,1)$-smooth.
Then, the operator domains $D(-\Delta^\sharp_{\alpha,Y,U})$ 
($\sharp=D,N$) are given as
\begin{align}
 D(-\Delta^D_{\alpha,Y,U})
=&\ 
\{u\in H^2(U);u=0\mbox{ on }\partial U,\notag\\
&\ u(y)=0 \mbox{ for every }y\in Y_U\}
\oplus \bigoplus_{y\in Y_U}\mathbf{C}\phi_y,
\label{06_14}
\\
 D(-\Delta^N_{\alpha,Y,U})
=&\ 
\{u\in H^2(U);\frac{\partial u}{\partial n}=0\mbox{ on }\partial U,\notag\\
&\ u(y)=0 \mbox{ for every }y\in Y_U\}
\oplus \bigoplus_{y\in Y_U}\mathbf{C}\phi_y.
\label{06_15}
\end{align}

\item[(iii)]
For $\sharp=D$ or $N$,
let
$-\Delta^\sharp_U:=-\Delta^\sharp_{\emptyset,\emptyset,U}$
be the self-adjoint operator corresponding to the quadratic form
$q^\sharp_U:=q^\sharp_{\emptyset,\emptyset,U}$.
Then, 
the operator
$-\Delta^\sharp_{\alpha,Y,U}$
has compact resolvent
if and only if
$-\Delta^\sharp_U$ has compact resolvent.
\end{enumerate}
\end{enumerate}
\end{lemma}
\noindent\textbf{Remark.}
1. For the definition of $(2,1)$-smooth,
see \cite[Definition 2.7]{Wl}.
Roughly speaking, 
a $(2,1)$-smooth region is a region whose boundary is locally
represented as the graph of a $C^2$-function $f$,
and all second order derivatives of $f$ are Lipschitz continuous.

2. 
A lower semi-bounded self-adjoint operator $A$
has compact resolvent
if and only if $\lambda_j(A)\to \infty$ as $j\to \infty$,
so $N(\lambda;A)$ takes a finite value
for every $\lambda\in \mathbf{R}$.

3. The operator $-\Delta^D_U$ (resp.\ $-\Delta^N_U$)
is the usual Dirichlet Laplacian 
(resp.\ Neumann Laplacian) for $U$ (see \cite[section XIII.15]{Re-Si4}).
It is known that 
the operator $-\Delta^D_{U}$
has compact resolvent for every bounded open set $U$.
The operator
$-\Delta^N_{U}$
has compact resolvent
if $\partial U$  
has some regularity.
For example, if 
$U$ satisfies the uniform cone property (see \cite[Definition 2.3]{Wl}),
then there exists a continuous extension operator
from $H^1(U)$ to $H^1_0(\widetilde{U})$,
where $\widetilde{U}$ is a bounded open set such that 
$\overline{U}\subset \widetilde{U}$
(cf.\ \cite[Theorem 5.4]{Wl}),
so $-\Delta^N_{U}$ has compact resolvent.

\begin{proof}
(1) 
Clearly
the quadratic forms
$q^\sharp_{\alpha,Y,U}$ ($\sharp=D,N$) are symmetric.
First we show that 
$q^\sharp_{\alpha,Y,U}$ ($\sharp=D,N$) are lower semi-bounded.

Consider the case where $\sharp=D$.
Since $H^1_0(U)\subset H^1(\mathbf{R}^d)$,
we have that 
$Q(q^D_{\alpha,Y,U})\subset Q(q_{\alpha_U,Y_U})$ 
by Lemma \ref{lemma_quadratic-form}, and
\begin{align}
 q^D_{\alpha,Y,U}[u]
=
q_{\alpha_U,Y_U}[u]
\geq C\|u\|^2_{L^2(U)},\ 
C=\inf\sigma(-\Delta_{\alpha_U,Y_U})
\label{06_16}
\end{align}
for every $u\in Q(q^D_{\alpha,Y,U})$.
Thus $q^D_{\alpha,Y,U}$ is lower semi-bounded.

Next, consider the case where $\sharp=N$.
Let
\begin{align*}
& U_1=\{x\in U;\dist(x,\partial U)>d_0/2\},\\
& U_2=\{x\in U;d_0/3\leq \dist(x,\partial U)\leq d_0/2\},\\
& U_3=\{x\in U;\dist(x,\partial U)<d_0/3\}.
\end{align*}
Notice that $(U_2\cup U_3)\cap \supp \phi_y=\emptyset$ for $y\in Y$.
Then we can construct two functions $\chi_1,\chi_2\in C^\infty(U)$
satisfying the following conditions
(cf.\ \cite[Proposition 8.1]{Do-Iw-Mi}):
\begin{align*}
& \chi_1^2+\chi_2^2=1,\\
&
\begin{cases}
\chi_1(x) =1 & (x\in U_1),\\
0\leq \chi_1(x) \leq 1 & (x\in U_2),\\
\chi_1(x) =0 & (x\in U_3),
\end{cases}
\quad 
\begin{cases}
\chi_2(x) =0 & (x\in U_1),\\
0\leq \chi_2(x) \leq 1 & (x\in U_2),\\
\chi_2(x) =1 & (x\in U_3).
\end{cases}
\end{align*}
By a simple calculation (cf.\ \cite[Theorem 3.2]{Cy-Fr-Ki-Si} or 
\cite[(6.14)]{Do-Iw-Mi})
and (\ref{06_16}), 
we have
\begin{align}
 q^N_{\alpha,Y,U}[u]
=&
 q^D_{\alpha,Y,U}[\chi_1 u]
-
\|(\nabla\chi_1)u\|_{L^2(U)}^2
+
 q^N_{U}[\chi_2 u]
-
\|(\nabla\chi_2)u\|_{L^2(U)}^2\notag\\
\geq&\ 
(C
-\|\nabla\chi_1\|_\infty^2
-\|\nabla\chi_2\|_\infty^2
)\|u\|_{L^2(U)}^2
\label{06_17}
\end{align}
for every $u\in Q(q^N_{\alpha,Y,U})$,
since $\chi_1 u\in Q(q^D_{\alpha,Y,U})$ and
\begin{align*}
 q^N_{U}[\chi_2u]=\|\nabla(\chi_2 u)\|_{L^2(U)^d}^2 \geq 0.
\end{align*}
Thus $q^N_{\alpha,Y,U}$ is also lower semi-bounded.

Next we shall prove that $q^\sharp_{\alpha,Y,U}$ ($\sharp=D,N$) are closed,
that is,
the form domain $Q(q^\sharp_{\alpha,Y,U})$ is complete with 
respect to the norm $\|\cdot\|_{Q(q^\sharp_{\alpha,Y,U})}$.
For $\sharp=D$, it is obvious since $Q(q^D_{\alpha,Y,U})$
is a closed subspace of $Q(q_{\alpha_U,Y_U})$.
For $\sharp=N$,
we define the norm on $q^N_{\alpha,Y,U}$ by
\begin{align*}
& \|u\|^2_{Q(q^N_{\alpha,Y,U})}
:=
q^N_{\alpha,Y,U}[u]+M\|u\|_{L^2(U)}^2,\\
&
M:=
-C
+\|\nabla\chi_1\|_\infty^2
+\|\nabla\chi_2\|_\infty^2
+1
\end{align*}
for $u\in Q(q^N_{\alpha,Y,U})$,
where $C$ is the constant in (\ref{06_17}).
Take a Cauchy sequence 
$\{u_k\}_{k=1}^\infty$ in 
$(Q(q^N_{\alpha,Y,U}),\|\cdot\|_{Q(q^N_{\alpha,Y,U})})$.
Put $\eta_j=\chi_j^2$ ($j=1,2$),
where $\chi_j$ is the function defined above.
Then, the maps
\begin{align*}
& Q(q^N_{\alpha,Y,U})\ni u \mapsto \eta_1 u\in  Q(q^D_{\alpha,Y,U}),\\
& Q(q^N_{\alpha,Y,U})\ni u \mapsto \eta_2 u\in  Q(q^N_{U})
\end{align*}
are continuous, 
so $\{\eta_j u_k\}_{k=1}^\infty$ ($j=1,2$) are Cauchy sequences.
Since the quadratic forms 
$q^D_{\alpha,Y,U}$
and
$q^N_{U}$
are closed
(notice that
$Q(q^N_{U})=H^1(U)$ is complete),
there exist 
$v_1\in Q(q^D_{\alpha,Y,U})$
and 
$v_2\in Q(q^N_{U})$ such that
\begin{align*}
& \|\eta_1 u_k-v_1\|_{Q(q^D_{\alpha,Y,U})}\to 0\quad (k\to \infty),\\
& \|\eta_2 u_k-v_2\|_{Q(q^N_{U})}\to 0\quad (k\to \infty).
\end{align*}
Thus we have $u_k=(\eta_1+\eta_2)u_k\to v_1+v_2\in Q(q^N_{\alpha,Y,U})$,
and we conclude that $Q(q^N_{\alpha,Y,U})$ is complete.

(2) 
First, the self-adjoint operators $-\Delta^\sharp_{\alpha,Y,U}$ 
($\sharp=D,N$) are well-defined, 
by the statement (1) and the method of quadratic forms.

(i) Let $v\in D(-\Delta^\sharp_{\alpha,Y,U})$.
Since $D(-\Delta^\sharp_{\alpha,Y,U})\subset
Q(q^\sharp_{\alpha,Y,U})$,
$v$ is written as
\begin{align*}
 v=v_0+\sum_{y\in Y_U}d_y\phi_y,\  
d_y\in \mathbf{C},\ 
v_0\in 
\begin{cases}
 H^1_0(U) & (\sharp=D),\\
 H^1(U) & (\sharp=N),
\end{cases}
\end{align*}
and there exists $-\Delta^\sharp_{\alpha,Y,U} v\in L^2(U)$ such that
\begin{align}
\label{06_18}
 q^\sharp_{\alpha,Y,U}(u,v)=(u,-\Delta^\sharp_{\alpha,Y,U} v)_{L^2(U)}
\end{align}
for every $u\in Q(q^\sharp_{\alpha,Y,U})$.
If we put $u_0\in C_0^\infty(U\setminus Y)$
and $c_y=0$ in (\ref{06_13}),
we have by (\ref{06_11}) and integration by parts
\begin{align}
q^\sharp_{\alpha,Y,U}(u_0,v)
=&\ 
(\nabla u_0,\nabla v_0)_{L^2(U)}\notag\\
&+
\sum_{y\in Y_U}
\Bigl(
(\nabla u_0,d_y \nabla \phi_y)_{L^2(A_y)^d}
-d_y\int_{\partial A_y}\overline{u_0}\frac{\partial \phi_y}{\partial n}dS
\Bigr)
\notag\\
&=(-\Delta u_0, v_0)_{L^2(U)}+
\sum_{y\in Y_U}(u_0, -\Delta(d_y\phi_y))_{L^2(A_y)}
\notag\\
&=
\left(-\Delta u_0,v_0+\sum_{y\in Y_U}d_y \phi_y\right)_{L^2(U)}
=
(-\Delta u_0,v)_{L^2(U)}.
\label{06_19}
\end{align}
Comparing (\ref{06_18}) with $u=u_0$ and (\ref{06_19}),
we have 
\begin{align}
 -\Delta^\sharp_{\alpha,Y,U}v= -\Delta|_{U\setminus Y}v,
\label{06_20}
\end{align}
where $-\Delta|_{U\setminus Y}$ is the distributional 
Laplacian on $U\setminus Y$.
Moreover, (\ref{06_20}) implies that
$-\Delta|_{U\setminus Y}v\in L^2(U)$,
then 
we have $v\in H^2_{\rm loc}(U\setminus Y)$
by the elliptic interior regularity (cf.\ \cite{Ag}).

Take $\eta\in C_0^\infty(U)$ such that $\eta=1$ 
on $\bigcup_{y\in Y_U}\supp \phi_y$.
Then, by putting $\eta v=0$ on $U^c$,
we have $\eta v\in L^2(\mathbf{R}^d)$ and
\begin{align*}
-\Delta_{\mathbf{R^d}\setminus Y_U}(\eta v)
=(-\Delta \eta)v-2\nabla\eta\cdot \nabla v
+\eta(-\Delta|_{U\setminus Y} v)
\in L^2(\mathbf{R}^d), 
\end{align*}
since $v\in H^2_{\rm loc}(U\setminus Y)$,
$-\Delta|_{U\setminus Y} v\in L^2(U)$,
and $\supp \nabla \eta$ and $\supp\Delta \eta$ 
are compact sets in $U\setminus Y$.
By the definition of the distributional derivative,
we have
\begin{align}
 (u,-\Delta_{\mathbf{R}^d\setminus Y_U}(\eta v))_{L^2(\mathbf{R}^d)}
=
 (-\Delta u,\eta v)_{L^2(\mathbf{R}^d)}
\label{06_21}
\end{align}
for every $u\in C_0^\infty(\mathbf{R}^d\setminus Y_U)$.
The equality (\ref{06_21}) implies that $\eta v\in D(-\Delta_{\min,Y_U}^*)$,
where $-\Delta_{\min ,Y_U}$ is the minimal operator defined by
\begin{align*}
& -\Delta_{\min ,Y_U}u :=-\Delta u,\\
&u\in D(-\Delta_{\min ,Y_U}):=H^2_0(\mathbf{R}^d\setminus Y_U),\\
&H^2_0(\mathbf{R}^d\setminus Y_U)
=\{u\in H^2(\mathbf{R}^d);u(y)=0\mbox{ for every }y\in Y_U\},
\end{align*}
and $-\Delta_{\min,Y_U}^*$ is the adjoint operator of $-\Delta_{\min,Y_U}$
(or the maximal operator),
explicitly given by
\begin{align*}
& -\Delta_{\min ,Y_U}^*u =-\Delta|_{\mathbf{R}^d\setminus Y_U}u,\\
&u\in D(-\Delta_{\min ,Y_U}^*)=
\{u\in L^2(\mathbf{R}^d);-\Delta|_{\mathbf{R}^d\setminus Y_U}u\in L^2(\mathbf{R}^d) \}.
\end{align*}
The operator $-\Delta_{\min ,Y_U}$ is closed and symmetric,
and it is known that the deficiency indices of 
$-\Delta_{\min ,Y_U}$ are $(N,N)$
($N=\# Y_U$; see \cite{Al-Ge-Ho-Ho}).
Thus we have
\begin{align*}
\dim(D(-\Delta_{\min ,Y_U}^*)/D(-\Delta_{\min ,Y_U}))=2N.
\end{align*}
Put $\psi_y(x):=\chi(|x-y|)$,
where $\chi$ is the function 
given in Lemma \ref{lemma_quadratic-form}.
Then the vectors $\{\phi_y,\psi_y\}_{y\in Y_U}$ 
form a basis of 
$D(-\Delta_{\min ,Y_U}^*)/ D(-\Delta_{\min ,Y_U})$,
so we have 
\begin{align*}
& D(-\Delta_{\min ,Y_U}^*)
=
D(-\Delta_{\min ,Y_U})\oplus\bigoplus_{y\in Y_U}
(\mathbf{C}\phi_y\oplus \mathbf{C}\psi_y).
\end{align*}
Since $\eta v\in D(-\Delta_{\min ,Y_U}^*)$,
$\eta v$ is represented as the following form
\begin{align*}
& \eta v = \widetilde{v_1} + \sum_{y\in Y_U}(d_y \phi_y+e_y \psi_y),
\ d_y,e_y\in \mathbf{C},
\\
&\widetilde{v_1}\in H^2_0(\mathbf{R}^d\setminus Y_U),\ 
\supp \widetilde{v_1}\subset \supp \eta.
\end{align*}
We regard $\eta v$ and $\widetilde{v_1}$ as the functions on $U$,
so $\widetilde{v_1}\in H^2_0(U\setminus Y)$.
Then $v$ is represented as
\begin{align}
\label{06_22}
& v = v_1+ \sum_{y\in Y_U}(d_y \phi_y+e_y \psi_y),\ 
d_y,e_y\in \mathbf{C},
\\
& v_1:
=
(1-\eta)v+\widetilde{v_1}
=
(1-\eta)v_0+\widetilde{v_1}
\in  
\begin{cases}
H^1(U)\cap H^2_{\rm loc}(U) & (\sharp=N),\\
H^1_0(U)\cap H^2_{\rm loc}(U) & (\sharp=D).
\end{cases}
\label{06_23}
\end{align}
Since $\eta=1$ on $\bigcup_{y\in Y_U}\supp \phi_y$,
$v\in L^2(U)\cap H^2_{\rm loc}(U\setminus Y)$,
$-\Delta|_{U\setminus Y}v\in L^2(U)$, 
and $\widetilde{v_1}\in H^2_0(U\setminus Y)$,
we have
\begin{align}
&
v_1\in L^2(U),\ 
-\Delta|_U v_1= -\Delta|_U((1-\eta)v)-\Delta|_U \widetilde{v_1}\in L^2(U),
\label{06_24}\\
&v_1(y)=0\mbox{ for every }y\in Y_U.
\label{06_25}
\end{align}

Next, put $u_0=0$, $c_{y'}=\delta_{y y'}$, 
$v_0=v_1+\sum_{y\in Y_U}e_y \psi_y$ in (\ref{06_13}),
we have by (\ref{06_22})
\begin{align}
&\ q^\sharp_{\alpha,Y,U}(\phi_y,v)\notag\\
=&\ \left(\nabla \phi_y,\nabla(v_1+e_y \psi_y)\right)_{L^2(A_y)^d}\notag\\
&-\int_{\partial A_y}\overline{\frac{\partial \phi_y}{\partial n}}
(v_1+e_y \psi_y)dS
+
d_y(\phi_y,-\Delta \phi_y)_{L^2(A_y)}\notag\\
=&\ 
(-\Delta|_{U\setminus Y}\phi_y, v_1+e_y\psi_y)_{L^2(U)}
+
d_y(\phi_y,-\Delta|_{U\setminus Y}\phi_y)_{L^2(U)}.
\label{06_26}
\end{align}
On the other hand, (\ref{06_18}) and (\ref{06_20}) imply
\begin{align}
 q^\sharp_{\alpha,Y,U}(\phi_y,v)
=&\ 
(\phi_y,-\Delta_{U\setminus Y}v)_{L^2(U)}\notag\\
=&\ 
(\phi_y,-\Delta|_{U\setminus Y}(v_1+d_y\phi_y+e_y\psi_y))_{L^2(U)}.
\label{06_27}
\end{align}
Notice that $v_1=\widetilde{v_1}$ on $\supp \phi_y$ and
\begin{align}
&\  (-\Delta|_{U\setminus Y}\phi_y,v_1)_{L^2(U)}
=
 (-\Delta|_{U\setminus Y}\phi_y,\widetilde{v_1})_{L^2(U)}
\notag\\
=&\ 
 (\phi_y,-\Delta|_{U\setminus Y}\widetilde{v_1})_{L^2(U)}
=
 (\phi_y,-\Delta|_{U\setminus Y} v_1)_{L^2(U)},
\label{06_28}
\end{align}
since (\ref{06_28}) holds 
when $\widetilde{v_1}\in C_0^\infty(U\setminus Y)$,
and all inner products  in (\ref{06_28}) are continuous
with respect to $\widetilde{v_1}\in H^2_0(U\setminus Y)$.
Thus (\ref{06_26}), (\ref{06_27}), and (\ref{06_28}) imply
\begin{align}
 e_y
\bigl\{
(-\Delta|_{U\setminus Y} \phi_y,\psi_y)_{L^2(U)}
-
 (\phi_y,-\Delta|_{U\setminus Y} \psi_y)_{L^2(U)}
\bigr\}=0.
\label{06_29}
\end{align}
 Moreover,
when $d=2$, we have by integration by parts
\begin{align}
& (-\Delta|_{U\setminus Y}\phi_y,\psi_y)_{L^2(U)}
-
 (\phi_y,-\Delta|_{U\setminus Y}\psi_y)_{L^2(U)}
\notag\\
=&\ 
\int_0^{d_0/3}
\biggl\{
\chi(r)\left(\frac{1}{r}\frac{d}{dr}r \frac{d}{dr}\right)
(\chi(r)(\log r-2\pi\alpha_y))
\notag\\
&\hspace{2cm}
-
\chi(r)(\log r-2\pi\alpha_y)\left(\frac{1}{r}\frac{d}{dr}r \frac{d}{dr}\right)\chi(r)
\biggr\}rdr
\notag\\
=&\ 
\left[
\chi(r)r\frac{d}{dr}(\chi(r)\log r)-\chi(r)\log r \cdot r \frac{d}{dr}\chi(r)
\right]_0^{d_0/3}
\notag\\
=&\ 
\left[
\chi(r)^2
\right]_0^{d_0/3}=-1,
\label{06_30}
\end{align}
where we used the polar coordinate change $x-y=r(\cos\theta,\sin\theta)$.
When $d=3$, we have similarly
\begin{align}
& (-\Delta|_{U\setminus Y}\phi_y,\psi_y)_{L^2(U)}
-
 (\phi_y,-\Delta|_{U\setminus Y}\psi_y)_{L^2(U)}
\notag\\
=&\ 
\int_0^{d_0/3}
\biggl\{
\chi(r)\left(-\frac{1}{r^2}\frac{d}{dr}r^2 \frac{d}{dr}\right)
(\chi(r)(r^{-1}+4\pi\alpha_y))
\notag\\
&\hspace{2cm}
-
\chi(r)(r^{-1}+4\pi\alpha_y)\left(-\frac{1}{r^2}\frac{d}{dr}r^2 \frac{d}{dr}\right)\chi(r)
\biggr\}r^2dr
\notag\\
=&\ 
\left[
-\chi(r)r^2\frac{d}{dr}(\chi(r)r^{-1})+\chi(r)r^{-1} \cdot r^2 \frac{d}{dr}\chi(r)
\right]_0^{d_0/3}
\notag\\
=&\ 
\left[
\chi(r)^2
\right]_0^{d_0/3}=-1.
\label{06_31}
 \end{align}
By (\ref{06_29}), (\ref{06_30}) and (\ref{06_31}),
we have $e_y=0$.
Thus we have by 
(\ref{06_22})
\begin{align}
\label{06_32}
 v = v_1 + \sum_{y\in Y_U}d_y \phi_y,
\end{align}
where $v_1$ satisfies 
(\ref{06_23}), (\ref{06_24}) and (\ref{06_25}).
Then (\ref{06_32}) and 
(\ref{06_25}) means that $v$ satisfies $(BC)_y$.

(ii) 
We additionally assume that
$U$ is $(2,1)$-smooth,
and continue the argument of part (i).
Let $\widetilde{-\Delta^D_{\alpha,Y,U}}$
(resp.\ $\widetilde{-\Delta^N_{\alpha,Y,U}}$)
be the operator $-\Delta|_{U\setminus Y}$
with the operator domain 
given by (\ref{06_14}) (resp.\ (\ref{06_15})).

First we shall prove that
\begin{align}
D(\widetilde{-\Delta^\sharp_{\alpha,Y,U}})\subset 
D(-\Delta^\sharp_{\alpha,Y,U})\quad
(\sharp=D,N).
\label{06_33}
\end{align}
Let
\begin{align*}
&u=u_0+\sum_{y\in Y_U}c_y \phi_y\in Q(q^\sharp_{\alpha,Y}),
\ u_0\in 
\begin{cases}
H^1_0(U) & (\sharp=D),\\
H^1(U) & (\sharp=N),
\end{cases}\ 
c_y\in \mathbf{C},\\
&v=v_0+\sum_{y\in Y_U}d_y \phi_y\in D(\widetilde{-\Delta^\sharp_{\alpha,Y,U}}),\ 
v_0\in H^2(U),\ d_y\in \mathbf{C},\\
& v_0(y)=0\mbox{ for every $y\in Y_U$},\ 
\begin{cases}
 v_0|_{\partial U}=0 & (\sharp=D),\\
\frac{\partial v_0}{\partial n}|_{\partial U}=0 & (\sharp=N).
\end{cases}
\end{align*}
Since $U$ is $(2,1)$-smooth, we have
by the definition (\ref{06_13}) of $q^\sharp_{\alpha,Y,U}$
and integration by parts
\begin{align}
&\  q^\sharp_{\alpha,Y,U}(u,v)\notag\\
=&\ 
(u_0,-\Delta v_0)_{L^2(U)} +\int_{\partial U}\overline{u_0}\frac{\partial v_0}{\partial n}dS\notag\\
&+\sum_{y\in Y_U}
\biggl(
\overline{c_y}(-\Delta \phi_y, v_0)_{L^2(A_y)}
+
d_y(u_0,-\Delta \phi_y)_{L^2(A_y)}\notag\\
&
\hspace{1.5cm}
+\overline{c_y}d_y(\phi_y,-\Delta \phi_y)_{L^2(A_y)}
\biggr).
\label{06_34}
\end{align}
By (\ref{06_11}) and (\ref{06_28}), we have
\begin{align}
 (-\Delta \phi_y,v_0)_{L^2(A_y)}
=
 (-\Delta|_{U\setminus Y} \phi_y,v_0)_{L^2(U)}
=
 (\phi_y,-\Delta v_0)_{L^2(U)}.
\label{06_35}
\end{align}
By (\ref{06_11}), (\ref{06_34}) and (\ref{06_35}), we have
\begin{align*}
 q^\sharp_{\alpha,Y,U}(u,v)
=
(u,-\Delta|_{U\setminus Y}v)_{L^2(U)}
+\int_{\partial U}\overline{u_0}\frac{\partial v_0}{\partial n}dS.
\end{align*}
The last surface integral vanishes,
since $u_0|_{\partial U}=0$ ($\sharp=D$) or
$\partial v_0/\partial n|_{\partial U}=0$ ($\sharp=N$).
Thus we have
\begin{align}
 q^\sharp_{\alpha,Y,U}(u,v)=
(u,-\Delta|_{U\setminus Y}v)_{L^2(U)}
\label{06_36}
\end{align}
for every $u\in Q(q^\sharp_{\alpha,Y,U})$.
The equality (\ref{06_36}) implies that
$v\in D(-\Delta^\sharp_{\alpha,Y,U})$ and
$-\Delta^\sharp_{\alpha,Y,U}v=-\Delta|_{U\setminus Y}v$.
Therefore (\ref{06_33}) holds.

Next, we shall prove
\begin{align}
 \label{06_37}
D(-\Delta^\sharp_{\alpha,Y,U})
\subset 
D(\widetilde{-\Delta^\sharp_{\alpha,Y,U}})
\quad
(\sharp=D,N).
\end{align}
Take $v\in D(-\Delta^\sharp_{\alpha,Y,U})$.
In the proof of the statement (i),
we have proved that
$v$ is represented as the form
(\ref{06_32}), where $v_1$ satisfies
(\ref{06_23}), (\ref{06_24}), and (\ref{06_25}).
Since $\phi_y\in D(\widetilde{-\Delta^\sharp_{\alpha,Y,U}})\subset D(-\Delta^\sharp_{\alpha,Y,U})$,
we have
\begin{align}
\label{06_38}
 v_1 = v-\sum_{y\in Y_U} d_y \phi_y\in 
D(-\Delta^\sharp_{\alpha,Y,U}).
\end{align}
And (\ref{06_38}) implies that
\begin{align}
& q^\sharp_{\alpha,Y,U}(u,v_1)
=
(\nabla u,\nabla v_1)_{L^2(U)^d}
=
(u,-\Delta v_1 )_{L^2(U)}
\label{06_39}
\end{align}
for every $u\in H^1_0(U)$ ($\sharp=D$)
or $u\in H^1(U)$ ($\sharp=N$).
Under 
the assumption that $U$ is $(2,1)$-smooth,
the conditions 
(\ref{06_23}), (\ref{06_24}), and (\ref{06_39})
imply that 
\begin{align}
 v_1\in 
\begin{cases}
H^2(U)\cap H^1_0(U) & (\sharp=D),\\
H^2(U) & (\sharp=N),
\end{cases}
\label{06_40}
\end{align}
by the regularity theorem of the elliptic equation (\cite[Theorem 20.4]{Wl}).

Notice that $v=v_1$ around the boundary $\partial U$ of $U$.
We shall check that $v_1$ satisfies the corresponding boundary conditions.
If $\sharp=D$, then 
(\ref{06_40}) implies that $v_1\in H^1_0(U)$,
so $v_1|_{\partial U}=0$ by the trace theorem,
since $U$ is $(2,1)$-smooth.

Next, let $\sharp=N$.
Since $v_1\in H^2(U)$ and $U$ is $(2,1)$-smooth,
we have by integration by parts
\begin{align}
 (\nabla u, \nabla v_1)_{L^2(U)^d}
=
(u,-\Delta v_1)_{L^2(U)}+\int_{\partial U}\overline{u}\frac{\partial v}{\partial n}dS
\label{06_41}
\end{align}
for every $u\in H^1(U)$.
Comparing (\ref{06_39}) and (\ref{06_41}),
we have
\begin{align}
\label{06_42}
\int_{\partial U}\overline{u}\frac{\partial v}{\partial n}dS=0
\end{align}
for every $u\in H^1(U)$.
If $U$ is $(2,1)$-smooth,
then the range of $H^1(U)\ni u \mapsto u|_{\partial U}\in L^2(\partial U;dS)$ 
is dense in $L^2(\partial U;dS)$
(actually, the range is equal to $H^{1/2}(\partial U;dS)$ 
by \cite[Theorem 8.9]{Wl}),
and (\ref{06_42}) implies ${\partial v}/{\partial n}=0$ on $\partial U$.
Thus we conclude (\ref{06_37}) holds,
and the statement (ii) is proved.

(iii)
By the definition of $Q(q^\sharp_{\alpha,Y,U})$,
we have
\begin{align*}
&q^\sharp_{U}[u]
=
q^\sharp_{\alpha,Y,U}[u]\quad (u\in Q(q^\sharp_{U})),\\
& \dim(Q(q^\sharp_{\alpha,Y,U})/Q(q^\sharp_{U}))
=\# (Y\cap U)=:k.
\end{align*}
Then the statements (1) and (2) of Proposition \ref{proposition_comparing}
imply
\begin{align*}
 \lambda_j(-\Delta^\sharp_{U})
\leq
 \lambda_{j+k}(-\Delta^\sharp_{\alpha,Y,U})
\leq
 \lambda_{j+k}(-\Delta^\sharp_{U})
\end{align*}
for every $j\in \mathbf{N}$.
The conclusion follows from this inequality.
\end{proof}

Next we shall prove some properties of eigenvalue counting functions,
including the Dirichlet--Neumann bracketing technique
(\cite[section XIII-15]{{Re-Si4}}).
For $\lambda\in \mathbf{R}$ and a bounded open set $U$,
we denote $N^\sharp_{\alpha, Y,U}(\lambda):=N(\lambda;-\Delta^\sharp_{\alpha,Y,U})$ ($\sharp=D,N$).

\begin{proposition}
 \label{proposition_DNbracketing}
Let $d=1,2$ or $3$.
Let $Y$ be a locally finite set in $\mathbf{R}^d$,
and $\alpha=(\alpha_y)_{y\in Y}$ be a sequence of real numbers.
\begin{enumerate}
 \item [(1)]
Let $U$ be a bounded open set in $\mathbf{R}^d$,
such that $\partial U\cap Y=\emptyset$.
Then, we have for every $\lambda\in \mathbf{R}$
\begin{align*}
N^D_{\alpha,Y,U}(\lambda)
\leq 
N^N_{\alpha,Y,U}(\lambda).
\end{align*}

 \item[(2)]
Let $U_1, \ldots, U_n$ be disjoint bounded open sets in $\mathbf{R}^d$
such that
\begin{enumerate}
 \item[(i)]
$\bigoplus_{j=1}^k L^2(U_j)=L^2(U)$,
where $U$ is the interior of $\overline{U_1\cup\cdots \cup U_n}$, and
 \item[(ii)]  $\partial U_j\cap Y=\emptyset$ for every $j=1,...,n$.
\end{enumerate}
Then, we have for every $\lambda\in \mathbf{R}$
\begin{align}
& N^N_{\alpha,Y,U}(\lambda)
\leq
\sum_{j=1}^n 
 N^N_{\alpha,Y,U_j}(\lambda),
\label{06_43}
\\
& N^D_{\alpha,Y,U}(\lambda)
\geq
\sum_{j=1}^n 
 N^D_{\alpha,Y,U_j}(\lambda).
\label{06_44}
\end{align}

 \item[(3)] Let $U$ be a bounded open set in $\mathbf{R}^d$,
such that $\partial U\cap Y=\emptyset$.
Let $\widetilde{Y}$ is a subset of $Y_U:=Y\cap U$,
and $\widetilde{\alpha}=\alpha|_{\widetilde{Y}}$.
Then, we have for $\sharp=D,N$ and for every $\lambda\in \mathbf{R}$
\begin{align}
\label{06_45}
 N^\sharp_{\alpha,Y,U}(\lambda)
\leq
 N^\sharp_{\widetilde{\alpha},\widetilde{Y},U}(\lambda)+\#(Y_U\setminus \widetilde{Y}).
\end{align}
\end{enumerate}
\end{proposition}
\begin{proof}
(1) 
By Lemma \ref{lemma_DNdef},
we have $Q(q^D_{\alpha,Y,U})\subset Q(q^N_{\alpha,Y,U})$
and $q^D_{\alpha,Y,U}[u]=q^N_{\alpha,Y,U}[u]$
for $u\in Q(q^D_{\alpha,Y,U})$.
Then the conclusion follows from 
the statement (1) of Proposition \ref{proposition_comparing}
with $T=Id$.

(2)
Let $V=\bigcup_{j=1}^n U_j$,
so $U$ is the interior of $\overline{V}$.
By assumption, $L^2(U)$ and $L^2(V)$ can be identified.
By Lemma \ref{lemma_DNdef},
we have 
\begin{align*}
& Q(q^N_{\alpha,Y,U})
\subset
Q(q^N_{\alpha,Y,V}),\\
&q^N_{\alpha,Y,U}[u]
=q^N_{\alpha,Y,V}[u],\quad
u\in Q(q^N_{\alpha,Y,U}),\\
& Q(q^D_{\alpha,Y,U})
\supset
Q(q^D_{\alpha,Y,V}),\\
&q^D_{\alpha,Y,U}[u]
=q^D_{\alpha,Y,V}[u],\quad
u\in Q(q^D_{\alpha,Y,V}).
\end{align*}
Then (\ref{06_43}) and (\ref{06_44}) 
are consequences of 
the statement (1) of Proposition \ref{proposition_comparing},
with $T=Id$.

(3)
When $d=2$ or $3$,
(\ref{06_45}) is a consequence of
Lemma \ref{lemma_DNdef}
and the statement (2) of Proposition \ref{proposition_comparing} with $T=Id$.
When $d=1$,
we assume that $Y_U=\{y_1,...,y_k\}\cup\widetilde{Y}$, where $k=\#(Y_U\setminus\widetilde{Y})$.
Take a subspace $W$ of $L^2(U)$ such that 
$Q(q^\sharp_{\alpha,Y,U})\oplus W=L^2(U)$ and 
$Q(q^\sharp_{\alpha,Y,U})\cap W=\{0\}$.
If a subspace $V$ of $L^2(U)$ satisfies $\codim V\leq j-1$,
then the subspace
\begin{align*}
 \widetilde{V}:=\{u\in V\cap Q(q^\sharp_{\alpha,Y,U});u(y_1)=\cdots=u(y_k)=0\}\oplus W
\end{align*}
satisfies $\codim \widetilde{V}\leq j+k-1$.
Then we have by Proposition \ref{proposition_min-max}
\begin{align*}
&\  \lambda_{j+k}(-\Delta^\sharp_{\alpha,Y,U})\\
=&\ 
\sup_{\codim V\leq j+k-1}
\left(
\inf_{u\in Q(q^\sharp_{\alpha,Y,U})\cap V,\,\|u\|_{L^2(U)}=1}q^\sharp_{\alpha,Y,U}[u]\right)\\
\geq &\ 
\sup_{\codim V\leq j-1}
\left(
\inf_{u\in Q(q^\sharp_{\alpha,Y,U})\cap \widetilde{V},\|u\|_{L^2(U)}=1}q^\sharp_{\widetilde{\alpha},\widetilde{Y},U}[u]\right)\\
\geq &\ 
\sup_{\codim V\leq j-1}
\left(
\inf_{u\in Q(q^\sharp_{\alpha,Y,U})\cap V,\|u\|_{L^2(U)}=1}q^\sharp_{\widetilde\alpha,\widetilde{Y},U}[u]\right)\\
=&\ \lambda_{j}(-\Delta_{\widetilde{\alpha},\widetilde{Y},U}^\sharp).
\end{align*}
Thus we have (\ref{06_45}) for $d=1$.
\end{proof}
For $\lambda<0$,
we can compare counting functions $N^\sharp_{\alpha,Y,U}(\lambda)$
($\sharp=D,N$)
with the third counting function $N_{\alpha,Y,U}(\lambda):=N(\lambda;-\Delta_{\alpha_U,Y_U})$,
where $Y_U=Y\cap U$ and $\alpha_U=\alpha|_{Y_U}$.
\begin{proposition}
 \label{proposition_DNbracketing2}
Let $d=1,2$ or $3$.
Let $U$ be a bounded open set
such that 
$L^2(\mathbf{R}^d)=L^2(U)\oplus L^2(V)$,
where $V$ is the interior of $\mathbf{R}^d\setminus U$.
Let $Y$ be a locally finite set in $\mathbf{R}^d$
with $Y\cap \partial U=\emptyset$,
and $\alpha=(\alpha_y)_{y\in Y}$ be a sequence of real numbers.
Then, for any $\lambda<0$, we have
\begin{align*}
N^D_{\alpha,Y,U}(\lambda)
\leq 
N_{\alpha,Y,U}(\lambda)
\leq
N^N_{\alpha,Y,U}(\lambda).
\end{align*}
\end{proposition}
\begin{proof}
Since $Y_U\cap V=\emptyset$,
we have by Lemma \ref{lemma_quadratic-form} and \ref{lemma_DNdef}
\begin{align*}
 Q(q^D_{\alpha,Y,U})\oplus
 Q(q^D_{V})
\subset
 Q(q_{\alpha_U,Y_U})
\subset
 Q(q^N_{\alpha,Y,U})\oplus
 Q(q^N_{V}),
\end{align*}
and
\begin{align*}
& q^D_{\alpha,Y,U}[u]
+
 q^D_{V}[v]
=
 q_{\alpha_U,Y_U}[u\oplus v]
\quad
(u\in Q(q^D_{\alpha,Y,U}),
v\in Q(q^D_{V})),\\
& 
 q_{\alpha_U,Y_U}[u]
=
q^N_{\alpha,Y,U}[u|_U]
+
q^N_{V}[u|_V]
\quad
(u\in Q(q_{\alpha_U,Y_U})).
\end{align*}
Since the operators $-\Delta^\sharp_{V}$ ($\sharp=D,N$)
are non-negative,
we have $N^\sharp_{V}(\lambda)=0$ 
for $\lambda<0$.
Thus the conclusion follows from the statement (1) of Proposition
\ref{proposition_comparing} with $T=Id$.
\end{proof}

\subsection{Existence of IDS}
\label{subsection_eIDS}
In this subsection,
we shall prove the almost sure existence of IDS,
defined in (\ref{01_03}),
by the functional analytic method
using the Dirichlet--Neumann bracketing (cf.\ \cite{Ca-La} or 
\cite[section 2.3]{Ki-Me}).
In the present paper $Y_\omega$ is assumed to be the
Poisson configuration,
but here we prove our results under more general assumptions
described below,
for the future development.

Let $d=1,2$ or $3$.
Let $(\Omega_1,P_{\Omega_1})$ be a probability space,
and $\xi_{\omega_1}$ ($\omega_1\in \Omega_1$) be a 
non-negative integer-valued random measure on 
$(\mathbf{R}^d,\mathcal{B}(\mathbf{R}^d))$ 
($\mathcal{B}(\mathbf{R}^d)$ is the Borel $\sigma$-algebra in $\mathbf{R}^d$)
satisfying the following conditions.
\begin{enumerate}

 \item[(I1)]
For any Borel measurable set $S$ in $\mathbf{R}^d$,
the function $\omega_1\mapsto \xi_{\omega_1}(S)$ is a measurable 
map from $\Omega_1$ to $[0,\infty]$.
Moreover, if $S$ is bounded, then
$\xi_{\omega_1}(S)$ is finite almost surely.

 \item[(I2)]
The measure $\xi_{\omega_1}$ is simple, almost surely.
\end{enumerate}
Here, we say a non-negative integer-valued measure $\xi$ on $\mathbf{R}^d$ is 
\textit{simple} if $\xi(\{x\})\leq 1$ for every $x\in \mathbf{R}^d$
(cf.\ \cite{La-Pe}).
If $\xi_{\omega_1}$ satisfies (I1),
then we call $\xi_{\omega_1}$ a \textit{boundedly finite point process},
and if $\xi_{\omega_1}$ additionally satisfies (I2),
then we say that the process $\xi_{\omega_1}$ is \textit{simple}
(cf.\ \cite{La-Pe, Da-Ve1, Kal}).
A point process $\xi_{\omega_1}$ satisfying (I1) and (I2) 
is identified with the support 
$Y_{\omega_1}$ of the measure $\xi_{\omega_1}$.
The set $Y_{\omega_1}$ is a locally finite subset of $\mathbf{R}^d$
and $\xi_{\omega_1}(S)=\#(Y_{\omega_1}\cap S)$,
almost surely.
In the sequel, 
we also call $Y_{\omega_1}$ a point process 
by abuse of terminology,
and describe the assumptions in terms of $Y_{\omega_1}$.

Additionally, we assume the following.
Let $X$ be the measure space of all locally finite subsets $Y$ of 
$\mathbf{R}^d$,
equipped with the $\sigma$-algebra generated by the maps
$\{Y\mapsto\#(Y\cap S)\}_S$, where $S$ ranges over 
all bounded Borel subsets of $\mathbf{R}^d$.
We use the notation
\begin{align*}
 A+x :=\{Y+x\in X;Y\in A\},\quad
Y+x:=\{y+x\in \mathbf{R}^d;y\in Y\},
\end{align*}
for $A\in X$ and $x\in \mathbf{R}^d$.
The terminologies used below are taken from the standard textbooks
(cf.\ \cite{La-Pe,Da-Ve1,Kal}).
\begin{enumerate}
 \item[(I3)]
The intensity measure 
$\mu(S):=\mathbf{E}_{\Omega_1}[\#(Y_{\omega_1}\cap S)]$
is \textit{boundedly finite},
that is, $\mu(S)<\infty$ for any bounded measurable set $S$.

 \item[(I4)] The process $Y_{\omega_1}$ is $\mathbf{Z}^d$-\textit{stationary}, that is, there exist measure preserving transformations 
$\{T_n\}_{n\in \mathbf{Z}^d}$ on $\Omega_1$ such that
\begin{align}
\label{06_46}
Y_{\omega_1}+n = Y_{T_n \omega_1}
\end{align}
almost surely, for every $n\in \mathbf{Z}^d$.

 \item[(I5)] The process $Y_{\omega_1}$ is $\mathbf{Z}^d$-ergodic,
that is, if a measurable subset $A$ of $X$ satisfies 
\begin{align*}
 A+n =A
\end{align*}
for every $n\in \mathbf{Z}^d$, then $\mathbf{P}_{\Omega_1}(A)=0$ or $1$.

 \item[(I6)] 
$\#(Y_{\omega_1}\cap(\partial Q_1+n))=0$ 
for every $n\in \mathbf{Z}^d$ almost surely,
where $Q_1=(0,1)^d$.
\end{enumerate}
We impose the assumption (I6) in order to avoid the difficulty
caused by $Y_{\omega_1}\cap \partial Q_L\not=\emptyset$
($Q_L=(0,L)^d$),
when we define the operator $-\Delta^\#_{\alpha,Y_{\omega_1},Q_L}$.

Let us consider the case where $Y_{\omega_1}$ is the Poisson point process 
on $\mathbf{R}^d$ with intensity measure $\rho dx$,
where $\rho$ is a positive constant.
In this case, (I1)-(I4) and (I6) are satisfied
by definition.
The ergodicity assumption (I5) follows from 
Kolmogorov's 0-1 law,
since the random variables 
$\{\#(Y_{\omega_1}\cap(E+n))\}_{n\in \mathbf{Z}^d}$ 
are independent for any Borel measurable set $E\subset Q_1$
(cf.\ \cite[Chapter 20]{Kl}).

Next, we introduce the randomness of the coefficients 
$\alpha$.
We call a simple point process $Y_{\omega_1}$ on $\mathbf{R}^d$ is \textit{proper},
if there exists a random variable $n_{\omega_1}$ 
with $n_{\omega_1}\in \{0,1,2,...\}\cup\{\infty\}$,
and $\mathbf{R}^d$-valued random variables
$y_{\omega_1,j}$ ($j=1,2,\ldots$) such that
\begin{align}
\label{06_461}
Y_{\omega_1}=\bigcup_{j=1}^{n_{\omega_1}}\{y_{\omega_1,j}\},
\quad
y_{\omega_1,j}\not=y_{\omega_1,j'} \quad (j\not=j')
\end{align}
almost surely,
where $Y_{\omega_1}$ is the empty set if $n_{\omega_1}=0$.
It is known that every point process on $\mathbf{R}^d$ satisfying (I1) and (I2) is
proper (cf.\ \cite[Proposition 6.3]{La-Pe}).
Moreover, if the process 
$Y_{\omega}$ 
satisfies (I4) and is non-trivial 
(i.e., $\mathbf{E}[\#(Y_{\omega_1}\cap Q_1)]>0$),
then $n_{\omega_1}=\infty$ almost surely.
We use the expression (\ref{06_461}) 
with $n_{\omega_1}=\infty$ in the sequel.

\begin{enumerate}
 \item[(I7)] The single site measure $\nu$ 
is a probability measure on the measure space $(\mathbf{R},\mathcal{B}(R))$.
The probability space $\Omega_2$ is given by the infinite product space
\begin{align*}
 \Omega_2:=\prod_{j\in \mathbf{Z}^d} \Omega_{2,j},\quad
\Omega_{2,j}:=(\mathbf{R},\mathcal{B}(\mathbf{R}),\nu).
\end{align*}
\end{enumerate}

We put $\Omega=\Omega_1\times \Omega_2$.
Notice that an element $\omega_2$ of $\Omega_2$ 
is written as
\begin{align*}
 \omega_2=(\omega_{2,j})_{j=1}^\infty,\quad \omega_{2,j}\in \mathbf{R}.
\end{align*}
For $\omega=(\omega_1,\omega_2)\in \Omega$,
we put 
\begin{align*}
Y_\omega:=Y_{\omega_1},\ 
\alpha_\omega:=\{\omega_{2,j}\}.
\end{align*}
This is the meaning of the independence of 
the set $Y_{\omega}$ and the values $\alpha_{\omega}$,
stated in the introduction.

Now we can prove the existence of IDS as follows.
\begin{theorem}
\label{theorem_existence_IDS}
Let $d=1,2$ or $3$.
Let $(Y_\omega,\alpha_\omega)$ satisfy
assumptions (I1)-(I7) above.
For $L\in \mathbf{N}$,
put $Q_L=(0,L)^d$.
Then,
for any $\lambda\in \mathbf{R}$,
the following equalities hold almost surely.
\begin{align}
&\lim_{L\to\infty}\frac{N^N_{\alpha_\omega,Y_\omega,Q_L}(\lambda)}{|Q_L|}
=
\lim_{L\to\infty}\frac{\mathbf{E}[N^N_{\alpha_\omega,Y_\omega,Q_L}(\lambda)]}{|Q_L|},
\label{06_47}\\
&\lim_{L\to\infty}\frac{N^D_{\alpha_\omega,Y_\omega,Q_L}(\lambda)}{|Q_L|}
=
\lim_{L\to\infty} \frac{\mathbf{E}[N^D_{\alpha_\omega,Y_\omega,Q_L}(\lambda)]}{|Q_L|}.
\label{06_48}
\end{align}
We denote the right hand side of of (\ref{06_47}) (resp.\ (\ref{06_48}))
by $N^N(\lambda)$ (resp.\ $N^D(\lambda)$), respectively.
Then, we have for every $\lambda\in \mathbf{R}$
\begin{align}
\label{06_49}
 N^D(\lambda) \leq N^N(\lambda).
\end{align}
\end{theorem}
\noindent{\bf Remark.} 
1. The right hand sides of (\ref{06_47}) and (\ref{06_48}) are independent of $\omega$,
so the left hand sides are also independent of $\omega$,
almost surely.
Moreover, 
$\lim_{L\to\infty}$ in the right hand side of (\ref{06_47}) (resp.\  (\ref{06_48}))
can be replaced by $\inf_L$ (resp.\ $\sup_L$),
by the subadditivity (resp.\ superadditivity)
in the sense of the statement (2) of Proposition \ref{proposition_DNbracketing}.

2. By (I7), the set $\Omega_0=\{\omega\in \Omega\,;\, (\partial Q_1+n)\cap Y_\omega=\emptyset$ for every $n\in \mathbf{Z}^d\}$ has probability $1$.
In the sequel, we replace $\Omega$ by $\Omega_0$,
so we always assume that $(\partial Q_1+n)\cap Y_\omega=\emptyset$
for every $n\in \mathbf{Z}^d$.
So $-\Delta^\sharp_{\alpha_\omega,Y_\omega,Q_L}$ is well-defined.

\begin{proof}
Let $m_1,m_2\in \mathbf{Z}^d$ such that $m_{1,j}<m_{2,j}$ for any $j=1,...,d$,
and let $R_{m_1,m_2}=\prod_{j=1}^d(m_{1,j},m_{2,j})$ be an open rectangle.
Let us check the measurability of 
the $j$-th eigenvalue
$\lambda_j(-\Delta^\sharp_{\alpha_\omega,Y_\omega,R_{m_1,m_2}})$
of $-\Delta^\sharp_{\alpha_\omega,Y_\omega,R_{m_1,m_2}}$ 
($\sharp=D,N$) 
with respect to $\omega$.
By Proposition \ref{proposition_min-max},
\begin{align*}
&\  \lambda_j(-\Delta^\sharp_{\alpha_\omega,Y_\omega,R_{m_1,m_2}})\\
=&\ \sup_{v_1,\ldots,v_{j-1}\in S}
\left(
\inf_{u\in Q(q^\sharp_{\alpha_\omega,Y_\omega,R_{m_1,m_2}})\cap 
\langle v_1,\ldots,v_{j-1}\rangle^\bot
,\,
\|u\|_{L^2(R_{m_1,m_2})}=1
}q^\sharp_{\alpha_\omega,Y_\omega,R_{m_1,m_2}}[u]
\right),
\end{align*}
where $S$ is a countable dense subset of $L^2(R_{m_1,m_2})$.
If we fix $k=\#(Y_\omega\cap R_{m_1,m_2})$,
$d_0$, $u_0$, and $c_y$ in Lemma \ref{lemma_DNdef},
then the quadratic form $q^\sharp_{\alpha,Y,R_{m_1,m_2}}[u]$ 
is continuous with respect to $y\in Y\cap R_{m_1,m_2}$ and  $\alpha_y$.
Then, the function 
$q^\sharp_{\alpha_\omega,Y_\omega,R_{m_1,m_2}}[u]$ 
is measurable with respect to $\omega\in \Omega$,
and we can prove the measurability of $\lambda_j(-\Delta^\sharp_{\alpha_\omega,Y_\omega,R_{m_1,m_2}})$
by replacing the range of the infimum by some countable dense subset.
Thus the function $N^\sharp_{\alpha_\omega,Y_\omega,R_{m_1,m_2}}(\lambda)$
is also measurable.

By the assumption (I4) and (\ref{06_461}),
there exists a number $j_{\omega_1,n,j}$ such that
\begin{align*}
 y_{\omega_1,j_{\omega_1,n,j}}+n
=
y_{T_n\omega_1,j}.
\end{align*}
Then the map from $\Omega$ to $\Omega$ defined by
\begin{align*}
&
\Omega=\Omega_1\times \Omega_2
\ni \omega=(\omega_1,\omega_2)
\mapsto
\widetilde{\omega}=(T_n \omega_1,S_n \omega_2)\in \Omega,\\
&S_n\omega_2 :=(\omega_{2,j_{\omega_1,n,j}})_{j=1}^\infty,
\end{align*}
is a measure preserving map, and satisfies
\begin{align*}
 N^\sharp_{\alpha_{\widetilde{\omega}},Y_{\widetilde{\omega}},R_{m_1+n,m_2+n}}
(\lambda)
=
 N^\sharp_{\alpha_{\omega},Y_{\omega},R_{m_1,m_2} }(\lambda)
\end{align*}
almost surely, for every $\lambda \in \mathbf{R}$
and $m_1,m_2\in \mathbf{Z}^d$ with $m_{1,j}<m_{2,j}$ ($j=1,\ldots,d$).
Moreover, 
the family
$\{N^N_{\alpha_\omega,Y_\omega,R_{m_1,m_2}}(\lambda)\}_{m_1,m_2}$ 
(resp.\ $\{N^D_{\alpha_\omega,Y_\omega,R_{m_1,m_2}}(\lambda)\}_{m_1,m_2}$ 
) satisfies the subadditivity (resp.\ superadditivity),
in the sense of the statement (2) of Proposition \ref{proposition_DNbracketing}.
Further, by the statement (3) of
Proposition \ref{proposition_DNbracketing},
we have
\begin{align}
\label{06_50}
 N^\sharp_{\alpha_\omega,Y_\omega,R_{m_1,m_2}}(\lambda)
\leq  N^\sharp_{R_{m_1,m_2}}(\lambda)+\#(Y_\omega\cap R_{m_1,m_2}).
\end{align}
By Proposition \ref{proposition_DNbracketing},
we have
\begin{align}
\label{06_51}
 N^D_{R_{m_1,m_2}}(\lambda)
\leq
 N^N_{R_{m_1,m_2}}(\lambda)
\leq |R_{m_1,m_2}|
 N^N_{Q_1}(\lambda).
\end{align}
By assumptions (I3) and (I4), we have
\begin{align}
\label{06_52}
 \mathbf{E}[\#(Y_\omega\cap R_{m_1,m_2})]
\leq |R_{m_1,m_2}|\mu(Q_1).
\end{align}
By the statement (1) of Proposition \ref{proposition_DNbracketing},
(\ref{06_50}), (\ref{06_51}), and (\ref{06_52}), we have a bound
\begin{align*}
 0\leq&\  
\sup_{m_1,m_2} \frac{\mathbf{E}[N^D_{\alpha_\omega,Y_\omega,R_{m_1,m_2}}(\lambda)]}{|R_{m_1,m_2}|}
\leq 
\inf_{m_1,m_2} \frac{\mathbf{E}[N^N_{\alpha_\omega,Y_\omega,R_{m_1,m_2}}(\lambda)]}{|R_{m_1,m_2}|}\\
\leq&\ 
 N^N_{Q_1}(\lambda)
+
\mu(Q_1).
\end{align*}
Thus the almost sure existence of the limits in the left hand sides
of (\ref{06_47}) and (\ref{06_48})
are guaranteed by
the superadditive ergodic theorem
(cf.\ \cite{Ak-Kr} or \cite{Ca-La}).
Then the equalities 
(\ref{06_47}) and (\ref{06_48}) follow from 
the ergodicity assumption (I5),
and (\ref{06_49}) follows from 
the statement (1) of Proposition \ref{proposition_DNbracketing}.
\end{proof}

\subsection{Independence of IDS from the boundary conditions}
In this subsection,
we prove that IDS $N^\sharp(\lambda)$
is independent of the boundary conditions $\sharp=D,N$,
except at most countable $\lambda$ on $\mathbf{R}$.
The following proof is based on \cite{Do-Iw-Mi},
in which the same statement
for the magnetic Schr\"odinger operators 
(without point interactions) is proved.

\begin{theorem}
 \label{theorem_independence_IDS}
Let $d=1,2$ or $3$.
Suppose that the assumptions (I1)-(I7) hold.
Let $N^\sharp(\lambda)$ ($\sharp=D,N$) be the functions defined in 
Theorem \ref{theorem_existence_IDS}.
Then, for any $\lambda\in \mathbf{R}$ and any $\epsilon>0$,
we have
\begin{align}
 \label{06_53}
N^D(\lambda)\leq N^N(\lambda) \leq N^D(\lambda+\epsilon).
\end{align}
In particular, 
we have for every $\lambda\in \mathbf{R}$
\begin{align}
\label{06_54}
N^D(\lambda+0)= N^N(\lambda+0),
\end{align}
where $f(\lambda+0):=\lim_{\epsilon\to+0}f(\lambda +\epsilon)$.
Moreover, 
if either $N^D(\lambda)$ or $N^N(\lambda)$ is continuous at $\lambda$,
then 
\begin{align}
\label{06_55}
N^D(\lambda)= N^N(\lambda)=N^D(\lambda+0)=N^N(\lambda+0).
\end{align}
\end{theorem}
\noindent\textbf{Remark.} 
By definition, the function $N^\sharp(\lambda)$ 
is monotone non-decreasing on $\mathbf{R}$.
The set of points of discontinuity is at most countable,
and (\ref{06_55}) holds except this at most countable set of points 
of discontinuity.

\begin{proof}
First we assume that (\ref{06_53}) holds, and prove
that (\ref{06_54}) and (\ref{06_55}) hold.
Notice that 
the limit from the right
$N^\sharp(\lambda+0)$ exists for every $\lambda\in \mathbf{R}$,
since the function $N^\sharp(\lambda)$ is monotone non-decreasing.
By (\ref{06_53}), we have 
\begin{align}
\label{06_56}
  N^D(\lambda+\epsilon)\leq N^N(\lambda+\epsilon)\leq N^D(\lambda+2\epsilon)
\end{align}
for every $\lambda\in \mathbf{R}$ and $\epsilon>0$.
Taking the limit $\epsilon\to +0$ in (\ref{06_56}), we have
(\ref{06_54}).
Moreover, we have by (\ref{06_53}) 
\begin{align}
\label{06_57}
 N^N(\lambda-\epsilon)\leq N^D(\lambda)\leq N^N(\lambda)\leq N^D(\lambda+\epsilon)
\end{align}
for every $\lambda\in\mathbf{R}$ and $\epsilon>0$.
If either $N^D(\lambda)$ or $N^N(\lambda)$ is continuous at $\lambda$,
then (\ref{06_57}) and (\ref{06_54}) imply (\ref{06_55}).

Next we shall prove (\ref{06_53}).
We already know that the first inequality
in (\ref{06_53}) holds, by (\ref{06_49}).
We shall prove the second inequality in (\ref{06_53}).

For a positive integer $L$,
we put $\ell=[\sqrt{L}]$,
where $[x]$ is the maximal integer which is less than or equal to $x$.
 Let $Q_L=(0,L)^d$, $Q_{L,1}=(\ell,L-\ell)^d$, and 
$Q_{L,2}=(2\ell,L-2\ell)^d$.
Put
 $Y_{\omega,2}=Y_\omega\cap Q_{L,2}$
and
$\alpha_{\omega,2}=\alpha_\omega|_{Y_{\omega,2}}$.
By the statement (3) of Proposition \ref{proposition_DNbracketing},
we have
\begin{align*}
 N^N_{\alpha_\omega,Y_\omega,Q_L}(\lambda)
\leq 
 N^N_{\alpha_{\omega,2},Y_{\omega,2},Q_L}(\lambda)
+
\#(Y_\omega\cap(Q_L\setminus Q_{L,2})).
\end{align*}
Taking expectation, we have by (I3), (I4) and (I6)
\begin{align}
\mathbf{E}[ N^N_{\alpha_\omega,Y_\omega,Q_L}(\lambda)]
\leq &\ 
\mathbf{E}[ N^N_{\alpha_{\omega,2},Y_{\omega,2},Q_L}(\lambda)]
+
\mathbf{E}[\#(Y_\omega\cap(Q_L\setminus Q_{L,2}))]\notag\\
\leq&\ 
\mathbf{E}[ N^N_{\alpha_{\omega,2},Y_{\omega,2},Q_L}(\lambda)]
+|Q_L\setminus Q_{L,2}|\mu(Q_1).
\label{06_58}
\end{align}

We can construct two functions $\chi_1,\chi_2\in C^\infty(Q_L)$
satisfying the following conditions
(cf.\ \cite[Proposition 8.1]{Do-Iw-Mi}):
\begin{align}
& \chi_1^2+\chi_2^2=1,
\label{06_59}
\\
&
\begin{cases}
\chi_1(x) =0 & (x\in Q_{L,1}^c) ,\\
0\leq \chi_1(x)\leq 1 & (x\in Q_{L,1}\setminus Q_{L,2}),\\
\chi_1(x)=1 & (x\in Q_{L,2}),
\end{cases}
\begin{cases}
\chi_2 (x)=1 & (x\in Q_{L,1}^c),\\
0\leq \chi_2(x)\leq 1 & (x\in Q_{L,1}\setminus Q_{L,2}),\\
\chi_2(x) =0 & (x\in Q_{L,2}),
\end{cases}
\notag
\\
&
\|\nabla \chi_1\|_\infty^2
+\|\nabla \chi_2\|_\infty^2
\leq \frac{C}{\ell^2},
\label{06_60}
\end{align}
where $C$ is a positive constant independent of $L$.
By a simple calculation (cf.\ \cite[Theorem 3.2]{Cy-Fr-Ki-Si}
or \cite[(6.14)]{Do-Iw-Mi},
we have
\begin{align*}
q^N_{\alpha_{\omega,2},Y_{\omega,2},Q_L}[u]
=&\ 
q^D_{\alpha_{\omega,2},Y_{\omega,2},Q_L}[\chi_1 u]
-\|(\nabla \chi_1)u\|_{L^2(Q_L)}^2\notag\\
&\ 
+
q^N_{Q_L\setminus \overline{Q_{L,2}}}[\chi_2 u]
-\|(\nabla \chi_2)u\|_{L^2(Q_L\setminus \overline{Q_{L,2}})}^2,
\end{align*}
for every $u\in Q(q^N_{\alpha_{\omega,2},Y_{\omega,2},Q_L})$.
Then, (\ref{06_59}) and (\ref{06_60}) imply
\begin{align}
 q^N_{\alpha_{\omega,2},Y_{\omega,2},Q_L}[u]
\geq&\  (
q^D_{\alpha_{\omega,2},Y_{\omega,2},Q_L}
-C\ell^{-2})[\chi_1 u]\notag\\
&+
(q^N_{Q_L\setminus \overline{Q_{L,2}}}
-C\ell^{-2})[\chi_2 u].
\label{06_61}
\end{align}
The inequality (\ref{06_61}) and 
the statement (1) of Proposition \ref{proposition_comparing} with
\begin{align*}
 T:L^2(Q_L)\ni u\mapsto
\chi_1 u\oplus \chi_2 u
\in L^2(Q_L)\oplus
L^2(Q_L\setminus \overline{Q_{L,2}})
\end{align*}
imply
\begin{align}
&\  N^N_{\alpha_{\omega,2},Y_{\omega,2},Q_L}(\lambda)\notag\\
\leq &\ 
 N^D_{\alpha_{\omega,2},Y_{\omega,2},Q_L}(\lambda+C\ell^{-2})
+
 N^N_{Q_L\setminus\overline{Q_{L,2}}}(\lambda+C\ell^{-2})
\notag\\
\leq &\ 
 N^D_{\alpha_{\omega,2},Y_{\omega,2},Q_L}(\lambda+C\ell^{-2})
+
|Q_L\setminus\overline{Q_{L,2}}|
 N^N_{Q_1}(\lambda+C\ell^{-2}).
\label{06_62}
\end{align}
For any $\epsilon>0$, take $L$ so large that 
$C\ell^{-2}=C[\sqrt{L}]^{-2}<\epsilon$.
Then,
(\ref{06_58}) and (\ref{06_62}) imply
\begin{align*}
& \mathbf{E}[ N^N_{\alpha_\omega,Y_\omega,Q_L}(\lambda)]\\
\leq&\ 
\mathbf{E}[ N^D_{\alpha_\omega,Y_\omega,Q_L}(\lambda+\epsilon)]
+
|Q_L\setminus\overline{Q_{L,2}}|(\mu(Q_1)+ N^N_{Q_1}(\lambda+\epsilon)).
\end{align*}
Dividing the both sides by $|Q_L|$ and taking the limit $L\to \infty$,
we have (\ref{06_53}).
\end{proof}

\section{Appendix}
\subsection{Poisson point process}
\label{section_pp}

The definition of the Poisson point process is as follows (see e.g., \cite{Re}).
\begin{definition}
\label{definition_poisson}
Let $S$ be a Borel measurable set in $\mathbf{R}^d$ ($d\geq 1$).
Let $\mu_\omega$ be a random measure on $S$
dependent on $\omega\in \Omega$ for some probability space $\Omega$.
For a positive constant $\rho$,
we say $\mu_\omega$ is the Poisson point process on $S$
with intensity measure $\rho dx$ 
if the following conditions hold.
\begin{enumerate}
 \item For every Borel measurable set $E$ in $S$ 
with the Lebesgue measure $|E|<\infty$, 
$\mu_\omega(E)$ is an integer-valued random variable on $\Omega$ and
\begin{align*}
 \mathbf{P}(\mu_\omega(E)=k)=\frac{(\rho|E|)^k}{k!}e^{-\rho|E|}
\end{align*}
for every non-negative integer $k$.

 \item 
For any disjoint Borel measurable sets $E_1,\ldots,E_n$ in $S$
with finite Lebesgue measure,
the random variables $\{\mu_\omega(E_j)\}_{j=1}^n$
are independent.
\end{enumerate}
We call the support $Y_\omega$
of the Poisson point process measure $\mu_\omega$
the \textit{Poisson configuration on $S$}.
\end{definition}

We review a construction of the Poisson configuration on $S$
when $|S|$ is finite
(for the proof, see \cite[Theorem 1.2.1]{Re}).

\begin{proposition}
\label{prop_poisson_construction}
Let $S$ be a Borel measurable set in $\mathbf{R}^d$ ($d \geq 1$)
with Lebesgue measure $|S|<\infty$.
Let $\rho>0$ be a constant.
Let $n_\omega$ be a random variable 
obeying the Poisson distribution with parameter $\rho|S|$.
Let $\{y_{j,\omega}\}_{j=1}^\infty$ be independent $\mathbf{R}^d$-valued 
random variables obeying the uniform distribution on $S$,
which are independent of $n_\omega$.
Define the random set $Y_\omega$ by
\begin{align*}
 Y_\omega = 
\begin{cases}
\{y_{j,\omega}\}_{j=1}^{n_\omega} & (n_\omega\geq 1), \\
\emptyset & (n_\omega=0).
\end{cases}
\end{align*}
Then, $Y_\omega$ is a Poisson configuration on $S$ with intensity
$\rho dx$.
\end{proposition}
The Poisson configuration on $\mathbf{R}^d$ can be constructed
by dividing $\mathbf{R}^d$ into disjoint cubes $Q_j$,
constructing independent Poisson configurations $Y_j$  on $Q_j$,
and taking the union of $Y_j$.

\subsection{Palm distribution}
\label{subsection_palm}
We quote some basic definitions
about the Palm distribution of a random measure,
from 
Daley--Vere-Jones \cite{Da-Ve2}.
We say a Borel measure $\xi$ on $\mathbf{R}^d$ is 
\textit{boundedly finite}
if $\xi(U)<\infty$ for any bounded Borel set $U$ in $\mathbf{R}^d$.
Let $M$
be the space of boundedly finite measures $\xi$ on $\mathbf{R}^d$,
equipped with the topology of weak convergence
and the corresponding Borel $\sigma$-algebra.
Let $\xi$ be a boundedly finite \textit{random measure} on $\mathbf{R}^d$,
that is,
$\Omega\ni \omega\mapsto \xi_\omega \in M$
is a measurable map
from some probability space $\Omega$
to $M$.
We assume that the \textit{first moment measure}
(or the \textit{intensity measure})
\begin{align*}
 \mu(U):=\mathbf{E}[\xi_\omega(U)]
\end{align*}
is also boundedly finite,
where $\mathbf{E}[f(\omega)]=\int_\Omega f(\omega) dP(\omega)$
is the expectation of a measurable function $f$
with respect to the probability measure $P$ on $\Omega$.
Then, it is known that
there exists
a family of probability measures $\{P_x\}_{x\in \mathbf{R}^d}$
on $\Omega$ such that 
\begin{align*}
\mathbf{E}\left[
\int_{\mathbf{R}^d}g(x,\xi_\omega)d\xi_\omega(x)
\right]
=
\int_{\mathbf{R}^d} 
\mathbf{E}_x
\left[
g(x,\xi_\omega)
\right]
d\mu(x)
\end{align*}
for any non-negative measurable
function $g(x,\xi)$
on $\mathbf{R}^d\times M$,
where
\begin{align*}
 \mathbf{E}_x[f(\omega)]=\int_\Omega f(\omega)dP_x(\omega)
\end{align*}
(cf.\ \cite[Proposition 13.1.IV]{Da-Ve2}).
The measure
$P_x$ is called a 
\textit{local Palm distribution for $\xi$},
and the family 
$\{P_x\}_{x\in \mathbf{R}^d}$
is called the \textit{Palm kernel associated with $\xi$}.
The local Palm distribution $P_x$ is,
roughly speaking,
the conditional distribution of 
the random measure $\xi$ 
under the condition $x\in \supp \xi$.

In the case that $\xi$ is 
the homogeneous Poisson point process on $\mathbf{R}^d$
with intensity measure $\rho dx$,
it is known that
$P_x$ and the distribution of $\delta_x + \xi$ on $\Omega$ 
are equal (cf. \cite[Proposition 13.1.VII]{Da-Ve2}),
where $\delta_x$ is the Dirac measure supported on $x$.
Thus we have the following corollary.
\begin{corollary}
 \label{corollaly_palm}
Let $\xi$ be the Poisson point process on $\mathbf{R}^d$
with intensity measure $\rho dx$, where $\rho$ is a positive constant.
Then, for any non-negative measurable function $g(x,\xi)$ 
on $\mathbf{R}^d\times M$,
we have
\begin{align}
\label{07_01}
 \mathbf{E}\left[\sum_{x\in \supp \xi_\omega}g(x,\xi_\omega)\right]
=
\int_{\mathbf{R}^d}
\mathbf{E}\left[g(x,\delta_x+\xi_\omega)\right]
\rho dx.
\end{align}
\end{corollary}

By using (\ref{07_01}),
we give another proof of Proposition \ref{prop_cluster}.
For $n=1,2,3,\ldots$, put
\begin{align*}
 g(x,\xi_\omega)
=
\begin{cases}
 1 & (x\in Q_L,\ n_x(R)=n),\\
 0 & (\mbox{otherwise}),
\end{cases}
\end{align*}
where $n_x(R)=\xi_\omega(B_x(R))=\#(\supp\xi_\omega \cap B_x(R))$, 
just as (\ref{03_01}) with $Y_\omega=\supp \xi_\omega$.
Then, the left hand side of (\ref{07_01}) divided by $|Q_L|$ 
coincides with the left hand side of (\ref{03_02}).
Moreover, 
since $(\delta_x + \xi_\omega)(B_x(R))= 1 + \xi_\omega(B_x(R))$,
the right hand side of (\ref{07_01}) is calculated as follows.
\begin{align*}
 \int_{Q_L}
\mathbf{P}(\xi_\omega(B_x(R))=n-1)\rho dx
=
\rho |Q_L|\frac{(\rho|B_0(R)|)^{n-1}}{(n-1)!}e^{-\rho |B_0(R)|}.
\end{align*}
Dividing the last expression by $|Q_L|$,
we have the right hand side of (\ref{03_02}).

\subsection{Basic formulas for the point interactions}
We review some basic formulas for the spectrum of the
Hamiltonian with point interactions on $\mathbf{R}^3$.

First we introduce an auxiliary matrix.
For a finite set $Y=(y_j)_{j=1}^n$ ($n=\#Y$) in $\mathbf{R}^3$,
a sequence of real numbers $\alpha=(\alpha_j)_{j=1}^n$
(we use the abbreviation $\alpha_{y_j}=\alpha_j$),
and $s>0$,
define an $n\times n$ matrix $A=A(s)=(a_{jk}(s))$ by
\begin{align}
\label{07_02}
 a_{jk}(s)
=
\begin{cases}
\displaystyle
\alpha_j + \frac{s}{4\pi} & (j=k),\\[2mm]
\displaystyle
-\frac{e^{-s|y_j-y_k|}}{4\pi |y_j-y_k|} & (j\not=k).
\end{cases}
\end{align}
When we need to specify $\alpha$ and $Y$,
we write $A(s)=A_{\alpha,Y}(s)$.
We also recall the definition of the eigenvalue counting function
\begin{align*}
  N_{\alpha,Y}(\lambda):=&\ 
\#\{\mu \leq \lambda \,;\,
\mu \mbox{ is an eigenvalue of }- \Delta_{\alpha,Y} \}
\end{align*}
for $\lambda<0$,
where eigenvalues are counted with multiplicity.

\begin{proposition}
\label{prop_al}
Let $d=3$.
Let 
$Y=(y_j)_{j=1}^n$ ($n=\#Y$) be a finite set,
and $\alpha=(\alpha_j)_{j=1}^n$ be a sequence of real numbers.
Let $A(s)$ be the matrix given by (\ref{07_02}).
Let $\mu_j(s)$ be the $j$-th smallest eigenvalue of $A(s)$,
counted with multiplicity.

\begin{enumerate}
 \item[(1)]
For $s>0$,
$\lambda =-s^2$ is an eigenvalue of $-\Delta_{\alpha,Y}$ if 
and only if $\det A(s)=0$.
Moreover, the multiplicity of the eigenvalue $-s^2$ is
equal to the dimension of $\Ker A(s)$.

 \item[(2)]
The function $\mu_j(s)$ is continuous on $(0,\infty)$,
and real-analytic on $(0,\infty)\setminus S$,
where the set $S$ does not have an accumulation point in 
$(0,\infty)$.
Moreover,
$\mu_j(s)$ is strictly monotone increasing with respect to $s$,
and 
$\displaystyle \lim_{s\to \infty} \mu_j(s)=\infty$.

 \item[(3)]
We have $0\leq N_{\alpha,Y}(-s^2)\leq n$ for every $s>0$.
For $m=0,\ldots, n$,
the condition $N_{\alpha,Y}(-s^2)=m$ 
is equivalent to the condition
\begin{align*}
 \mu_1(s) \leq \cdots \leq \mu_m(s) \leq 0 
< \mu_{m+1}(s)\leq \cdots \leq \mu_n(s).
\end{align*}
\end{enumerate}
\end{proposition}
\noindent
\begin{proof}
For the proof of (1), (2), and the first statement 
of (3),
see \cite[Theorem II-1.1.4]{Al-Ge-Ho-Ho}.
The equivalence in (3) follows from the fact
that $\mu_j(s')=0$ has a unique solution with $s'\geq s$
if and only if $\mu_j(s)\leq 0$.
\end{proof}

Next we summarize properties of the counting function $N_{\alpha,Y}(\lambda)$
when we change the value $\alpha$ or the set $Y$.
\begin{proposition}
Let $d=3$.
\label{prop_al2} 
\begin{enumerate}
 \item[(1)]
Let $Y$ be a finite set in $\mathbf{R}^3$.
If two sequences 
$\alpha=(\alpha_y)_{y\in Y}$
and 
$\alpha'=(\alpha_y')_{y\in Y}$
satisfy
$\alpha_y \geq \alpha_y'$ for every $y\in Y$,
then
\begin{align*}
 N_{\alpha,Y}(\lambda)\leq N_{\alpha',Y}(\lambda)
\end{align*}
for every $\lambda<0$.

 \item[(2)]
Let $Y$ and $Y'$ be finite sets in $\mathbf{R}^3$ with $Y'\subset Y$.
Let $\alpha = (\alpha_y)_{y\in Y}$ 
be a real-valued sequence
 and $\alpha' = \alpha|_{Y'}$ 
(restriction of $\alpha$ on $Y'$).
Then,
\begin{align*}
 N_{\alpha', Y'}(\lambda) \leq  N_{\alpha, Y}(\lambda)
\leq N_{\alpha', Y'}(\lambda) + \#(Y\setminus Y')
\end{align*}
for every $\lambda<0 $.
\end{enumerate}

\end{proposition}
\begin{proof}
 (1) Let $n=\#Y$, and we use the notation
\begin{align*}
 Y =\{y_1,\ldots,y_n\},\quad
\alpha=\{\alpha_1,\ldots,\alpha_n\},\quad \alpha_j:=\alpha_{y_j}.
\end{align*}
We fix $Y$, $\alpha_2,\cdots,\alpha_n$, 
and regard $A=A_{\alpha,Y}(s)$ as an analytic matrix-valued function 
of $\alpha_1$.
This function is symmetric in the sense of
\cite[section II-6.1]{Ka},
and 
\cite[II-Theorem 6.1]{Ka} states that
the $j$-th smallest eigenvalue $\mu_j=\mu_j(s,\alpha_1)$ of $A$
is continuous in $\alpha_1 \in \mathbf{R}$,
and real-analytic in $\alpha_1 \in \mathbf{R}\setminus S$,
where $S$ is some exceptional set having no accumulation points in $\mathbf{R}$.
We can also take 
a normalized eigenvector $\phi_j=\phi_j(s,\alpha_1)$
corresponding to the eigenvalue $\mu_j$
such that $\phi_j(s,\alpha_1)$ is real-analytic in $\alpha_1\in \mathbf{R}\setminus S$.
Then, by the Feynman--Hellmann theorem and the explicit form of $A$,
we have
\begin{align*}
 \frac{\partial \mu_j}{\partial \alpha_1}
=
\left(\phi_j, \frac{\partial A}{\partial \alpha_1}\phi_j\right)_{\mathbf{C}^n}
=
|(\phi_j)_1|^2\geq 0,
\end{align*}
where $(\phi_j)_1$ is the first component of the vector $\phi_j$.
This implies
\begin{equation}
\label{07_03}
 \mu_j(s,\alpha_1)\geq  \mu_j(s,\alpha_1')
\end{equation}
if $\alpha_1\geq \alpha_1'$.

Let 
\begin{align}
\label{07_04}
\alpha=\{\alpha_1,\alpha_2,\ldots,\alpha_n\},\quad
\alpha'=\{\alpha_1',\alpha_2,\ldots,\alpha_n\},\quad
\alpha_1\geq \alpha_1'.
\end{align}
Assume $N_{\alpha,Y}(-s^2)=m$. 
By (3) of Proposition \ref{prop_al},
\begin{align*}
  \mu_1(s,\alpha_1) \leq \cdots \leq \mu_m(s,\alpha_1) \leq 0 
< \mu_{m+1}(s,\alpha_1)\leq \cdots \leq \mu_n(s,\alpha_1).
\end{align*}
By this inequality and (\ref{07_03}), we have
\begin{align*}
  \mu_1(s,\alpha_1') \leq \cdots \leq \mu_m(s,\alpha_1') \leq 0.
\end{align*}
By (3) of Proposition \ref{prop_al}, this inequality implies
\begin{align}
\label{07_05}
 N_{\alpha', Y}(-s^2)\geq m =  N_{\alpha, Y}(-s^2)
\end{align}
for $\alpha$ and $\alpha'$ satisfying (\ref{07_04}).
Then we can prove (\ref{07_05}) in general case
by using (\ref{07_05}) in the case (\ref{07_04}) repeatedly.

(2) This statement is a consequence of
Lemma \ref{lemma_quadratic-form}
and Proposition \ref{proposition_comparing} with $T=Id$.
\end{proof}

In \cite[Theorem II.1.1.4]{Al-Ge-Ho-Ho},
it is stated that $0$ is not an eigenvalue of 
the operator $-\Delta_{\alpha,Y}$,
if $Y$ is a finite set.
But this statement is incorrect.
According to `Errata and Addenda' in \cite[page 485]{Al-Ge-Ho-Ho},
the unpublished manuscript
`A remark on the existence of zero-energy bound states for 
point interaction Hamiltonians' by G.\ F.\  Dell'Antonio and G.\ Panati
gives an example of $(\alpha,Y)$ with $\#Y=3$ such that
$0$ is an eigenvalue of $-\Delta_{\alpha,Y}$.
In general, there exist pairs $(\alpha,Y)$ with 
$2\leq \# Y<\infty$, such that
$0$ is an eigenvalue of $-\Delta_{\alpha,Y}$.
This fact is proved in \cite[Theorem 1.1]{Co-Mi-Ya} when 
the space dimension $d=2$,
and also pointed out in the remark of \cite[Theorem 1.1]{Co-Mi-Ya} when $d=3$.
Here we give a simple proof of this fact when $d=3$,
for readers' use.

\begin{proposition}
 \label{proposition_zeroev}
Let $d=3$.
Let $Y=(y_j)_{j=1}^n$ be a finite subset of $\mathbf{R}^3$
with $\# Y =n<\infty$.
Let $\alpha=(\alpha_j)_{j=1}^n$ be a sequence of real numbers.
Let $A(0)$ be the $n\times n$ matrix $A(s)$ given in (\ref{07_02}) with $s=0$,
that is,
\begin{align*}
 A(0)=(a_{jk}(0)),\quad
a_{jk}(0)
=
\begin{cases}
 \alpha_j  & (j=k),\\
-\frac{1}{4\pi|y_j-y_k|} & (j\not=k).
\end{cases}
\end{align*}
If  $\det A(0)=0$  and $\Ker A(0)$ has a 
non-zero vector $c=(c_j)_{j=1}^n$ with
\begin{align}
\label{zero_01}
 \sum_{j=1}^n c_j=0, 
\end{align}
then
$0$ is an eigenvalue of $-\Delta_{\alpha,Y}$.
\end{proposition}
\begin{proof}
Suppose that $\det A(0)=0$ and $\Ker A(0)$ has a non-zero vector
$c=(c_j)_{j=1}^n$ with $\sum_{j=1}^n c_j=0$.
Let $u$ be the function given by
\begin{align*}
 u(x)=\sum_{j=1}^n
\frac{c_j}{4\pi|x-y_j|}.
\end{align*}
Since
\begin{align*}
 &u(x)=\frac{u_{y_j,0}}{|x-y_j|}+u_{y_j,1}+O(|x-y_j|)\quad (x\to y_j),\\
&u_{y_j,0}=\frac{c_j}{4\pi},\quad
u_{y_j,1}=\sum_{k\not=j}\frac{c_k}{4\pi|y_k-y_j|},
\end{align*}
the function $u$ satisfies 
$(BC)_{y_j}: -4\pi \alpha_j u_{y_j,0}+u_{y_j,1}=0$
for every $j=1,\ldots,n$,
if and only if $c\in \Ker A(0)$.
The function $u$ also satisfies 
the Laplace equation $-\Delta u=0$ on $\mathbf{R}^3\setminus Y$.
Moreover, since
\begin{align*}
\frac{c_j}{4\pi|x-y_j|}=\frac{c_j}{4\pi|x|}+O(|x|^{-2}) \quad(|x|\rightarrow \infty),
\end{align*}
the condition (\ref{zero_01}) implies that
the function $u$ satisfies $u(x)=O(|x|^{-2})$ as $|x|\rightarrow \infty$,
and $u\in L^2(\mathbf{R}^3)$.
Thus $u$ is an eigenfunction of $-\Delta_{\alpha,Y}$
with eigenvalue $0$.
\end{proof}
For example, let $n=2$,
$Y=\{y_1,y_2\}$ with $|y_1-y_2|=R>0$,
and $\alpha_1=\alpha_2=-1/(4\pi R)$.
Then the matrix
\begin{align*}
 A(0)
=-\frac{1}{4\pi R}
\begin{pmatrix}
 1 & 1\\ 1 & 1
\end{pmatrix}
\end{align*}
satisfies $\det A(0)=0$,
and $\Ker A(0)$ has a non-zero vector $c=(1,-1)^{t}$.
Thus $0$ is an eigenvalue of the operator $-\Delta_{\alpha,Y}$.

\subsection{Estimate for the operator norm}
We review an elementary estimate 
for the operator norm of a matrix operator.
\begin{proposition}
\label{prop_operator-norm}
For a countable set $X$,
let $M=(m_{xy})_{x,y\in X}$ be a 
$\mathbf{C}$-valued matrix satisfying
\begin{align*}
C_1:=\sup_{x\in X}\sum_{y\in Y}|m_{xy}|<\infty,\quad
C_2:=\sup_{y\in Y}\sum_{x\in X}|m_{xy}|<\infty.
\end{align*}
Then, the operator on $\ell^2(X)$
defined by
\begin{align*}
 M f(x)
=
\sum_{y\in Y} m_{xy} f(y)
\quad (f\in \ell^2(X))
\end{align*}
is a bounded linear operator on $\ell^2(X)$ 
and the operator norm $\|M\|$ satisfies
\begin{align*}
 \|M\|\leq \sqrt{C_1C_2}.
\end{align*}
\end{proposition}
Especially when $M$ is a Hermitian operator ($\overline{m_{xy}}=m_{yx}$),
we have $C_1=C_2$ and $\|M\|\leq C_1$.
\begin{proof}
For every $x\in X$, we have by the Schwarz inequality
\begin{align*}
 |Mf(x)|^2
\leq&\ 
\left(
\sum_{y\in Y}|m_{xy}| |f(y)|
\right)^2
\leq
\sum_{y\in Y}|m_{xy}|
\cdot
\sum_{y\in Y}|m_{xy}| |f(y)|^2
\\
\leq&\ 
C_1\sum_{y\in Y}|m_{xy}| |f(y)|^2.
\end{align*} 
Therefore
\begin{align*}
 \|Mf\|^2
=& 
\sum_{x\in X}|Mf(x)|^2
\leq
C_1\sum_{x\in X}\sum_{y\in Y}|m_{xy}||f(y)|^2\\
\leq&\ 
C_1\sum_{y\in Y}
\left(\sum_{x\in X}
|m_{xy}|\right)|f(y)|^2
\leq
C_1C_2\|f\|^2.
\end{align*}
\end{proof}

\textbf{Acknowledgments.}
The work of T.\ M.\ is partially supported by
JSPS KAKENHI Grant Number JP18K03329.
The work of F.\ N.\ is partially supported by
JSPS KAKENHI Grant Number 20K03659.

\end{document}